\documentclass{article}
\usepackage[a4paper]{geometry}
\usepackage{graphicx} % Required for inserting images
\usepackage{bbold}
\usepackage{lmodern}
\usepackage[T1]{fontenc}
\usepackage[utf8]{inputenc} 
\usepackage[utf8]{inputenc} 		% encodage des caracteres utilise (pour les caracteres accentues) -- non utilise ici.
\usepackage{amsmath,amssymb,amsthm}	% pour les maths
\usepackage{hyperref}			% si vous souhaitez que les references soient des hyperliens
\usepackage{color}			% pour ajouter des couleurs dans vos textes 
\usepackage{amssymb,amsmath,amsfonts,amsthm,amscd,amstext,mathtools}
\usepackage[english]{babel}
\usepackage{inputenc}
\usepackage[version=3]{mhchem}
\usepackage{subfig}
\usepackage{cite}
\usepackage{perpage} %useful for resetting page for footnotes
\usepackage{url}
\usepackage[export]{adjustbox}
\usepackage{authblk}

\usepackage{soul}
\usepackage{xcolor}
\usepackage{float}
\usepackage{enumitem}
\theoremstyle{definition}
\newtheorem{theo}{Theorem}[section]
\newtheorem{lemma}[theo]{Lemma}

\newtheorem{countex}[theo]{Counter-example}
\newtheorem{exx}[theo]{Example}
\newtheorem{defi}[theo]{Definition}
\newtheorem{remark}[theo]{Remark}

\newtheorem{assumption}[theo]{Assumption}
\author[1]{Escobar-Bach, M.}
\author[2]{Popier, A.}
\author[1]{Sahin, M.}
\affil[1]{LAREMA, Université d'Angers}
\affil[2]{Laboratoire Manceau de Mathématiques, Le Mans Université}

\title{A dependent and censored first hitting-time model with compound Poisson processes}

\begin{document}
\maketitle
\begin{abstract}
  We consider a bivariate first hitting-time model in which durations are the crossing times of dependent compound Poisson processes with fixed thresholds. The identifiability of the model is discussed, and likelihood estimators of the model parameters are proposed. We obtain the asymptotic properties of the estimators and underline their finite sample performance with a simulation study on synthetic data. The practical applicability of our approach is demonstrated by an application using data from patients suffering from mushroom poisoning.          
\end{abstract}

\noindent
\textit{Some key words :} Survival analysis, dependent censoring, compound Poisson process, first hitting-time model.

\section{Introduction}
The analysis of time-to-event data has always been of great interest in various applications, embedded with challenging statistical analysis when subject to censoring. First hitting-time (FHT) or threshold crossing models are attractive survival modeling approaches that allow censoring mechanisms to be described by stochastic processes. The construction of FHT models typically relies on a latent or unobservable continuous-time random process $L=(L_t,\,t\geq 0)$ and an event $B$. The event time of interest is then given by 
\begin{eqnarray}
\label{def::hittingtime}
    \tau_B=\inf\{t\geq0,\,L_t\in B\}
\end{eqnarray}
referring to the first time a stop rule is met. In particular, scenarios with crossing rules consider real thresholds $x\in\mathbb{R}$ where events $B=B_x$ have the forms given by $(-\infty,x]$ or $[x,+\infty)$ depending on whether they are up- or down-crossings. For multivariate analysis, several random times can be considered, each associated with different reaching rules. However, in the context of censored or competing events, only the smallest time value is returned, leaving the others unknown. Thus, model identification is usually ensured at the cost of independence between random times, which is referred to as independent censoring. In fact, the initial studies in \cite{Tsi75} emphasize that experiments with dependent censoring certainly prohibit statistical identifiability of the model and thus lead to drastic bias estimation. Recent work has therefore been concerned with this problem, where premises can be found in \cite{Zhe95} with the introduction of a copula-graphic estimator based on a known copula function. The latter has been further explored in \cite{Riv01} by considering Archimedean copulas and then extended to several regression models with known copula/archimedean generator in \cite{Bra05,Chen10,de17,Suj18,Emu18} among others. By definition, a copula is a joint multivariate distribution function with uniform margins that gives a margin-free characterization of the dependence structure, and whose existence has been universally proved in \cite{Skl59}. Some other recent developments have managed to relax the assumption of the known copula by considering parametric models in \cite{Cza22} and semiparametric models in \cite{Der23}. The use of copulas in dependent censoring is thus still in its infancy. In practice, copulas are convenient tools since they covers all possible dependence structures. However, it can be challenging for investigators without a strong mathematical background to grasp the formalism of copulas and link them effectively to the data source. Furthermore, most methods with parametric copula families rely on strong regularity conditions, such as absolute continuity or differentiability, and neglect those with singular parts, although they represent a large fraction of general copulas \cite{Dur15}. As such, FHT models allow to consider alternative characterizations that internalize the description of the dependence structures through stochastic constructions instead of model fitting. In survival analysis, this point of view presents substantial interest due to its realism and applicability \cite{Lee06}. For examples, the analyses in \cite{Dav15,Dav19,Esc24} use common shock models for possibly equally dependent event times, and a model extension referred to as mixed hitting-time model, considers random levels for the analysis of optimal stopping decisions by heterogeneous agents \cite{Abb12}.\\  
In real world applications, FHT models have a long history in medicine, economics, and engineering sciences. For example, the analysis in \cite{Lon89} has proposed a survival study on the follow-up to death for patients with human immunodeficiency virus (HIV) where the various stages of infection are based on a homogeneous time Markov process. Another study in \cite{Mar13} compares the length times of Italian students at the university with those of European students. The students leaving the study prematurely are considered as censored and the length time of each of them is simulated throughout a hidden Markov chain, sometimes referred to as phase-type transition models. Other recent works have successfully used FHT models in biology, such as in \cite{Hob18} with DNA modeling or in life insurance in \cite{Asm19}. From a heuristic approach, a study in \cite{DeBin23} has recently proposed a boosting gradient type algorithm to fit phase-type distributions with high-dimensional covariates. Extensions of this model have also been considered with for instance crossing rules at random levels. Such construction is typically referred to as mixed hitting-time models and have examples in economics, with real applications such as patent races \cite{Rei82,Choi91} and technology adoption \cite{Rich82,Far98}.\\
In this paper we propose a bivariate FHT model with compound Poisson processes. The dynamic of the model is based on the competition between two monotone processes to reach fixed levels. In line with the idea of longitudinal analysis, we assume that events are checked for completion at the same random time intervals. This means that the model can return the same event times when the processes cross their respective levels simultaneously. Such interdependency between the processes is crucial and plays a central role in identifying the model structures. Likewise competing models with common shocks, it allows to consider complex follow-up experiments with two dependent competing events, even when deterministic or continuous recording is not feasible. It also extends to applications with a different scope than copula-based models with standard statistical assumptions, which are not usually suited to copulas with singular components. Our estimation uses a maximum likelihood approach applied to the parameters of the jump size distributions and the inter-arrival time intensity. The identification and consistency analyses are addressed with properties from the frameworks of renewal processes and convolution of densities. In particular, our study discusses the necessary settings and conditions on the jump size distributions to ensure the model identifiability. Answering such questions has already been the subject of several contributions with observed Lévy processes \cite{Com15,Coc18,Duv21}, and is a even more challenging in our study due to the censoring mechanism and the bi-variate setting. In particular, the fixed-threshold setting turns out to be more complex to handle than the random-threshold setting, since the level heterogeneity helps to select the correct jump size distributions. Under this context, we prove the asymptotic normality of the proposed estimators under mild conditions. Our approach finite sample performances are discussed with a study on synthetic data and its practical application is underlined with a study on mushroom poisoning. 

The remainder of the paper is given as follows. In Section \ref{sect:model_setting} we present the stochastic construction associated with the event times and discuss our model representation in the context of survival analysis. %We introduce two non-decreasing compound Poisson processes $L_1$ and $L_2$, which realize there respective jumping time simultaneously. The considered outcome is the couple $(Y,\Delta)$ composed of the minimum of the $T$ and $C$ respectively the FHT a fixed thresholds for $L_1$ and $L_2$ and the indicator that tells which one has occurred. 
Section \ref{sect:function_analysis} proposes an analysis of the model functions, highlighting remarkable properties of the hitting times, in perspective with the internal dependence structure of the bivariate stochastic process. In Section \ref{sect::identifia}, we propose a parameterized configuration and discuss on conditions that ensure its identifiability. A maximum likelihood approach is proposed along with asymptotic results in Section \ref{sect:prop_estim}. In Section \ref{sect:simulation}, we illustrate the finite-sample behavior of our approach on generic data and highlight its practical applicability with an original study on poisoning data. All proofs and technical lemmas are postponed to Section \ref{sect:appendix}.

\section{Model and settings} \label{sect:model_setting}

As indicated in the introduction section, our interest is in understanding the distributions of random crossing times as in (\ref{def::hittingtime}) for a bivariate Lévy process $L=(L_1,L_2)$ and fixed thresholds. In our study, we consider two pure jumps Lévy processes whose respective jump times occur simultaneously and the jump sizes are independent. Dependence within the latent process $L$ hence results from the common indexation of the marginal processes with the same Poisson process $N$, which characterizes the number of jumps at time $t$.\\ 
Formally, we denote $(E_n)_{n \in \mathbb{N}^*}$ as an \textit{i.i.d.} sequence of exponential laws with intensity parameter $\lambda>0$ and define $$N_t:= \sup \{ n \in \mathbb{N}^*, E_1 + \ldots + E_n \leq t \},\quad t>0$$ where the supremum of the empty set is set to 0. The jump sizes of the processes $L_1$ and $L_2$ are respectively regrouped in the i.i.d. sequences $(X_n)_{n \in \mathbb{N}^*}$ and $(Z_n)_{n \in \mathbb{N}^*}$. All the aforementioned random series are assumed independent and define the compound Poisson processes:
\begin{eqnarray*}
    L_{t,1} = \sum_{n=1}^{N_t} X_n = \sum_{n=1}^{+ \infty} \ X_n \mathbb{1}_{ \{ E_1 + \ldots + E_n \leq t \} } \quad \quad \forall t \geq 0, \\
    L_{t,2} = \sum_{n=1}^{N_t} Z_n = \sum_{n=1}^{+ \infty} \ Z_n \mathbb{1}_{ \{ E_1 + \ldots + E_n \leq t \} } \quad \quad \forall t \geq 0.
\end{eqnarray*} 
In this study, we assume that the marginal processes are almost surely monotonic and consider jump sizes with constant signs. Without loss of generality and to ease the forthcoming analysis, we consider the case where both are almost surely non-decreasing. The model is settled with $\mathbb{P}[Z_1=0],\mathbb{P}[X_1=0]\in[0,1)$. The possible positivity of both terms requires attention and will be important during the identifiability stage.
For the sake of rigor in the forthcoming calculations, let $X_0$ and $Z_0$ be such that $X_0 = Z_0 = 0$ a.s.\\

For any fixed thresholds $x,z>0$, we define the up-crossing times of the processes $L_1$ and $L_2$ by
\begin{eqnarray*}
    T\; = \; \inf \{ t \geq 0 ;  \quad L_{t,1} \geq x \}\quad\text{and}\quad C \; = \; \inf \{ t \geq 0 ;  \quad L_{t,2} \geq z \}.
\end{eqnarray*}
In our context, the process $L$ is unobservable so that the distributions of the event times $T$ and $C$ are the only statistical information available. In addition, due to the censoring mechanism, only the smallest random time is returned, in a sense that both processes $L_1$ and $L_2$ compete to be the first to cross their respective thresholds. In survival analysis, this setting is characterized by pairwise outcomes that return the uncensored random time and an indicator of the observed time. Related to this framework, we consider two types of survival models $(Y,\Delta)$, depending on the censoring information available.
\begin{align*}
    \textbf{Model I:} \quad \quad \quad & (Y, \Delta ) \; = \; ( \min \{ T ; C \} , \mathbb{1}_{ \{T  \leq C \} } ) \\
    \textbf{Model II:} \quad \quad \quad & (Y, \Delta ) \; = \; ( \min \{ T ; C \} , 0. \mathbb{1}_{\{ T > C \}}   + 1 . \mathbb{1}_{\{ T < C \}} + 2 . \mathbb{1}_{\{ T = C \}} )
\end{align*}
Scenarios drawn from Model I are usually studied with independent censoring while Model II can be seen as a competing risk configuration. The analysis in \cite{Esc24}, where dependent censoring occurs and $\mathbb{P}[ T = C ] > 0$, has shown that non-parametric identifiability is not always guaranteed considering only Model I as opposed to Model II. Beyond the applicability of Model II in some experimental contexts, these settings will be discussed to address the model identifiability in Section \ref{sect::identifia}. From a statistical perspective, we will work under a parametric framework for the inter-arrival times of the process $L$ and  the jump size distributions. Depending on the class of the model and the definition of the outcomes, our objective is to propose an estimation procedure of the model parameters that fully describes the law of $L$. This will then make explicit the distribution of event times $T$ and $C$, which are usually of primary interest in survival analysis. 
%A wide class of models ? are identifiable considering the less informative couple of variables  $(Y, \Delta )$, but one case remains where we can not obtain yet identifibiability without considering the richer couple of random variables $(Y, \Delta)$. All this remarks will be detailed and explained in Section 4. \\

\paragraph{Statistical context}
In our problem, $X_1$ and $Z_1$ are supposed to belong in some parametrized families of probability laws.
In our context, we observe realizations of the censoring couple $(Y, \Delta )$ from Model I or II, depending on which class of model we lie in. The usual objective of a survival model with censoring is to determinate the respective laws of the variable of interest T and of the censoring variable C. Our objective is to build an estimator of the vector of parameters that fully describe the law of $L$, naturally inducing laws of $T$ and $C$ are deduced.
The article \cite{Whi82} gives explicit assumptions to be verified on the parametrized density function $f_{(Y,\Delta)}$ to obtain respectively the consistency and the asymptotic normality of the Maximum likelihood estimator.

\section{Theoretical analysis of model functions} \label{sect:function_analysis}

\medskip
In the following section, we propose a prior analysis of the representations and the properties of the model functions. The structure of the model induces a dependency between the two hitting times due to the fact that the completion of the crossing rules can occur simultaneously with nonzero probability. This particularly implies that the couple $(T,C)$ dependence structure admits a singular part supported on the set $\{(u,v)\in\mathbb{R}^2_+,\,u=v\}$. We derive the couple density function in the following lemma and illustrate this property in Remark \ref{rem:pro_of_simult_occur}.

\begin{lemma}[Crossing time density function] \label{lem:dens_fct_hitting_times}
The absolutely continuous part of the law of the couple $(T,C)$ is supported by the set $\{(u,v) \in (\mathbb{R}_+)^2, \ u \neq v\}$ and has a density function given by: 
\begin{align} \nonumber 
     f_{ac}(u,v) \; = \; \sum_{n=1}^{+ \infty} & \sum_{j=1}^{n} \big[ c_{j-1,X} - c_{j,X} \big] \big[ c_{n,Z} - c_{n+1,Z} \big] 
    \lambda^2 e^{ - \lambda v} \frac{ (\lambda (v - u))^{n-j} }{ (n-j)!  } \frac{( \lambda u )^{j-1}}{(j-1)!} \mathbb{1}_{ \{ u < v \} }  \\ \label{eq:density_abs_cont_part}
    + \sum_{n=1}^{+ \infty} & \sum_{j=1}^{n} \big[ c_{n,X} - c_{n+1,X} \big] \big[ c_{j-1,Z} - c_{j,Z} \big]
    \lambda^2 e^{ - \lambda u} \frac{ (\lambda (u - v))^{n-j} }{ (n-j)!  } \frac{( \lambda v )^{j-1}}{(j-1)!} \mathbb{1}_{ \{ u > v \} }.
\end{align}
The singular part is supported by the diagonal $\{(u,u),  u \geq 0\}$ with a density function given by:
\begin{align} \label{eq:density_sing_part}
f_s(u)=  \partial_u \mathbb{P}[T\leq u, T=C] \; = \; \sum_{n=0}^{+ \infty} & \big[ c_{n,X} - c_{n+1,X} \big]  \big[ c_{n,Z} - c_{n+1,Z} \big] \lambda e^{ - \lambda u } \frac{(\lambda u)^n}{n!}  ,
\end{align}
%$$f_{ac}(u,v)\bold{1}_{\{u\neq v\}}+f_s(u)\bold{1}_{\{u=v\}},\quad u,v>0,$$ where $f_{ac}$ defines the absolutely continuous part with respect to the Lebesgue measure and $f_s$ is the singular part. The two parts of the density are then given by
where in both formulas for any $n \in \mathbb{N}$
\begin{eqnarray*}
    c_{n,X} = \mathbb{P} \left( \sum_{i=0}^n X_i < x \right) \quad \quad \quad \text{and} \quad \quad \quad c_{n,Z} = \mathbb{P} \left( \sum_{i=0}^n Z_i < z \right).
\end{eqnarray*}
\end{lemma}
\noindent
This singular property due to the coupling of the crossing times plays a central role in our model dynamics. In a sense, one can think of an external agent returning the times at which the crossing rules has been observed. This leaves open the possibility that the events share the same time window, which results in a common returned time from the agent. The following lemma gives an explicit formula of the probability that such event occurs in terms of the jump sizes distributions. 
%It also gives a qualitative information on the behavior of the couple of random variables, depending on its value.

\begin{lemma}
\label{lem:null_prob_induces_non_censoring}
    The probability that $T$ and $C$ occur at the same time is given by:
    \begin{equation} \label{eq:prob_T_equal_C}
        \mathbb{P} \left[ T = C \right] = \underset{n \in \mathbb{{N}}}{\sum} [c_{n,X} - c_{n+1,X} ] [c_{n,Z} - c_{n+1,Z} ] .
    \end{equation}
    Moreover, the previous probability is null if and only if one of the crossing times occurs almost surely before the other one.
\end{lemma}

\begin{remark}
\label{rem:pro_of_simult_occur}
It is worth-noting that the probability of equaled event times does not depend from the intensity parameter $\lambda$. This can be explained intuitively, since the simultaneous occurrence depends only on the number of jumps between the processes $L_1$ and $L_2$ and not on the jump times.
\end{remark}
\noindent
As an immediate consequence of Lemma \ref{lem:null_prob_induces_non_censoring}, we observe that the model is non trivial if the probability that $T$ and $C$ occurs at the same time is strictly positive. Indeed, configurations under the null probability automatically induce either almost certainly censored or uncensored data, which leads to the non-identifiability of the law of $T$ or $C$. In the sequel, we will thus only consider processes such that $\mathbb{P} \left[ T=C \right] >0$.
%To verify if the model is whether identifiable or not (see Definition \ref{def:identifi}), we need to exploit the density function of the censoring couple. 
In the next lemma, we derive the density of the censoring couple under models I and II.  

\begin{lemma}[Survival outcome density function]\label{lem:dens_fct_censor}
The density function of the couple $(Y,\Delta)$ is given by:
\begin{align*}
    \textbf{Model I:} \quad \quad f_{(Y, \Delta)} (t, \delta) = &
     \left\{ \begin{matrix} \sum_{n=0}^{+ \infty}  [c_{n,X} - c_{n+1,X} ] c_{n,Z} \lambda e^{- \lambda t} \frac{(\lambda t)^n}{n!} \mathbb{1}_{ \{ t \geq 0 \} } & \quad \quad \; \; \; \quad \quad \quad \text{if} \; \delta = 1 \\[.5cm]
     \sum_{n=0}^{+ \infty} [c_{n,Z} - c_{n+1,Z} ] c_{n+1,X} \lambda e^{- \lambda t} \frac{(\lambda t)^n}{n!} \mathbb{1}_{ \{ t \geq 0 \} } & \quad \quad \quad \; \; \; \quad \quad \text{if} \; \delta = 0
    \end{matrix} \right. \\[.5cm]
    \textbf{Model II:} \quad \quad f_{(Y, \Delta)} (t, \delta) = & \left\{ \begin{matrix}  \sum_{n=0}^{+ \infty} [c_{n,Z} - c_{n+1,Z} ] c_{n+1,X} \lambda e^{- \lambda t} \frac{(\lambda t)^n}{n!} \mathbb{1}_{ \{ t \geq 0 \} }  & \quad \quad \text{if} \; \delta = 0 \\[.5cm]
    \sum_{n=0}^{+ \infty}  [c_{n,X} - c_{n+1,X} ] c_{n+1,Z} \lambda e^{- \lambda t} \frac{(\lambda t)^n}{n!} \mathbb{1}_{ \{ t \geq 0 \} } & \quad \quad \text{if} \; \delta = 1 \\[.5cm]
    \sum_{n=0}^{+ \infty}  [c_{n,X} - c_{n+1,X} ] [c_{n,Z} - c_{n+1,Z} ] \lambda e^{- \lambda t} \frac{(\lambda t)^n}{n!} \mathbb{1}_{ \{ t \geq 0 \} }  & \quad \quad \text{if} \; \delta = 2  \end{matrix} \right.
\end{align*}
\end{lemma}

In the next section, analytical properties of the hazard function of $Y$ will help to identify the intensity parameter $\lambda$. The next lemma proposes a representation of the hazard function that observes proportionality with the latter parameter.

\begin{lemma}
\label{lem:hazard}
    The hazard function of the random variable $Y$ is given by:
    \begin{equation} \label{eq:hazard-fct}
    h_{Y}(t) = \frac{f_{Y}(t)}{S_{Y}(t)} = \lambda \left(1- \dfrac{ \sum_{n=0}^{+ \infty} c_{n+1,X} c_{n+1,Z} \frac{(\lambda t)^n}{n!} }{ \sum_{n=0}^{+ \infty} c_{n,X} c_{n,Z} \frac{(\lambda t)^n}{n!} }\right)
    \end{equation}
    where $S_{Y}$ denotes the survival function of $Y$.
\end{lemma}

\section{Identifiability}
\label{sect::identifia}

\paragraph{Notations}
In line with our statistical purposes, we consider a parametric framework and assume that the jump size laws belong to pre-specified parametric distribution classes. To reflect this setting in the rest of the paper, we propose to index the previous notations with new parameters. We thus assume that the distributions of $\mathbb{P}_{X_1}$ and $\mathbb{P}_{Z_1}$ respectively belong to the classes $\{ \mathbb{P}_{\alpha}, \alpha \in \Theta_1 \}$ and $\{ \mathbb{P}_{\beta}, \beta \in \Theta_2 \}$ where $\Theta_i$, $i=1,2$ stand as real parameter spaces. We also denote the full parameter space as $\Theta=\Lambda \times \Theta_1 \times \Theta_2$ and index any model functions or probabilities with $\theta=(\lambda,\alpha,\beta)\in\Theta$ if there is a dependence on the jump size distributions or the counting process intensity. In particular, the coefficients of the cumulative jumps are now denoted by\\

% We assumed that $X_1$ is such that $ \mathbb{P}_{X_1} \in \{ \mathbb{P}_{\alpha}, \alpha \in \alpha \}$, where $\alpha$ is a certain subset of $\mathbb{R}^{d_1}$. The same assumption is made with $Z_1$: $Z_1$ is such that $ \mathbb{P}_{Z_1} \in \{ \mathbb{P}_{\beta}, \beta \in \beta \}$, where $\beta$ is a subset of $\mathbb{R}^{d_2}$. The model is fully described by the vector of parameter $(\lambda,\alpha,\beta)$. From now on, we write
% \begin{eqnarray*}
%     \theta = (\lambda,\alpha,\beta) \quad \text{for any} \quad (\lambda,\alpha,\beta) \in \Lambda \times \Theta_1 \times \Theta_2 \quad \quad \quad \quad \text{and}  \quad \quad \quad \quad \Theta = \Lambda \times \Theta_1 \times \Theta_2
% \end{eqnarray*}
% where $\Theta$ is a subset of $\mathbb{R}^d$, with $d = 1 + d_1 + d_2$. \\
% Given $\theta = (\lambda,\alpha,\beta) \in \Theta$ such that $\mathbb{P}_{X_1} = \mathbb{P}_{\alpha}$ and $\mathbb{P}_{Z_1} = \mathbb{P}_{\beta}$, we denote by $f_{(Y,\Delta),\theta}$ ( resp. $f_{Y,\theta}, h_{Y,\theta}$) the induced density function of $(Y,\Delta)$ (resp. density function of $Y$, hazard function of $Y$).  We therefore also denote for any $n \in \mathbb{N}$ 

\begin{eqnarray*}
    c_{n,X}(\alpha) = \mathbb{P}_{\alpha} \left(\sum_{i=0}^n X_i < x \right) \quad \quad \quad \text{and} \quad \quad \quad c_{n,Z}(\beta) = \mathbb{P}_{\beta} \left(\sum_{i=0}^n Z_i < z \right). 
\end{eqnarray*}

\paragraph{Identification in survival analysis}
It is important to mention that the issue of identifiability is strongly intertwined with the dependency between event times. Unlike independent censoring, which allows non-parametric identification, the complex dependence structure resulting from our model requires much more sophisticated analysis. In particular, its singular part contradicts requirements for some well defined approaches such as in \cite{Bra05} with the copula-graphic estimator and in \cite{Cza22} with a parametric setting. In the survival analysis framework, model identifiability has to take into account the censoring issue with an alternate definition than in standard scenarios. Formally, for any two parameter values $\theta,\theta' \in \Theta$, we say that the survival model is identifiable if 

% As explained in the introduction, the case of dependent censoring complicates the question of identifiability of the model. Although the remark \ref{rem:pro_of_simult_occur} shows the existence of a dependence structure between the competing events $T$ and $C$, it also proves the existence of a singular part. Some well defined models such as in \cite{Bra05} and \cite{Cza22} allow to use copula-graphic estimator. Nevertheless, \cite{Bra05} requires to know the copula function, when in \cite{Cza22}, a parametrized family of archimedian copulas in which the true one belongs in is necessary. It is difficult in our case to determine an explicit form of copula for $(T,C)$. Consequently, we propose to analyze functions that characterize the couple $(Y,\Delta)$ to obtain identifiability. \\
% Considering either Model I or II, identifiability translates as follows:
% Let $\alpha,\beta \in \Theta$ be two vector of parameters. Then:

\begin{equation} \label{eq:indent_prop}
    f_{(Y, \Delta),\theta} = f_{(Y, \Delta), \theta'} \quad\Longrightarrow \quad \theta = \theta'.
\end{equation}
\noindent
However, criteria only based on the density function of the observed data does not allow to immediately obtain the above property. This difficulty is particularly illustrated in the following counter-example based on the exponential distribution.
\begin{countex}
\label{count:non_identifiable_1}
Consider two intensity parameters $\lambda,\lambda'>0$ and two sequences $(b_n)_{n\in\mathbb{N}}$, $(b'_n)_{n\in\mathbb{N}}$ with $\lambda = 1$, $\lambda' = 3/4$, $b_n= (1/2)^n$ and $b_n' = (1/4)^n$. Then for any $t\geq 0$, one we have
\begin{eqnarray*}
    e^{- \lambda t} \sum_{n=0}^{+ \infty} b_n \frac{(\lambda t)^n}{n!} = e^{- \lambda' t} \sum_{n=0}^{+ \infty} b'_n \frac{(\lambda' t)^n}{n!}.
\end{eqnarray*}

% Consider two vector of parameters $\alpha = (\lambda,\alpha,\beta),\beta=(\lambda',\alpha',\beta') \in \Theta$. Assume that $f_{(Y, \Delta),\alpha} = f_{(Y, \Delta), \beta}$. This identity gives equalities of the type:
% \begin{eqnarray*}
%     \forall t \geq 0, \quad \quad e^{- \lambda t} \sum_{n=0}^{+ \infty} b_n(\alpha,\beta) \frac{(\lambda t)^n}{n!} = e^{- \lambda' t} \sum_{n=0}^{+ \infty} b_n(\alpha',\beta') \frac{(\lambda' t)^n}{n!}
% \end{eqnarray*}
% where $(b_n(\alpha,\beta))_{n \in \mathbb{N}}$ and $(b_n(\alpha',\beta'))_{n \in \mathbb{N}}$ are sequences of non-negative numbers less or equal to 1.
% This equation doesn't guarantee that $(\lambda,b_n(\alpha,\beta))=(\lambda',b_n(\alpha',\beta'))$ and thus $(\lambda,\alpha,\beta)=(\lambda',\alpha',\beta')$. In fact, consider:
% \begin{eqnarray*}
%     \lambda = 1 , \quad \lambda' = 3/4, \quad b_n(\alpha,\beta) = (1/2)^n \quad \text{and} \quad b_n'(\alpha',\beta') = (1/4)^n
% \end{eqnarray*}
% leading to a non-identifiable censoring model. 
\end{countex}

\noindent
The previous example thus shows that fine-tuning parameters can result in equal density functions of $(Y,\Delta)$, even though they have different values. Nonetheless, we can prove that identification of the intensity parameter and cumulative jump coefficients are possible under different configurations of the parametric classes.
\begin{defi}
Our model complies with the assumption $I_n$, if for any parameter values $\theta \in \Theta$ and $\theta'\in\Theta$
\begin{eqnarray*}
    f_{(Y, \Delta),\theta} = f_{(Y, \Delta), \theta'} \quad \Longrightarrow\quad c_{k,X}(\alpha)=c_{k,X}(\alpha')\quad\text{and}\quad c_{k,Z}(\beta)=c_{k,Z}(\beta'),\quad\forall k\leq n
\end{eqnarray*}
where by convention, satisfying the hypothesis $I_n$ for any $n\in\mathbb{N}$ is written $I_\infty$. 
\end{defi}
It is important to emphasize that fulfillment of the $I_n$ assumption does not necessarily ensure survival identifiability. From a stochastic process perspective, this also represents a complicated problem, which can be seen as the difficulty to characterize renewal processes only from the resulting crossing time distributions. Theoretical arguments such as the absence of stationary increments for renewal processes jusitify the difficulty of this question. Contrarily to the particular case of Poisson processes, most of the renewal processes are not Lévy processes (see Lemma 1.1 from \cite{Mit14}). However, the analysis of our model with several configurations has shown that almost all standard parametric classes under $I_n$ assumption have identifiable survival models. This suggests that reasonable settings with our model description mostly return practical survival models with likelihood methods. This model property is completed in the end of this section where we propose several assumptions to guarantee $I_n$ or $I_\infty$ under mild conditions.\\

% In order to discuss on the identifiability of our model, we define for any $n \in \mathbb{N}$, $\alpha, \alpha' \in \alpha$ and $\beta, \beta' \in \beta$ the property:
% \begin{align*}
%     I_n \quad : \quad & \mathbb{P}_{\alpha}[X_1 + \ldots + X_k < x] = \mathbb{P}_{\alpha'}[X_1' + \ldots + X_k' < x]  \quad \quad \forall k \leq n, \\
%     & \mathbb{P}_{\beta}[Z_1 + \ldots + Z_k < z] = \mathbb{P}_{\beta'}[Z_1' + \ldots + Z_k' < z] \quad \quad \forall k \leq n,
% \end{align*}
% \begin{align*}
%     I_{\infty} \quad : \quad & \mathbb{P}_{\alpha}[X_1 + \ldots + X_k < x] = \mathbb{P}_{\alpha'}[X_1' + \ldots + X_k' < x] \quad \quad \forall k \in \mathbb{N}, \\
%     &\mathbb{P}_{\beta}[Z_1 + \ldots + Z_k < z] = \mathbb{P}_{\beta'}[Z_1' + \ldots + Z_k' < z] \quad \quad \forall k \in \mathbb{N}.
% \end{align*}

\paragraph{Model identifiability}
Our identification procedure is twofold: first, a prior analysis of the hazard function of $Y$ shows that our survival model allows identification of the intensity parameter, and then we show that the assumption $I_n$ is verified for some $n \in \mathbb{N}$.\\

Derivation of the hazard function in Lemma \ref{lem:hazard} allows to isolate the intensity parameter in such way that the function can be written as $\lambda(1-g)$ where one can determine $g$'s limit at infinity.
In the following lemma, we show how to obtain it based on conditions on the cumulative jump size sequences.

\begin{lemma} \label{lem:tech_resulu}
Let $(a_n)_{n \in \mathbb{N}}$ be a non-increasing sequence of positive real numbers.
Let $g$ be the function defined on $\mathbb{R}^*$ by
\begin{eqnarray*}
    g(t) = \frac{\sum_{n=0}^{+ \infty} a_{n+1} t^n}{ \sum_{n=0}^{+ \infty} a_{n} t^n }, \quad t > 0.
\end{eqnarray*}
Suppose that there exists $a \in [0,1]$ such that $\lim_{n \xrightarrow{} + \infty} \frac{a_{n+1}}{a_n} = a$, then $\lim_{t \xrightarrow{} + \infty} g(t) = a $.
\end{lemma}
\noindent
Lemma \ref{lem:tech_resulu} suggests that knowledge of the limit $a$ of $(c_{n+1,X}c_{n+1,Z})/(c_{n,X} c_{n,Z})$ can help isolating the intensity parameter since the limit of the hazard function at the infinity gives $\lambda(1-a)$. As such, it is important to determine under which conditions one can determine the value of $a$.

In order to answer that question, we define the three following families of random variable sequences. 
\begin{defi} \label{def:families}
Let $M = (M_n)_{n \in \mathbb{N}^*}$ be a sequence of i.i.d. non-negative random variables. 
\begin{description}
\item[Class $F_1$.] $M_1$ is almost surely greater than a positive constant.
\item[Class $F_2$.] $M_1$ is absolutely continuous with respect to the Lebesgue measure and the density denoted by $\phi$ satisfies
\begin{itemize}
    \item For all $\varepsilon > 0$, $\mathbb{P}(M_1 \in [0, \varepsilon] ) = \int_0^\varepsilon \phi  > 0$. %$\mathbb{P}[M_1 \in [0, \varepsilon] ] > 0$, $\forall \varepsilon >0$
    \item There exists $N > 0$ such that $\phi^{*N} = \phi * \ldots* \phi$ is non-decreasing on the subset $[0,x]$. 
\end{itemize} 
Note that $\phi^{*N}$ is the density function of the sum $S_N = \sum_{k=1}^{N} M_k$. 
\item[Class $F_3$.] $M_1$ is a discrete random variable, $0 \in M_1(\Omega)$ and $ \inf( M_1(\Omega)\backslash \{ 0 \}) > 0$. 
\end{description}
\end{defi}
\begin{remark}
    The existence of $N > 0$ such that $\phi^{*N}=\phi_{S_N}$ is non-decreasing on the subset $[0,x]$ is a property that can be verified in practice. Indeed, the function $\phi^{*N}$ is determinable for lot of density functions. Easiest cases are the infinitely divisible and stable laws. For example the class of gamma distributions, uniform laws on the compact subset $[0,a]$, with $a >0$, Fréchet distributions, Lévy distributions, etc. belong to $F_2$.
\end{remark}

In the sequel, we consider that $(X_n)_{n \in \mathbb{N}}$ and $(Z_n)_{n \in \mathbb{N}}$ respectively belong to the classes $F_1$, $F_2$ or $F_3$. If $(X_n)_{n \in \mathbb{N}}$ (resp. $(Z_n)_{n \in \mathbb{N}}$) belongs to the family $A_1$, there exists $N \in \mathbb{N}$ such that for any $n \geq N$, $c_{n,X} = 0$ (resp. $c_{n,Z}=0$). The following theorem determines the limit at infinity of the sequences $\left( c_{n+1,X} /c_{n,X} \right)_{n \in \mathbb{N}}$ and $\left(  c_{n+1,Z} / c_{n,Z} \right)_{n \in \mathbb{N}}$, depending on which family of probability laws they belong to $F_1$ or $F_2$. 

\begin{theo} \label{theo:behaviour_c_(n+1)/c_n}
If $(X_n)_{n \in \mathbb{N}^*} \in F_2$. 
Then we have:
\begin{eqnarray*}
     \lim_{n \xrightarrow{} + \infty} \frac{\mathbb{P} [ X_1 + \ldots + X_{n+1} < x ]}{\mathbb{P} [ X_1 + \ldots + X_n < x]} = 0 .
\end{eqnarray*}
If $(X_n)_{n \in \mathbb{N}^*} \in F_3$.
Then, for any $x>0$, there exists $C_x > 0$:
\begin{eqnarray*}
    \left| \frac{\mathbb{P} [ X_1 + \ldots + X_{n+1} < x] }{ \mathbb{P} [ X_{1} + \ldots + X_n < x] } - \mathbb{P} [X_1 = 0] \right| \leq \frac{C_x}{n} .
\end{eqnarray*}
\end{theo}
The sequence $\left(  c_{n+1,X} / c_{n,X} \right)_{n \in \mathbb{N}}$ admits different limits whether $(X_n)_{n \in \mathbb{N}}$ belongs to $F_2$ or $F_3$. From that observation, we understand the necessity to distinguish the identifiability problem in three different cases. 
\begin{assumption} \label{ass:identifiability}
We now define the following conditions that separate these cases. In all three statements, the property is supposed true regardless of the values of the parameters.
\begin{itemize}
    \item \textbf{(H1).} $(X_n)_{n \in \mathbb{N}}$ or $(Z_n)_{n \in \mathbb{N}}$ belong to $F_1$.
    \item \textbf{(H2.i).} $(X_n)_{n \in \mathbb{N}}$ or $(Z_n)_{n \in \mathbb{N}}$ belong to $F_2$. 
    \item \textbf{(H2.ii).}  $(X_n)_{n \in \mathbb{N}}$ and $(Z_n)_{n \in \mathbb{N}}$ belong to $F_3$.
\end{itemize}
\end{assumption}

Assumption {\bf (H1)} and {\bf (H2.i)} both ensure that $\lim_{t\to+\infty}h_{Y}(t) = \lambda$, while under assumption \textbf{(H2.ii)}, combining Lemma \ref{lem:tech_resulu} and Theorem \ref{theo:behaviour_c_(n+1)/c_n}, only gives  $\lambda(1 - \mathbb{P}_{\alpha} [ X_1 = 0 ] \mathbb{P}_{\beta} [ Z_1 = 0 ]) = \lambda'(1 - \mathbb{P}_{\alpha'} [ X_1 = 0 ] \mathbb{P}_{\beta'} [ Z_1 = 0 ]) $ which is insufficient to obtain the intensity parameter identifiability. Model I provides less information than Model II. Thus, configurations such as under \textbf{(H2.ii)} require to consider a more informative model to obtain the intensity parameter identification, explaining the necessity to consider Model II. \newline

\begin{theo}[Identifiability] \label{theo:identifiability} 
$\ $
\begin{itemize}
    \item Under \textbf{(H1)}, considering Model I, the identifiability of the intensity parameter $\lambda$ is ensured and $I_{N_{\max}}$ is verified, with $N_{\max} = \min \{ n \in \mathbb{N}, \mathbb{P}_{\alpha}[X_1 + \ldots + X_n < x] = 0 \text{ or } \mathbb{P}_{\beta}[Z_1 + \ldots + Z_n < z] = 0 \}$.
    \item Under \textbf{(H2.i)}, considering Model I, the identifiability of the intensity parameter $\lambda$ is ensured and $I_{\infty}$ is verified.
    \item Under \textbf{(H2.ii)}, considering Model II, the identifiability of the intensity parameter $\lambda$ is ensured and $I_{\infty}$ is verified.
\end{itemize}
\end{theo}

Theorem \ref{theo:identifiability} is a first step to obtain the identifiability of the model in three general cases. We obtain the identifiability of the intensity parameter $\lambda$, and depending on the class of model, $I_{\infty}$ or $I_n$ is verified for some $n \in \mathbb{N}$. It is worth-noting that \textbf{(H2)} ensures $I_{\infty}$ to be verified, guarantying the existence of an infinite dimensional system of equations that respectively link $\alpha$, $\alpha'$ and $\beta$, $\beta'$, which are finite dimensional vectors of parameters. Moreover, the property $I_n$ induces the identifiability of the jumping variable parameters  whenever $n\in \mathbb{N}^*$ is large enough. It therefore follows that a wide variety of type I models induce the identifiability from definition \eqref{eq:indent_prop}. The following example gives the number minimal of equations needed to obtain the identifiability of the jumping parameter.

\begin{exx} \label{exx:identifiable}
    Given $I_n$, we define $N_{\min} = \min \{ n \in \mathbb{N}^*, I_n \Rightarrow \alpha = \alpha' \}$.
If $X_1 \sim \mathcal{E}(\alpha)$ and $X_1' \sim \mathcal{E}(\alpha')$, $N_{\min} = 1$. If $X_1 \sim \mathcal{B}(p)$ and $X_1' \sim \mathcal{B}(p')$, $N_{\min} = \min\{ n \geq x\}$. If $X_1 \sim \mathcal{P}(\alpha)$ and $X_1' \sim \mathcal{P}(\alpha')$, $N_{\min} = 1$
where $\mathcal{E}$ denotes the exponential law, $\mathcal{B}$ denotes the Bernoulli law and $\mathcal{P}$ denotes the Poisson law.
\end{exx}

It is possible to build counter-examples such as Example \ref{count:non_identifiable_1}, where $I_n$ does not imply the jumping parameter identification $(\alpha,\beta)$.
\begin{countex}
Assume that the model lies under assumption {\bf (H1)}. Suppose that $Z_1 = Z_1' = 1$ a.s. and that $X_1$ and $X_1'$ belongs to the family of binomial laws $\{ \mathcal{B}(n,p), (n,p) \in \mathbb{N}^* \times ]0;1[ \}$. Finally suppose that $x=9$ and $z=3$. \\
If $X_1 \leadsto \mathcal{B}(n,p)$ and $X_1' \leadsto \mathcal{B}(n',p')$ it naturally comes from simple calculus that necessarily $\lambda = \lambda'$ but if $n=3$ and $n'=4$ only one equation links $p$ and $p'$ :
\begin{eqnarray*}
    \mathbb{P} [ X_1 + X_2 + X_3 < x] = \mathbb{P} [ X_1' + X_2' + X_3' < x] \iff  1 - p^9 = 1 - \sum_{k=9}^{12} \binom{12}{k} (p')^k (1-p')^{12-k}.
\end{eqnarray*}
And for any $p \in (0,1)$ fixed, there exists a unique $p'$ such that the the equality holds. While the initial processes are not the same (the laws of $X_1$ and $X_1'$ are not the same), they induce the same censoring density function. One can observe that in this counter-example, only one equation links $(n,p)$ to $(n',p')$.
\end{countex}
\noindent
Although the latter example shows the possible non-identifiability of some model, one can see that it has been precisely defined to permit only one equation to link the two unknown parameters. It illustrates that non-identifiability can occur in exotic models where the number of equation that link the parameters is strictly less than the dimension of the vector of parameters.
However for a great variety of examples, identifiability holds. See Example \ref{exx:identifiable} or examples of Section \ref{ssect:simulation}.

\section{Consistency and asymptotic normality of the MLE estimator}\label{sect:prop_estim}

From now on, Assumption \ref{ass:identifiability} holds, that is \textbf{(H1)}, \textbf{(H2.i)} or \textbf{(H2.ii)} is verified and we suppose that the identifiability equation \eqref{eq:indent_prop} is satisfied. 

In this section, we propose an MLE approach in order to estimate the parameter of our model $\theta$. All the results exposed in this section are based on convergence results from \cite{Whi82}. In his article, the author gives a list of six assumptions (denoted A1 to A6) to verify to obtain consistency and asymptotic normality of the MLE estimator. Verification must be done to see if the conditions of \cite{Whi82} are verified in our setting under mild assumptions. %The structure of the density function of our model 

\paragraph{Notations}
Given a n-sample of i.i.d. realisations of the couple $(Y, \Delta )$ $\{ (Y_i,\Delta_i), \; i \in \{1,\ldots,n\}  \}$, induced by the vector of parameters $\theta^0 = (\lambda^0, \alpha^0, \beta^0)$, we define the log-likelihood function of the sample as:
\begin{align*}
    \theta \in \Theta \mapsto \mathcal{L}_n ((Y,\Delta),\theta) &  = \frac{1}{n} \sum_{i=1}^n \log f_{(Y,\Delta),\theta}((Y_i,\Delta_i))
\end{align*}
where $f_{(Y,\Delta),\theta}$ is defined in Lemma \ref{lem:dens_fct_censor} (here we emphasize the dependance w.r.t. parameter $\theta$),
and we define the maximum likelihood estimator as a parameter vector $\hat{\theta}_n$:
\begin{align*}
    \hat{\theta}_n & = \arg \underset{\theta \in \Theta}{\max} \mathcal{L}_n ((Y,\Delta),\theta)
\end{align*}

It is worth-noting that \cite{Whi82} defines an MLE estimator for the more general case of misspecified model. Our model lies in the family of well specified models, simplifying the verification of the assumptions proposed in the previous article. In order to lighten the main text, all the properties needed to ensure the existence, the consistency and the asymptotic normality of the MLE are properly defined and detailled in the mathematical appendix. To guarantee the existence of a measurable MLE, classical assumptions such as the compactness of the space of parameters {\bf (H3)} and continuity in parameters of the density functions {\bf (H4)} are made.
\begin{itemize}
    \item \textbf{(H3).} $\Lambda$, $\Theta_1$ and $\Theta_2$ are respectively compact subsets of $\mathbb{R}_+^*$, $\mathbb{R}^{d_1}$ and $\mathbb{R}^{d_2}$ so that $\Theta = \Lambda \times \Theta_1 \times \Theta_2$ is a compact subset of the Euclidean space $\mathbb{R}^{1+d_1+d_2}$. And the true parameter $\theta^0$ is in the interior of $\Theta$.
    \item \textbf{(H4).} See Appendix, Assumption \ref{ass:consistency}.
\end{itemize}
Given the structure of the density function (Lemma \ref{lem:dens_fct_censor}), the continuity in $\lambda$ is immediately deduced.
Assumption \textbf{(H4)} is then settled to ensure the continuity of the respective coefficient functions $\alpha \mapsto c_{n,X} (\alpha)$ and $\beta \mapsto c_{n,Z} (\beta )$, for any $n \in \mathbb{N}$. 

\begin{theo}[Existence] \label{thm:existence_MLE}
Under assumptions \textbf{(H3)}-\textbf{(H4)}, there exists a measurable MLE  $\hat{\theta}_n$, for any $n \in \mathbb{N}^*$.
\end{theo}

\paragraph{}
Although the previous theorem ensures the existence of a MLE, an assumption on its uniqueness needs to be verified. The author recall in  \cite{Whi82} that well specified models admit a unique minimum to it. So that assumption A3)b) is always verified in our configuration. In order to obtain consistency, a uniform bound on the family of the log of the parametrized density functions is settled. A way to ensure this property is to consider the jumping variables $X_1$ and $Z_1$ with respectively regular support with respect to the parameter $\alpha$ and $\beta$. We then assume:
\begin{itemize}
    \item \textbf{(H5).}  The support of $X_1$ (resp. $Z_1$) does not depend on $\alpha$ (resp. $\beta$).
\end{itemize}
This assumption is also a common assumption made with parametrized family of distributions. It encompasses a wide class of distributions. Only few parametric family such as uniform laws on $[0,a]$, $a \in \Theta_1 \subset \mathbb{R}_+^*$, have varying support. Now that all the tools required are defined, we obtain the following theorem.
\begin{theo}[Consistency] \label{thm:consistency}
Assume that $\textbf{(H3)-(H5)}$ are verified. 
Considering either Model I under $\textbf{(H1)}$ or $\textbf{(H2.i)}$ or Model II under $\textbf{(H2.i)}$, we have for almost every sequence $((Y_i,\Delta_i))_{i \geq 1}$,  $\hat{\theta}_n \underset{n \to + \infty}{\longrightarrow{}} \theta^0$. 
\end{theo}

The asymptotic normality requires more regularity on the family of parametrized density functions. In particular, it needs to be of class $\mathcal{C}^2$ on the space of parameter in a such way that derivatives of first and second orders admit proper uniform dominating integrable function. Assumptions are required to ensure the functions $\alpha \mapsto c_{n,X} (\alpha)$ and $\beta \mapsto c_{n,Z} (\beta ) $ to be of class $\mathcal{C}^2$:
\begin{itemize}
    \item {\bf (H6).} See details in Appendix, Assumption \ref{ass:normality}. 
\end{itemize}
A major difficulty is to obtain functions that uniformly dominate the respective derivatives of first and second order. The compactness of the space of parameters {\bf (H3)} and the regularity of the support of the jumping variables {\bf (H5)} constitute the main theoretical arguments that allow to justify the existence of such functions.

To study the asymptotic normality, when $\textbf{(H6)}$ is verified, we define the matrices
\begin{align} \nonumber
    A_n(\theta) = & \left( \frac{1}{n} \sum_{i=1}^n \partial^2 \log f_{(Y,\Delta),\theta}((Y_i,\Delta_i)) / (\partial \theta_j \partial \theta_k) \right)_{1\leq j,k \leq n} \\ \nonumber
    B_n(\theta) = & \left( \frac{1}{n} \sum_{i=1}^n \partial \log f_{(Y,\Delta),\theta}((Y_i,\Delta_i))/ \partial \theta_j \; . \; \partial \log f_{(Y,\Delta),\theta}((Y_i,\Delta_i))/ \partial \theta_k  \right)_{1\leq j,k \leq n} \\  \label{eq:matrix_A_theta}
    A(\theta) = & \left( \mathbb{E} [ \partial^2 \log f_{(Y,\Delta),\theta}((Y,\Delta))/ (\partial \theta_j \partial \theta_k) ] \right)_{1\leq j,k \leq n} \\ \label{eq:matrix_B_theta}
    B(\theta) = & \left( \mathbb{E} [ \partial \log f_{(Y,\Delta),\theta}((Y,\Delta))/ \partial \theta_j \; . \; \partial \log  f_{(Y,\Delta),\theta}((Y,\Delta))/ \partial \theta_k ] \right)_{1\leq j,k \leq n} \ . 
\end{align}
In our framework $A(\theta)=-B(\theta)$ (see Lemma \ref{lem:A_equal_-B}). When the appropriate inverses exist, define
$$
    C_n(\theta)  =  A_n(\theta)^{-1} B_n(\theta) A_n(\theta)^{-1}, \qquad 
    C(\theta)  = - A(\theta)^{-1}.
$$

\begin{theo}[Asymptotic normality]     \label{theo:assymptotic_normaility}
    Assume that $\textbf{(H3)-(H6)}$ are verified and that $A(\theta^0)$ is non-singular. 
    Considering either Model I under $\textbf{(H1)}$ or $\textbf{(H2.i)}$ or Model II under $\textbf{(H2.ii)}$, we have the following asymptotic normality: in law 
    \begin{eqnarray*}
        \lim_{n\to +\infty} \sqrt{n} ( \hat{\theta}_n - \theta^0) \overset{d}{=} \mathcal{N} ( 0, C(\theta^0))
    \end{eqnarray*}
    where $\mathcal{N} ( 0, C(\theta^0))$ denotes a normal distribution of expected value 0 and covariance matrix $C(\theta^0)$. 
    Moreover, $C_n(\hat{\theta}_n)$ converges a.s. to $C(\theta^0)$, element by element.
\end{theo}

\section{Simulation and real data analysis} 
\label{sect:simulation}
\subsection{A small simulation study}\label{ssect:simulation}
In this subsection, we illustrate the performance of the proposed MLE estimator $(\widehat{\lambda}_n, \widehat{\alpha}_n,\widehat{\beta}_n)$ and on the estimated distribution functions of the marginal up-crossing variables T and C. We denote $\widehat{F}_{T,\lambda,\alpha,n}$ and $\widehat{F}_{C,\lambda,\beta,n}$ those functions. We consider $N=100$ samples of sizes $n = 50,100,200$ under the following models:
\begin{enumerate}
    \item[a-] $\lambda = 1.42$, $X_1 \sim \mathcal{B}(\alpha)$, with $\alpha = 0.36$ and $Z_1 = 1$ a.s., with $(x,z) = (7,17)$,
    \item[b-] $\lambda = 1.42$, $X_1 \sim \mathcal{E}(\alpha)$, with $\alpha = 0.71$ and $Z_1 \sim \mathcal{E}(\beta)$, with $\alpha = 2.04$, with $(x,z) = (14,7)$,
    \item[c-] $\lambda = 1.42$, $X_1 \sim \mathcal{B}(\alpha)$, with $\alpha = 0.36$ and $Z_1 \sim \mathcal{P}(\beta)$, with $\beta = 1.23$, with $(x,z) = (7,19)$,
\end{enumerate}
where $\mathcal{B}$, $\mathcal{E}$ and $\mathcal{P}$ respectively denote the Bernoulli, the Exponential and the Poisson distributions. The choice of those three model is made in order to illustrate all the identifiable models proposed in Section \ref{sect::identifia}. Examples a, b and c respectively belong to (H1), (H2)i) and (H2)ii). In order to obtain a qualitative measure of the efficiency of the estimator, we drew the mean squared error denoted as $nx$ and the $bias$ function of the estimated marginal distributions with respect to the true distribution functions. They are defined as follows. We denote $(\widehat{F}_{T,\lambda,\alpha,n}^{(k)})_{k \in \{1, \ldots , N \}}$, $(\widehat{F}_{C,\lambda,\beta,n}^{(k)})_{k \in \{1, \ldots , N \}}$ the marginal estimated distribution functions and $F_T$, $F_C$ the true marginal distribution functions. The $nx$ and $bias$ functions are defined as follows:
\begin{eqnarray*}
    nx_T (t) = \frac{1}{N} \sum_{k=1}^N (\widehat{F}_{T,\lambda,\alpha,n}^{(k)} (t)- F_T (t))^2 \quad \quad\text{and} \quad \quad bias_T (t) = \frac{1}{N} \sum_{k=1}^N | \widehat{F}_{T,\lambda,\alpha,n}^{(k)} (t) - F_T (t) | \quad ; 
\end{eqnarray*}
$nx_C$ and $bias_C$ are defined the same way. \\
Boxplots from figures 1-3 illustrate the consistency of the estimator. Contrarily to case a, the cases b and c require to estimate three parameters instead of two, inducing a larger variance of the estimation. 
\begin{figure}[H]
\center
    \includegraphics[scale = 0.21]{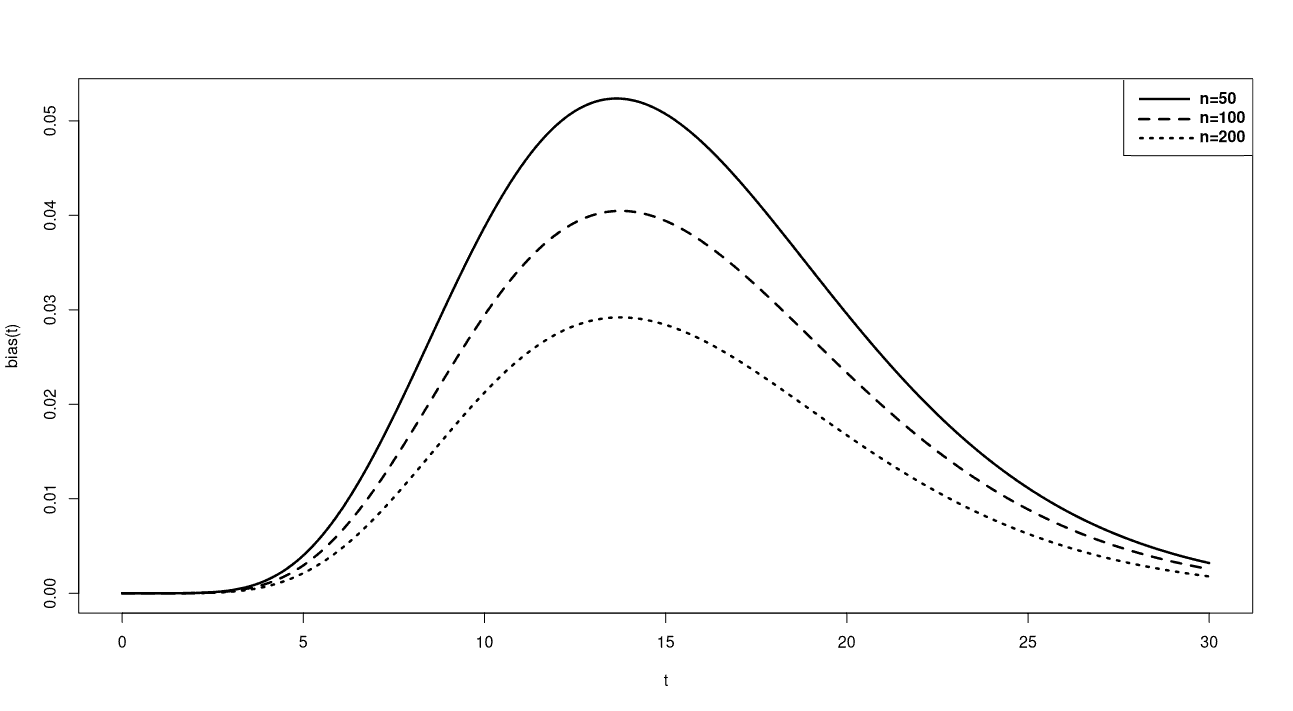}
    \includegraphics[scale = 0.21]{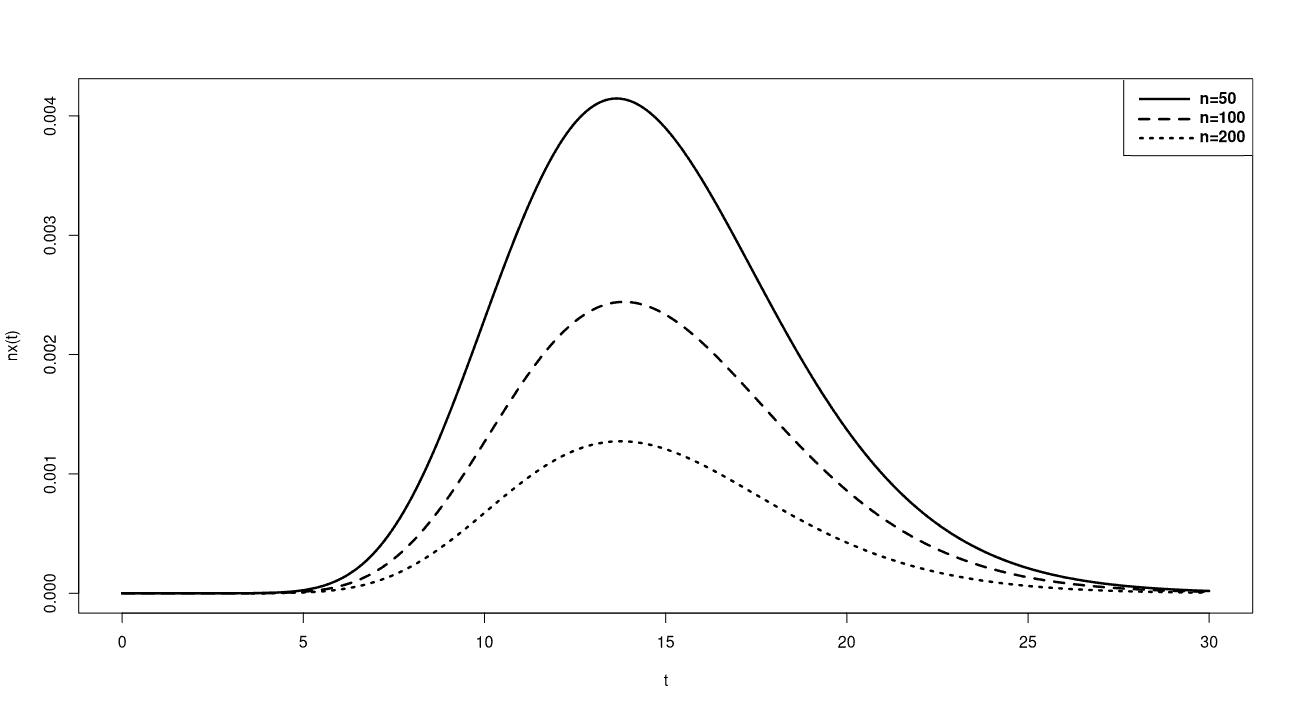}\\
    \includegraphics[scale = 0.3]{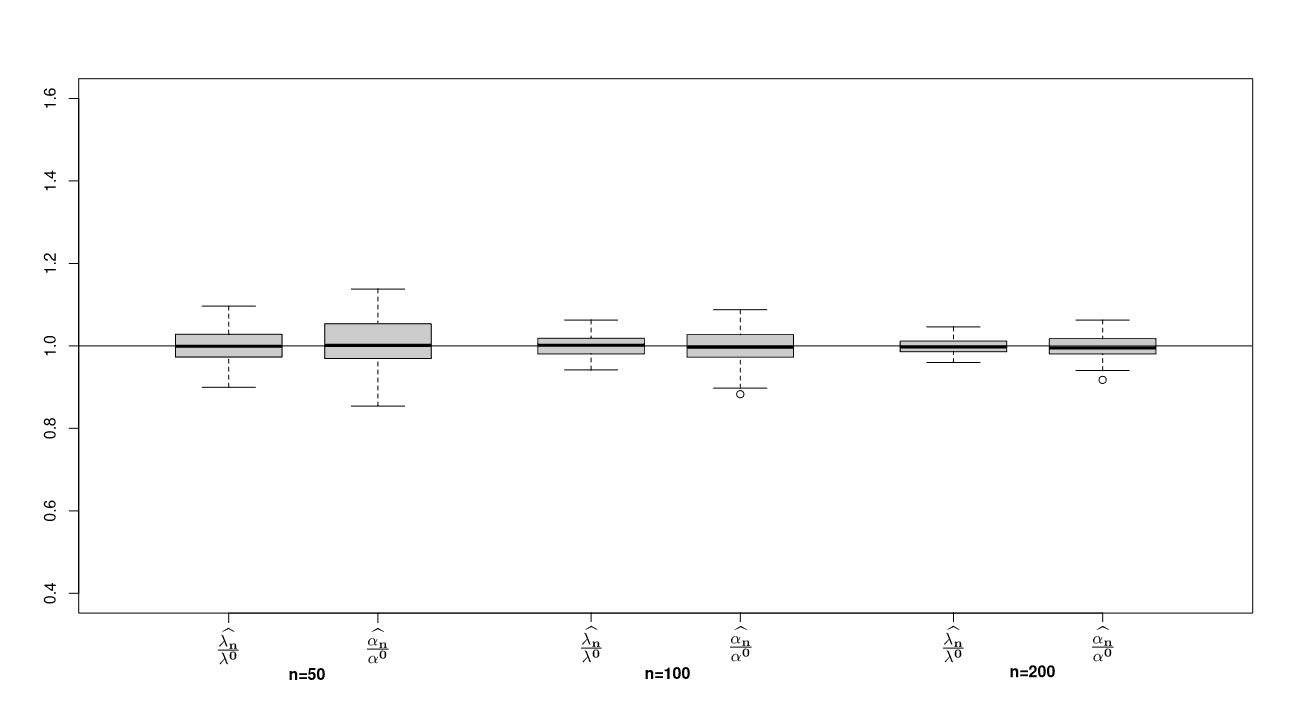}
    \caption{Case a - Bias function $bias_T(t)$ (left), mean squared error function $nx_T(t)$ (right) and boxplot of N= 100 realisations of the estimator divided by the true value,  for n=50, 100, 200.} 
\end{figure}

\begin{figure}[H]
\center
    \includegraphics[scale = 0.21]{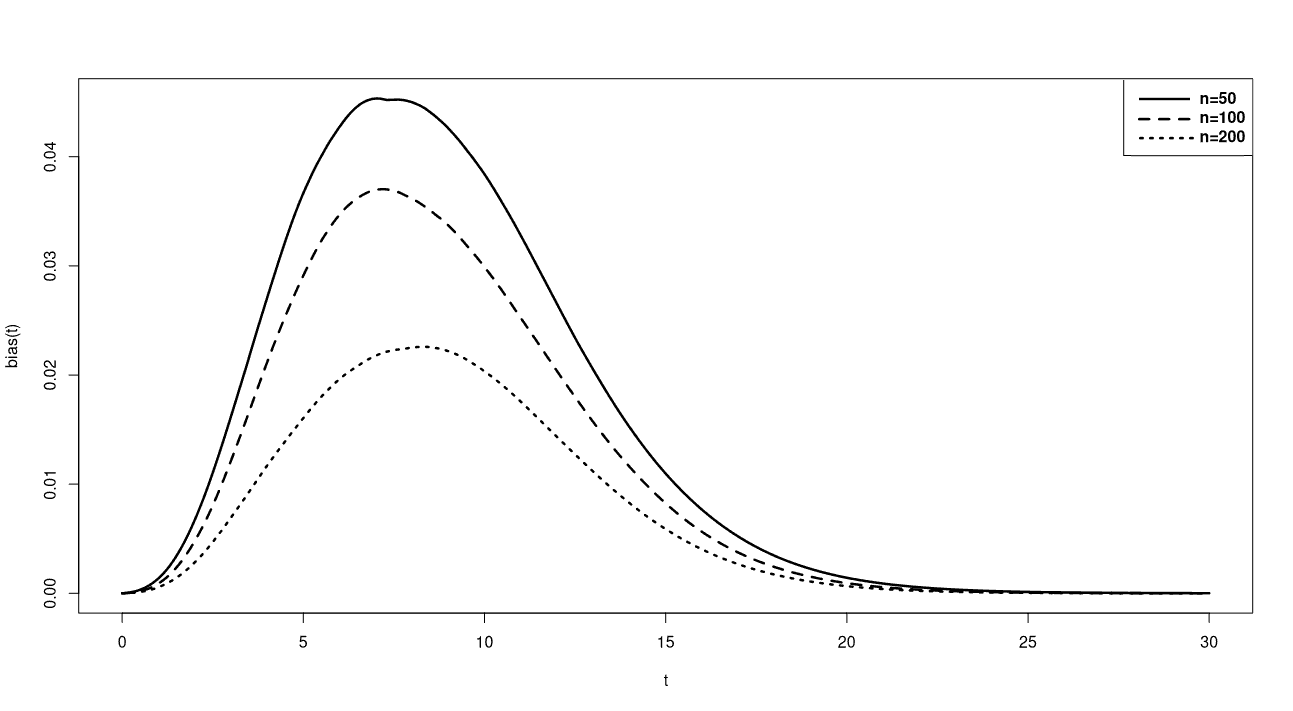}
    \includegraphics[scale = 0.21]{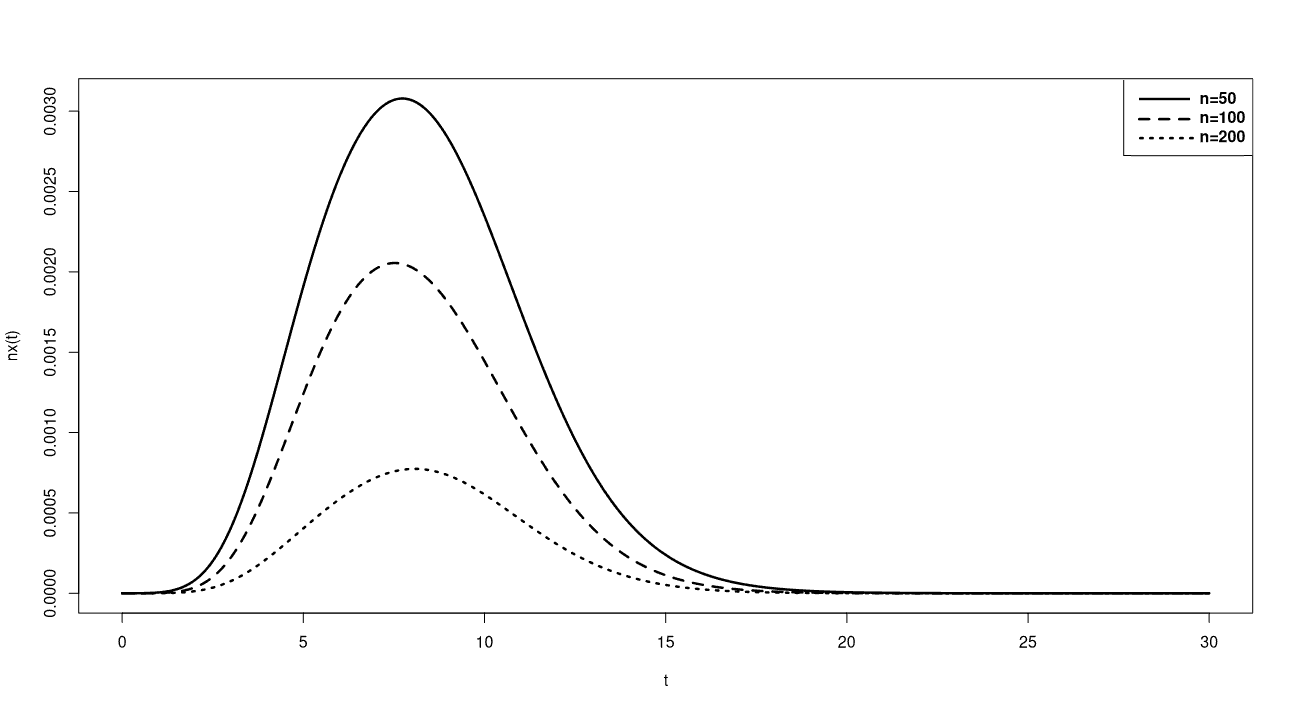} \\
    \includegraphics[scale = 0.3]{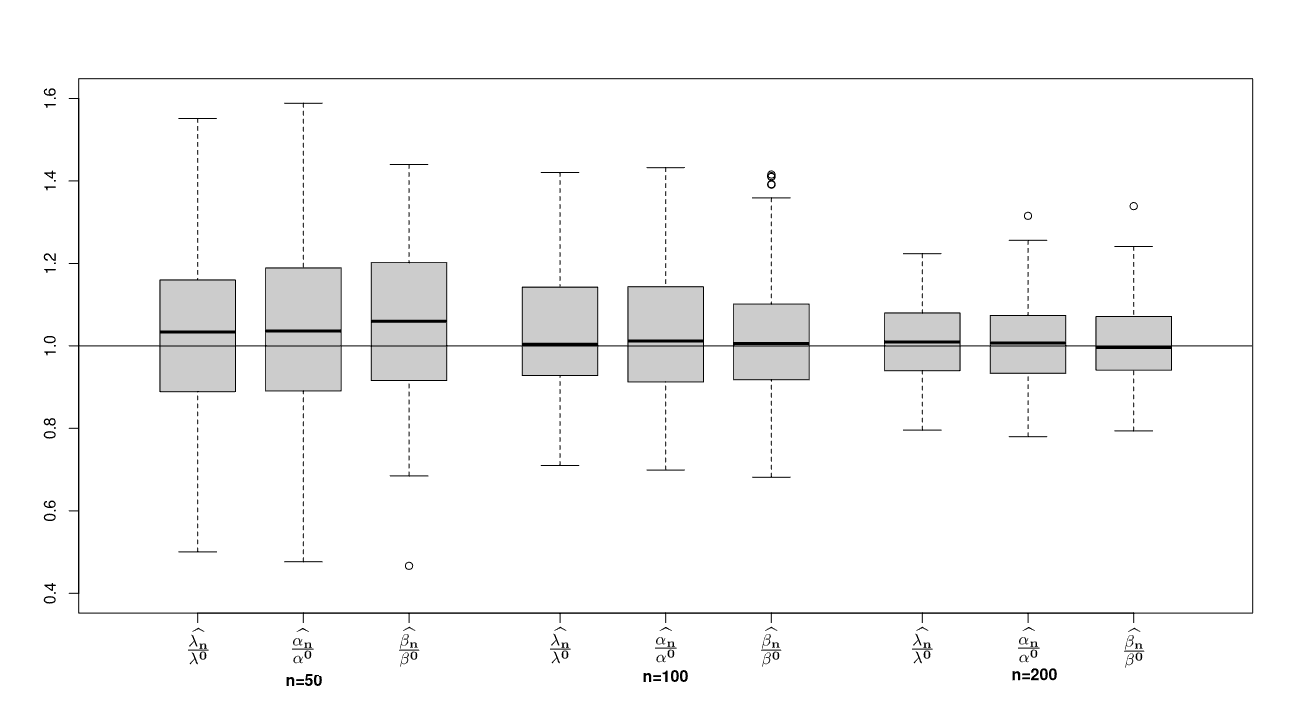}
    \caption{Case b - Bias function $bias_T(t)$ (left), mean squared error function $nx_T(t)$ (right) and boxplot of N= 100 realisations of the estimator divided by the true value, for n=50, 100, 200.}
\end{figure}

\begin{figure}[H]
\center
    \includegraphics[scale = 0.21]{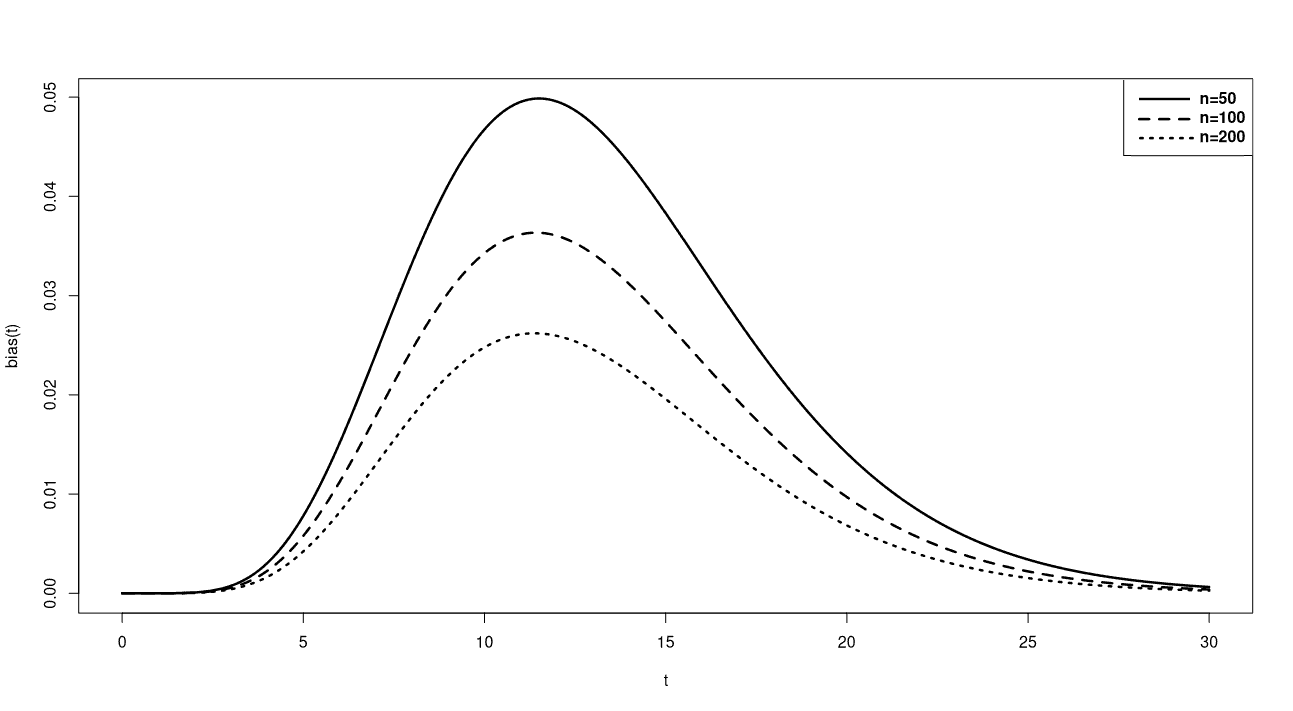}
    \includegraphics[scale = 0.21]{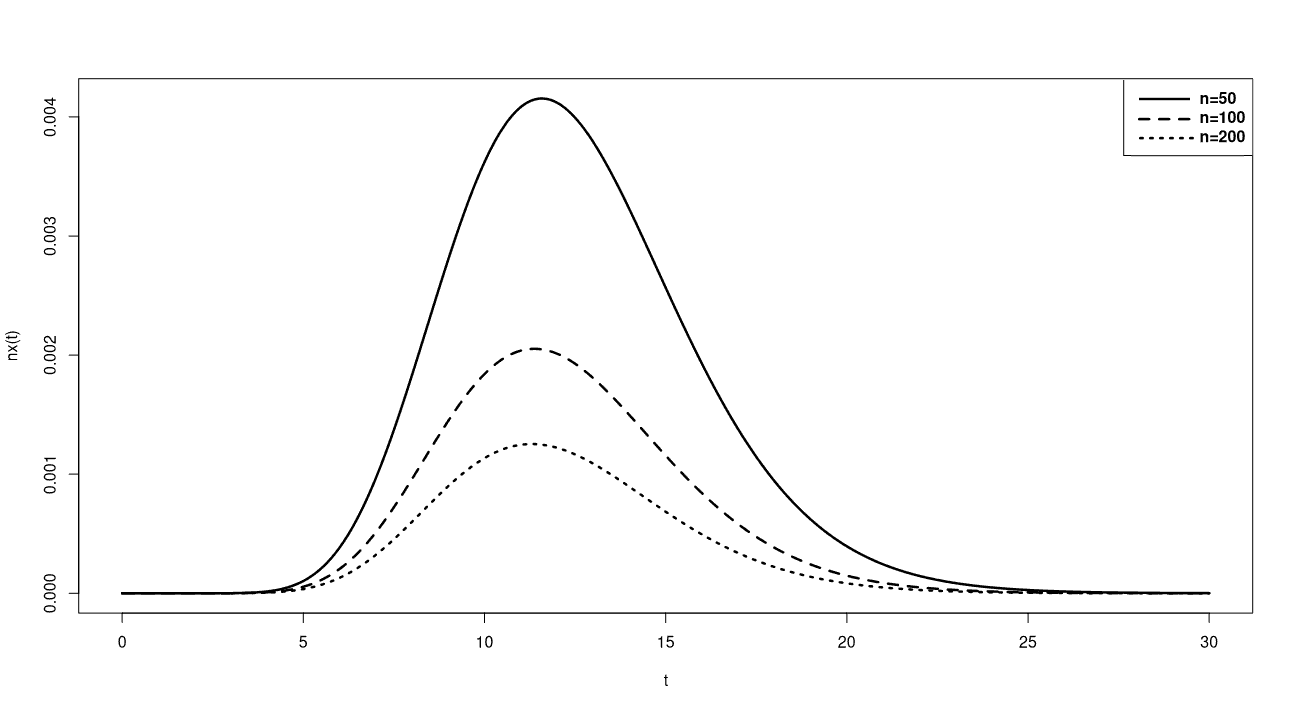} \\
    \includegraphics[scale = 0.3]{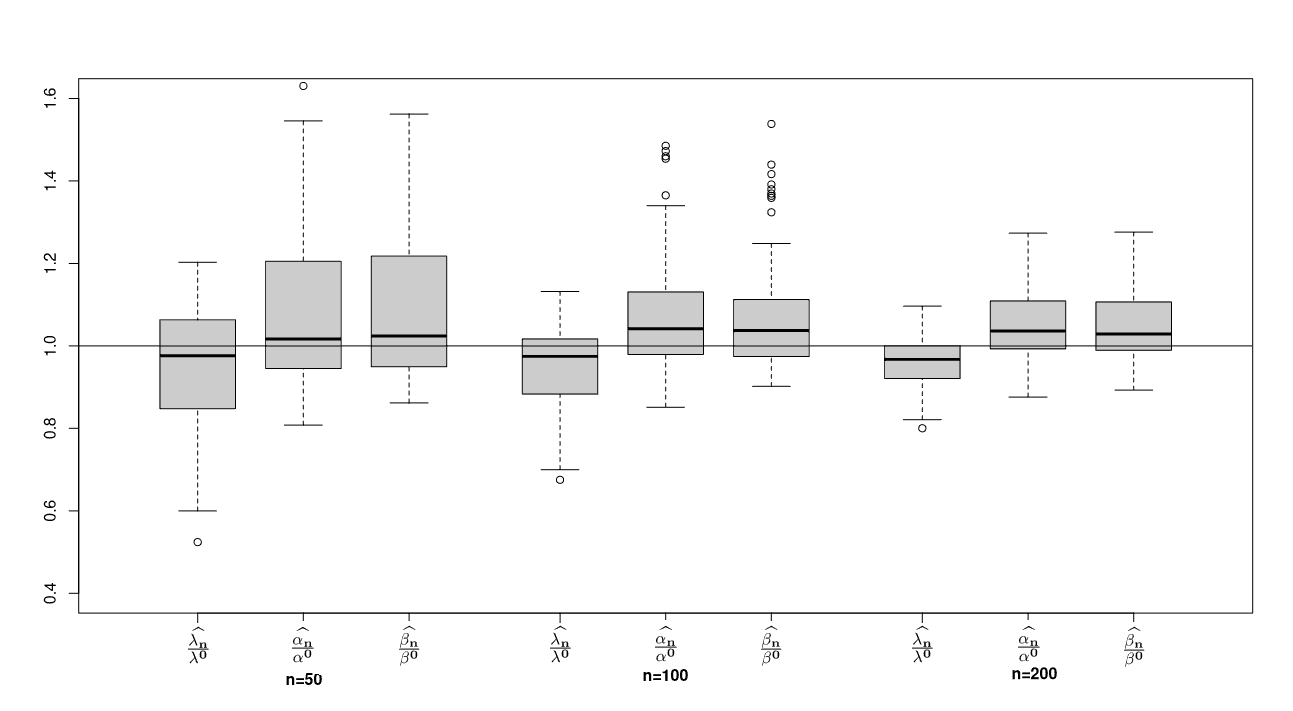}
    \caption{Case c - Bias function $bias_T(t)$ (left), mean squared error function $nx_T(t)$ (right) and boxplot of N= 100 realisations of the estimator  divided by the true value, for n=50, 100, 200.}
\end{figure}

\paragraph{Remark}
Although those examples perfectly show the consistency of the estimator, some remarks can be done. Firstly, the density of the model is composed of an infinite sum of terms. In practice, implementing the estimation requires to consider an approximation of the density function with only a finite number of terms in the sum. If we consider too less coefficients, the approximated density could be not enough close to the real density function, inducing a bad estimation, and if otherwise, we consider too much coefficients, it drastically extend the time of computation of the estimation. In those examples, we considered the first 20 terms of the sum. Secondly, each family of parameters induces a different variance, which is clearly illustrated in the last example, where the variance of the second case is visually higher than the first one.

\newpage

\subsection{Real data application: case of amanita poisoning}

%In this subsection, we evaluate the performance of our model on real dataset. 

\paragraph{Scientific context:} Poison control centers regularly respond to mushroom poisonings such as amanita. Toxins present in these mushrooms are rare but can significantly cause of acute liver failure. As an effective antidote, N-acetylcysteine is administered when a patient arrives at the hospital, preventing many cases of liver failure. Since it is rare to know the exact amount of amanita ingested, it becomes difficult to anticipate the level of risk. When the liver processes a toxin, it produces an enzyme called alanine aminotransferase (ALAT). Known thresholds for this enzyme are used to measure patient risk. Another commonly used prognostic indicator is Factor V (FV), expressed as a percentage. A healthy liver has a Factor V of $100\% $, while less than $30\%$ requires transplantation. After amanita ingestion, ALAT and FV levels show monotonic behavior until they reach their respective peaks. However, patient records of ALAT and FV levels are discontinuous and spaced at random times. When one of the measurements reaches critical levels, decisions must be made to operate the patient. Therefore, it is common to have censored times during follow-up.\\ 
To describe the evolution of ALAT and FV levels after ingestion, we consider two monotonic random processes. The measurements of these quantities are performed simultaneously. The times at which patients are tested are random, and we can assume that the inter-periods follow independent random variables with a common exponential distribution. These properties lead us to believe that our model is a suitable candidate to describe the dynamics of ALAT and FV.\\ %Depending on the quantity of amanita ingested, the value of the peak of ALAT varies from a patient to another. It is worth-noting that regardless of the value of the peak, the dynamic of evolution of the ALAT is proportional to it (the value of the peak).Notable facts on this configuration justify the model proposed in this article as a convenient candidate to describe the dynamics of ALAT and FV. They are respectively monotonic quantities until they reach their respective peaks. Measurements of this quantities are made at simultaneously at random times. Well-known fixed thresholds for both quantities are commonly used to evaluate the level of risk of a patient. This leads us to define the following description.

\noindent
\paragraph{Mathematical description:} For each individual, we consider the processes $L_1$ and $L_2$ which describe the levels of ALAT and FV respectively. Since everyone has a different range of ALAT values, we assume $L_1$ represents rescaled levels between 0 and 1, so it maximizes near 1. Based on the patient protocol, a recorded time is of interest if it returns the moment a patient reaches his maximum ALAT or when his FV drops below 30 $\%$. Understanding the distribution of these moments therefore provides important information for patient monitoring and care. Because treatment effects typically occur a few days after administration, it is common to have censored data with sometimes equal recording times for ALAT and FV events.\\ 

In this study, we collected $N=81$ patient data with a complete record of time to ALAT maximum and FV drop. The latter are denoted $(T_i)_{1\leq i\leq N}$ and $(C_i)_{1\leq i\leq N}$, respectively, and are used to construct the censored dataset $\{(Y_i,\Delta_i), i \in \{1,\dot,N\}$ based on the \textbf{Model I}. The process are formally described by
\begin{eqnarray*}
    L_{t,1} = \sum_{n=1}^{N_t} X_n\quad\text{and}\quad L_{t,2} = \sum_{n=1}^{N_t} Z_n, \qquad \forall t \geq 0.
\end{eqnarray*} 
where we assume that the jump sizes follows exponential distributions with parameters $\alpha$ and $\beta$ respectively. According to the aforementioned protocol, we consider the thresholds $x=0.95$ and $z= 0.65$. It is worth-mentioning that in this configuration, the exponential parameters and the thresholds share some proportional stability, in the sense that $\alpha x$ and $\beta z$ are constants. This property follows from the stability of exponential distributions to scalar multiplications and was observed in our analysis.\\

Table \ref{tabl:: results} shows the numerical results of the estimators for different subsamples of the dataset. The abbreviations N, M, F, $A_<$, $A_>$, $CVR_-$ and $CVR_+$ respectively denote no-categorization, male, female, patients below 59 years old, patients above 60 years old, patients with and without history of cerebrovascular risk. The probabilities of censoring are estimated throughout the empirical proportions $\mathbb{P}_{\text{Emp}}[ T \leq C]$ or based on our model with $\mathbb{P}_{\text{Est}}[ T \leq C]$. The curves of the distribution functions obtained from our model are compared with their empirical counterparts in Figure \ref{fig:amanita_1}. Similar plots are shown in Figures \ref{fig:amanita_2} and \ref{fig:amanita_3}, where The population is divided into patients over and under the age of 59. In Figure \ref{fig:amanita_4}, we illustrate the dependence structure of the couple $(T,C)$ with an estimate of the bi-variate distribution function based on our model.

\begin{table}[ht]
\centering
\begin{tabular}{|c|c|l|r|r|r|r|}
\hline  
Threshold values $(x,z)$ &  Category  & $\widehat{\lambda}_n$ & $\widehat{\alpha}_n$ & $\widehat{\beta}_n$ & $\mathbb{P}_{\text{Emp}}[ T \leq C]$ & $\mathbb{P}_{\text{Est}}[ T \leq C]$ \\  
\hline 
$(9.5,6.5)$  & N & $7.889$ & $2.5673$ & $3.3602 $ & 0.3950617 & 0.381685\\ 
\hline 
$(4.75,3.9)$  & N & $7.889$ & $5.1346$ & $5.6003 $ & 0.3950617 & 0.381685\\ 
\hline
$(0.95,0.65)$ & N & $7.889$ & $25.673$ & $33.602 $ & 0.3950617 & 0.381685\\  
\hline
$(0.95,0.65)$ & M & $7.889$ & $25.673$ & $33.602 $ & 0.3684211 & 0.381685\\  
\hline 
$(0.95,0.65)$ & F & $7.889$ & $25.673$ & $33.602 $ &  0.4186047 & 0.381685\\  
\hline 
$(0.95,0.65)$ & $A_<$ & $7.889$ & $25.673$ & $33.602 $ & 0.5 & 0.381685 \\  
\hline 
$(0.95,0.65)$ & $A_>$ & $7.889$ & $25.673$ & $33.602 $ & 0.3023256 & 0.381685\\  
\hline 
$(0.95,0.65)$ & $CVR_-$ & $7.889$ & $25.673$ & $33.602 $ & 0.4 & 0.381685 \\  
\hline 
$(0.95,0.65)$ & $CVR_+$ & $7.889$ & $25.673$ & $33.602 $ & 0.3846154 & 0.381685 \\  
\hline 
\end{tabular}
\caption{Parameter estimation result in accordance with threshold values and population categorization}
\label{tabl:: results}
\end{table}

The first three lines of the table underline the proportional stability of the exponential law mentioned above. Furthermore, it is worth noting that the estimated parameter $\widehat{\lambda}_n$ remains unchanged when the thresholds are changed, since it only affects the distribution of the jump sizes of both processes. Such a phenomenon is coherent since the expected number of observations needed to observe ALAT or FV peaking remains unchanged throughout the records. The last four rows also illustrate the accuracy of the estimator whether the group is categorized or not, especially with subcategories containing only 26 patients. Although our model do not take on consideration the value of the peaks for each patient, those categorization show that parameters such as sex, age, history of cerebrovascular risk have no impact on the moment of reaching the peak of ALAT and FV. In particular, our estimator still provides a good estimate even with reduced dataset sizes. This is also observed with comparison between the empirical and modeled curves in Figures \ref{fig:amanita_1}, \ref{fig:amanita_2} and \ref{fig:amanita_3} since we observe few discrepancies between the results. The curves from our model seem to correctly intercept all event times, even though it is based on censored data, and particularly for the largest observations. Lastly, Figure \ref{fig:amanita_4} and \ref{fig:amanita_5} illustrate the dependence structure of the couple $(T,C)$. The property of non-absolute continuity of the pair can be observed in the left plot with abrupt changes in the curve trajectories in the set $\{(u,v)\in\mathbb{R}^2_+,\, u=v \}$. This was expected from the theoretical model construction and analysis, and reflects the possibility of equal recording times among the data.

\subsection*{Acknowledgement}
The authors thank the members of the Poisoning Center and the Toxicological Surveillance Unit of Angers (CHU/Univ), in particular Dr. Bruneau, Dr. Lecot, Dr. Cellier, Prof. Descatha and Dr. Le Roux, for providing data and helping with the real application part.

\begin{figure}[H]
\centering
    \includegraphics[scale = 0.46]{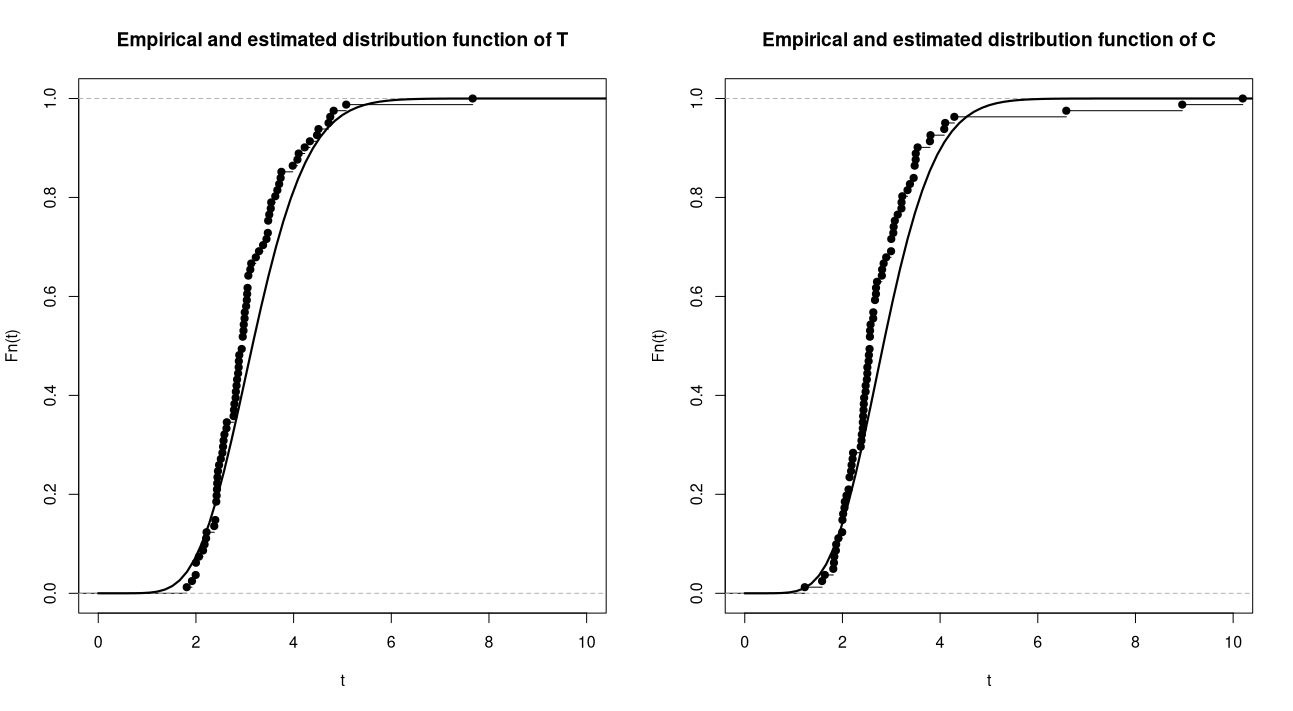}
    \caption{Plot of the estimated (line) and empirical (simple function) cumulative distribution functions of T (left) and C (right). }
    \label{fig:amanita_1}
\end{figure}

\begin{figure}[H]
\center
    \includegraphics[scale = 0.46]{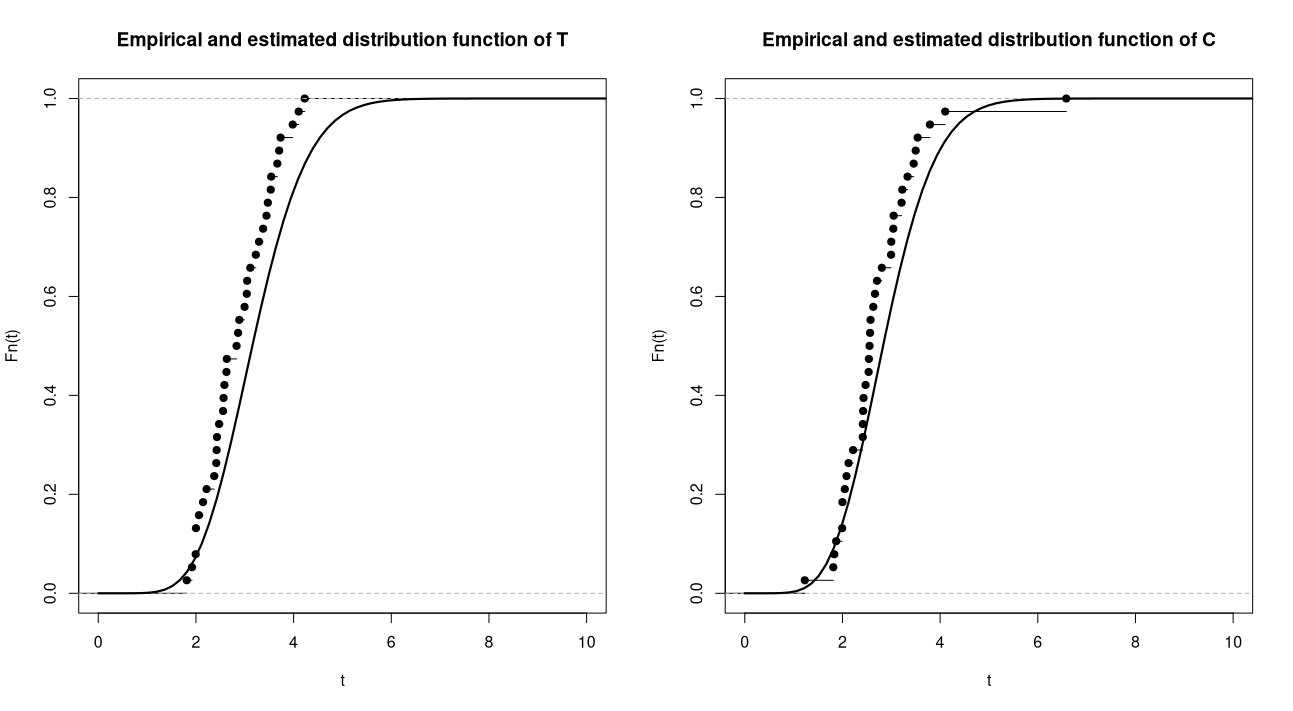}
    \caption{Plot of the estimated (line) and empirical (simple function) cumulative distribution functions of T (left) and C (right) for patient younger than 59 years old. }
    \label{fig:amanita_2}
\end{figure}

\begin{figure}[H]
\center
    \includegraphics[scale = 0.46]{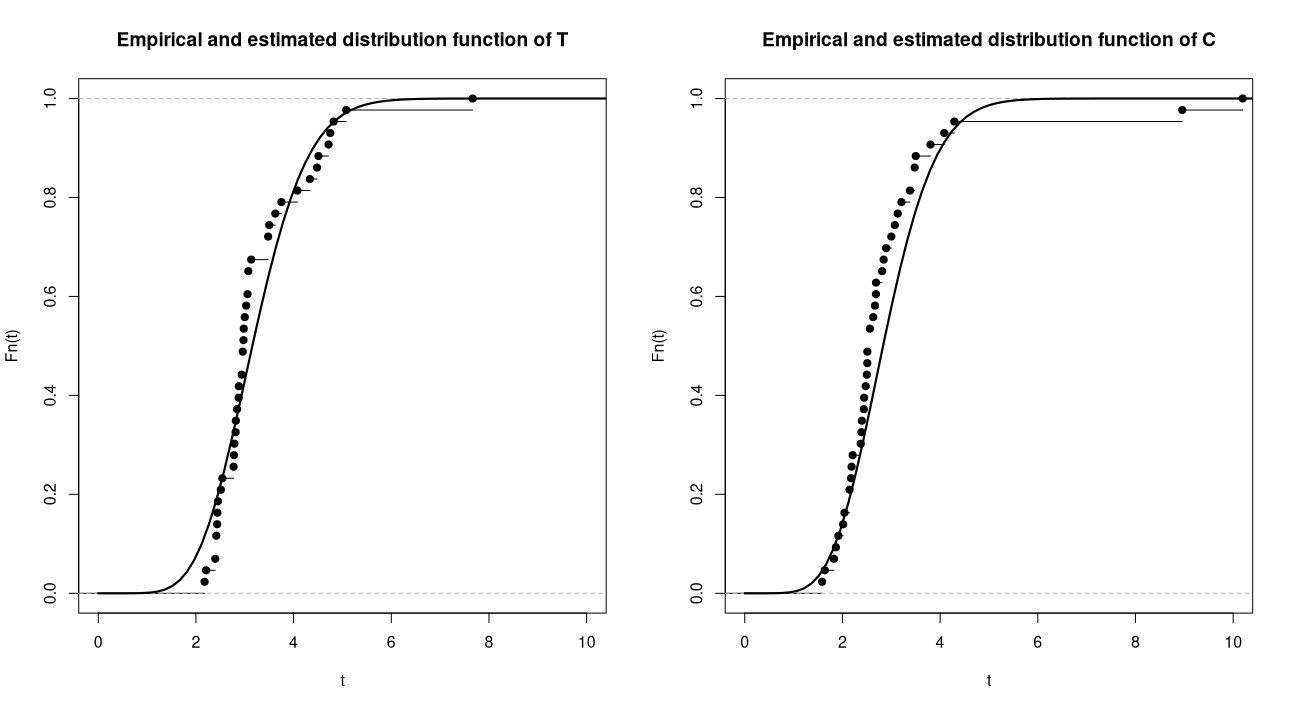}
    \caption{Plot of the estimated (line) and empirical (simple function) cumulative distribution functions of T (left) and C (right) for patient older than 60 years old. }
    \label{fig:amanita_3}
\end{figure}

\begin{figure}[H]
\center
    \includegraphics[scale = 0.35]{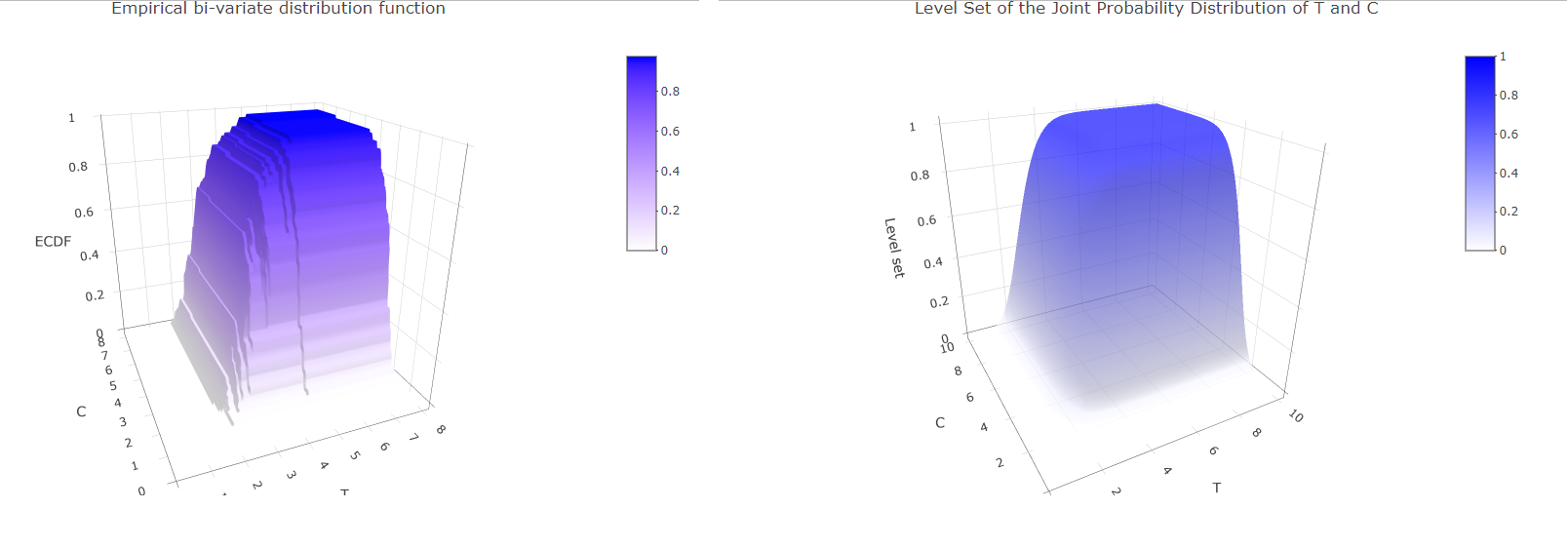}
    \caption{Plot of the empirical (left side) and estimated (right side) joint probability distribution of T and C}
    \label{fig:amanita_4}
\end{figure}

\begin{figure}[H]
\center
    \includegraphics[scale = 0.45]{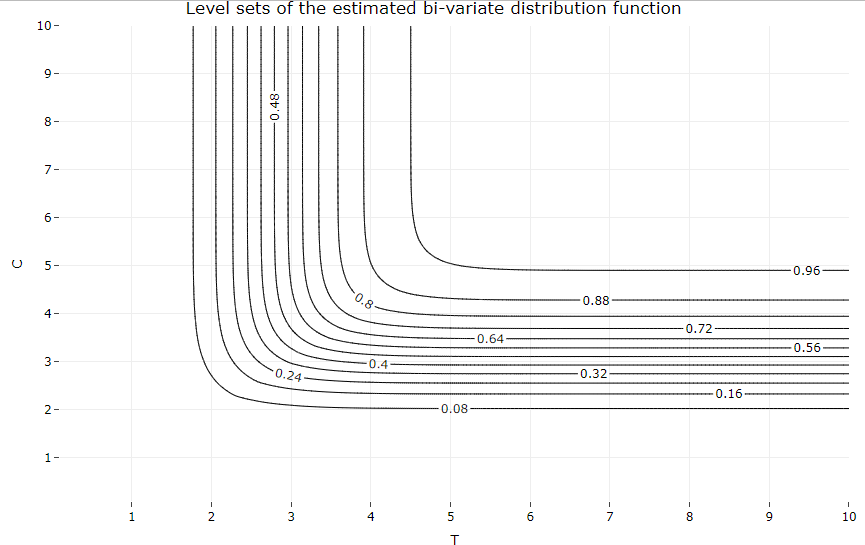}
    \caption{Plot of the level sets of the estimated joint probability distribution of T and C}
    \label{fig:amanita_5}
\end{figure}

%\newpage

\bibliographystyle{abbrv}
\bibliography{articleversion}

\section{Appendix - Mathematical proofs} \label{sect:appendix}

\paragraph{Proof of Lemma \ref{lem:dens_fct_hitting_times}}
Firstly, let us compute the distribution function of the couple of hitting time $(T,C)$. We have for $u \geq v$ \\
\begin{align*} \nonumber
    & \mathbb{P}[ T \leq u , C \leq v ] \\ 
    & \quad =   \sum_{n=1}^{+ \infty} \sum_{j=1}^{n} \mathbb{P} [X_1 + \ldots + X_n \geq x] \mathbb{P} [Z_1 + \ldots + Z_j \geq z] e^{- \lambda u} \frac{ (\lambda u)^n }{ n! } \binom{n}{j} \left( \frac{v}{u} \right)^j \left( 1 - \frac{v}{u} \right)^{n-j} ,% \quad \text{if} \quad u \geq v \\
   % & \sum_{j=1}^{+ \infty} \sum_{n=1}^{j} \mathbb{P} [X_1 + \ldots + X_n \geq x] \mathbb{P} [Y_1 + \ldots + Y_j \geq y] e^{- \lambda v} \frac{ (\lambda v)^j }{ j! } \binom{j}{n} \left( \frac{u}{v} \right)^n \left( 1 - \frac{u}{v} \right)^{j-n} \quad \text{else}
\end{align*}
and for $v > u$
\begin{align} \nonumber
    & \mathbb{P}[ T \leq u , C \leq v ] \\ \label{eq:cdf_T_C}
     %\quad =   \sum_{n=1}^{+ \infty} \sum_{j=1}^{n} \mathbb{P} [X_1 + \ldots + X_n \geq x] \mathbb{P} [Y_1 + \ldots + Y_j \geq y] e^{- \lambda u} \frac{ (\lambda u)^n }{ n! } \binom{n}{j} \left( \frac{v}{u} \right)^j \left( 1 - \frac{v}{u} \right)^{n-j} ,% \quad \text{if} \quad u \geq v \\
    &\quad = \sum_{j=1}^{+ \infty} \sum_{n=1}^{j} \mathbb{P} [X_1 + \ldots + X_n \geq x] \mathbb{P} [Z_1 + \ldots + Z_j \geq z] e^{- \lambda v} \frac{ (\lambda v)^j }{ j! } \binom{j}{n} \left( \frac{u}{v} \right)^n \left( 1 - \frac{u}{v} \right)^{j-n} . %\quad \text{else}
\end{align}
Since the proof use the same ideas on both of the subsets $\{0\leq u \leq v \}$ and $\{0\leq v < u \}$, we only prove it when $ u \geq v $. \\
Since both processes $L_1$ and $L_2$ are non-decreasing processes, using the total probability formula and using the fact that those processes are subordinators we obtain
\begin{align*}
    \mathbb{P}[ T \leq u , C \leq v ]  = & \mathbb{P}[ L_{u,1} \geq x , L_{v,2} \geq z  ] \\
    = & \sum_{n=1}^{+ \infty} \sum_{j=1}^{+ \infty} \mathbb{P} [ L_{u,1} \geq x ; L_{v,2} \geq z ; N_{v} = j ; N_{u} = n ] \\
    = & \sum_{n=1}^{+ \infty} \sum_{j=1}^{n} \mathbb{P} [ L_{u,1} \geq x ; L_{v,2} \geq z | N_{v} = j ; N_{u} = n ] \mathbb{P} [ N_{v} = j ; N_{u} = n ] \\
    = & \sum_{n=1}^{+ \infty} \sum_{j=1}^{n} \mathbb{P} [ X_1 + \ldots + X_n  \geq x ; Z_1 + \ldots + Z_j \geq z ] \mathbb{P} [ N_{v} = j ; N_{u} - N_{v} = n -j ]
\end{align*}
The change of index in the second sum comes from the fact that since $u \geq v$, necessarily $(N_t)_{t \geq 0}$ has jumped more at the time $u$ than at the time $v$.
By independence of $(X_i)_{i \geq 1}$ and $(Z_i)_{i \geq 1}$ and by the definition of a Poisson process, we get
\begin{align*}
= & \sum_{n=1}^{+ \infty} \sum_{j=1}^{n} \mathbb{P} [ X_1 + \ldots + X_n  \geq x ] \mathbb{P} [ Z_1 + \ldots + Z_j \geq z ] e^{- \lambda v} \frac{ (\lambda v)^j }{ j! } e^{- \lambda (u - v)} \frac{ (\lambda (u - v))^{n-j} }{ (n-j)! }  \\
= & \sum_{n=1}^{+ \infty} \sum_{j=1}^{n} \mathbb{P} [X_1 + \ldots + X_n \geq x] \mathbb{P} [Z_1 + \ldots + Z_j \geq z] e^{- \lambda u} \frac{ (\lambda u)^n }{ n! } \binom{n}{j} \left( \frac{v}{u} \right)^j \left( 1 - \frac{v}{u} \right)^{n-j}  
\end{align*}

\medskip

The law of the couple $(T,C)$ is not absolutely continuous with respect to the Lebesgue measure on $\mathbb{R}^2$. It is the case outside of the diagonal $\{ (u,u), \; u \in \mathbb{R}_+ \}$. A simple differentiation is enough to obtain the density function. Along the diagonal, another proof needs to be done. 

\paragraph{\it Density function of the couple $(T , C )$ outside of the diagonal.}
We want to prove that the density function $f_{ac}$ of the couple $(T , C )$ on the subset  $\{ u \ne v \}$ is given by \eqref{eq:density_abs_cont_part}.
%\begin{align*}
 %   f_{ (T , C )   }(u, v) \; = \; \sum_{n=1}^{+ \infty} & \sum_{j=1}^{n} \big[ c_{j-1,X} - c_{j,X} \big] \big[ c_{n,Z} - c_{n+1,Z} \big] 
   % \lambda^2 e^{ - \lambda v} \frac{ (\lambda (v - u))^{n-j} }{ (n-j)!  } \frac{( \lambda u )^{j-1}}{(j-1)!} \mathbb{1}_{ \{ u < v \} } \\
    %+ \sum_{n=1}^{+ \infty} & \sum_{j=1}^{n} \big[ c_{n,X} - c_{n+1,X} \big] \big[ c_{j-1,Z} - c_{j,Z} \big] 
    %\lambda^2 e^{ - \lambda u} \frac{ (\lambda (u - v))^{n-j} }{ (n-j)!  } \frac{( \lambda v )^{j-1}}{(j-1)!} \mathbb{1}_{ \{ u \geq v \} }.
%\end{align*}
%where $m$ denotes the Lebesgue measure on $(\mathbb{R}+,\mathcal{B}(\mathbb{R}+))$. \newline
In fact, $ (u , v) \in (\mathbb{R}_+)^2  \mapsto   \mathbb{P}[ T \leq u , C \leq v ] $ has a similar structure on $\{ u < v \} $ and $ \{ u > v \} $. Indeed the previous formulas for $\mathbb{P}[ T \leq u , C \leq v ]$ are symmetric on $(X_n)_{n \in \mathbb{N}}$ - $(Z_n)_{n \in \mathbb{N}}$ and $u$ - $v$. Thus only the proof on the subset $\{  u < v \}$ is detailed.\\

%Let $0 \leq u < v$. %To obtain the density function, using Schwarz's rule for two times differentiable functions, a simple differentiation of the distribution function gives the result:
%\begin{eqnarray*}
%(u ; v) \in (\mathbb{R}+)^2  \mapsto  \partial_{u} \partial_{v}  \mathbb{P}[ T \leq u ; C \leq v ]= \; \partial_{v} \partial_{u}  \mathbb{P}[ T \leq u ; C \leq v ] 
%\end{eqnarray*}

%By Schwarz, since the distribution function is on $(\mathbb{R}+^*)^2$ we know that: 
%\begin{eqnarray*}
 %   \partial_{u} \partial_{v}  \mathbb{P}[ T \leq u ; C \leq v ] \; = \; \partial_{v} \partial_{u}  \mathbb{P}[ T \leq u ; C \leq v ] 
%\end{eqnarray*}
%We can for example compute  $ \partial_{u} \partial_{v}  \mathbb{P}[ T \leq u ; C \leq v ] $. 
\noindent
On the subset $\{ u < v \}$, the distribution function is given by \eqref{eq:cdf_T_C} and can be rewritten as
$$
 \sum_{n=1}^{+ \infty} \sum_{j=1}^{n} [1 - c_{j,X}] [1 - c_{n,Z}] e^{- \lambda v} \frac{ (\lambda (v - u))^{n-j} }{ (n-j)! } \frac{ ( \lambda u)^j }{ j! } .
$$
To facilitate the calculus, we can separate this formula  in two terms: a first one where $j<n$ and a second one where $j=n$. Then, we get 
\begin{align*}
    \mathbb{P}[ T \leq u , C \leq v ] = & \sum_{n=2}^{+ \infty} \sum_{j=1}^{n-1} [1 - c_{j,X}] [1 - c_{n,Z}] e^{- \lambda v} \frac{ (\lambda (v - u))^{n-j} }{ (n-j)! } \frac{ ( \lambda u)^j }{ j! } \\
    & + \; \sum_{n=1}^{+ \infty} [1 - c_{n,X}] [1 - c_{n,Z}] e^{- \lambda v} \frac{ ( \lambda u)^n }{ n! } .
\end{align*}
Denote $A(u,v)$ the first term and $B(u,v)$ the second term. We obtain
\begin{align*}
 \partial_{v} A(u,v) \; = \; \sum_{n=2}^{+ \infty} \sum_{j=1}^{n-1} [1 - c_{j,X}] [1 - c_{n,Z}] \lambda e^{- \lambda v} \left[  \frac{( \lambda (v - u) )^{n-j-1}}{(n-j-1)!}  -  \frac{( \lambda (v - u) )^{n-j}}{(n-j)!}    \right]\frac{ ( \lambda u)^j}{j!} . 
\end{align*}
By distributing in to distinct sums and applying the index change $n' = n-1$ to the first sum we get: 
\begin{align*}
 \partial_{v} A(u,v) \; = &\; \sum_{n'=1}^{+ \infty} \sum_{j=1}^{n'} [1 - c_{j,X}] [1 - c_{n' +1,Z}] \lambda e^{- \lambda v}  \frac{( \lambda (v - u) )^{n'-j}}{(n'-j)!} \frac{ ( \lambda u)^j}{j!} \\
 & - \sum_{n=2}^{+ \infty} \sum_{j=1}^{n-1} [1 - c_{j,X}] [1 - c_{n,Z}] \lambda e^{- \lambda v}  \frac{( \lambda (v - u) )^{n-j}}{(n-j)!} \frac{ ( \lambda u)^{j}}{j!} . \end{align*}
 Moreover 
 $$\partial_{v} B(u,v) \; =  \; - \sum_{n=1}^{+ \infty} [1 - c_{n,X}] [1 - c_{n,Z}] \lambda e^{- \lambda v} \frac{ ( \lambda u)^{n}}{n!} .$$
Summing the second member of $ \partial_{v} A(u;v) $ and $ \partial_{v} B(u;v) $, we obtain 
\begin{eqnarray*}
 =  -  \sum_{n=1}^{+ \infty} \sum_{j=1}^{n} [1 - c_{j,X}] [1 - c_{n,Z}] \lambda e^{- \lambda v}  \frac{( \lambda (v - u) )^{n-j}}{(n-j)!} \frac{ ( \lambda u)^{j}}{j!} 
\end{eqnarray*}
Finally
\begin{align*}
\partial_{v}  \mathbb{P}[ T \leq u , C \leq v ] &  =  \partial_{v} A(u,v) + \partial_{v} B(u,v) \\
 & = \; \sum_{n=1}^{+ \infty} \sum_{j=1}^{n} [1 - c_{j,X}] [c_{n,Z} -c_{n+1,Z}] \lambda e^{- \lambda v}  \frac{( \lambda (v - u) )^{n-j}}{(n-j)!} \frac{ ( \lambda u)^{j}}{j!}  
\end{align*}
The same way as the differentiation on $v$, to differentiate on $u$ we can concatenate $ \partial_{v}  \mathbb{P}[ T \leq u , C \leq v ] $ in two parts, a first member where ${j=n}$ and another one where ${j<n}$ 
\begin{align*}
\partial_{v}  \mathbb{P}[ T \leq u , C \leq v ] \; = & \sum_{n=1}^{+ \infty} [1 - c_{n,X}] [c_{n,Z} -c_{n+1,Z}] \lambda e^{- \lambda v}  \frac{ ( \lambda u)^{n}}{n!}  \\
& + \sum_{n=2}^{+ \infty} \sum_{j=1}^{n-1} [1 - c_{j,X}] [c_{n,Z} -c_{n+1,Z}] \lambda e^{- \lambda v}  \frac{( \lambda (v - u) )^{n-j}}{(n-j)!} \frac{ ( \lambda u)^{j}}{j!} 
\end{align*}
The same method (differentiation, separation in two members, change of index) allows us to prove the wanted result for $u < v$
$$
    f_{ac}(u, v) \; = \; \sum_{n=1}^{+ \infty} \sum_{j=1}^{n} \big[ c_{j-1,X} - c_{j,X} \big] \big[ c_{n,Z} - c_{n+1,Z} \big] 
    \lambda^2 e^{ - \lambda v} \frac{ (\lambda (v - u))^{n-j} }{ (n-j)!  } \frac{( \lambda u )^{j-1}}{(j-1)!}.
$$

We are not able to use the same method (derivation of the distribution function) to determine the density along the diagonal. Indeed the law of the couple involves a part that is absolutely continuous with respect to the Lebesgue measure and a part that is singular.

\paragraph{\it Determination of the density function along the diagonal.}
The first step of the proof consists of calculating the following distribution function: $u \mapsto \mathbb{P} [ T \leq u, T = C ]$. We have
\begin{align} \nonumber
    \mathbb{P} [ T \leq u, & T = C ] \\ \nonumber
    = & \sum_{k=1}^{+ \infty} \mathbb{P} [ T \leq u, T = C, N_u = k ] \\  
    = & \sum_{n=1}^{+ \infty} \mathbb{P} [ T \leq u, T = C | N_u = k ] \mathbb{P} [N_u = k] \\  \nonumber
    = & \sum_{k=1}^{+ \infty} \mathbb{P} [ \exists n \in \{ 1 , \ldots , k \}, X_0 + \ldots + X_{n-1} < x \leq X_0 + \ldots + X_n ; \\  \nonumber
    & \qquad \qquad Z_0 + \ldots + Z_{n-1} < z \leq Z_0 + \ldots + Z_n ] \mathbb{P} [N_u = k] \\  \nonumber
    = & \sum_{k=1}^{+ \infty} \sum_{n = 1}^k [c_{n-1,X} -c_{n,X}] [c_{n-1,Z} -c_{n,Z}] \mathbb{P} [N_u = k] \\  \nonumber
    = & \sum_{n=1}^{+ \infty} \sum_{k = n}^{+ \infty} [c_{n-1,X} -c_{n,X}] [c_{n-1,Z} -c_{n,Z}] \mathbb{P} [N_u = k] \\  \nonumber
    = & \sum_{n=0}^{+ \infty} \sum_{k = n + 1} ^{+ \infty} [c_{n,X} -c_{n+1,X}] [c_{n,Z} -c_{n+1,Z}] \mathbb{P} [N_u = k] \\ \label{eq:cdf_T_equal_C}
    = & \sum_{n=0}^{+ \infty} [c_{n,X} -c_{n+1,X}] [c_{n,Z} -c_{n+1,Z}] e^{-\lambda u}  \sum_{k = n+1}^{+ \infty} \frac{(\lambda u)^k}{k!}.
\end{align}
Differentiating this function leads to the announced result (Equation \eqref{eq:density_sing_part}) and achieves the proof of Lemma \ref{lem:dens_fct_hitting_times}.
\qed

\bigskip

In order to prove Lemma \ref{lem:null_prob_induces_non_censoring}, we need the following lemma that describes the behavior of the probability terms that appear as coefficient in the density function. 
\begin{lemma}
\label{lem:equality_probab_term}
    Let $(X_n)_{n \in \mathbb{N}^*}$ be a sequence of real non-negative i.i.d. random variables. Let $X_0 := 0 \; \mathbb{P}$-as be independent of the sequence $(X_n)_{n \in \mathbb{N}^*}$ . Let $x>0$. Assume that for some $n \in \mathbb{N}$: 
    \begin{eqnarray*}
        \mathbb{P} [ X_0 + \ldots + X_n < x] = \mathbb{P} [ X_0 + \ldots + X_{n+1} < x]
    \end{eqnarray*}
Then $\mathbb{P} [ X_0 + \ldots + X_n < x] = 0$ or $\mathbb{P} [ X_0 + \ldots + X_n < x] = 1 $.
\end{lemma}

\paragraph{Proof of Lemma \ref{lem:equality_probab_term}}
Define $S_n = \sum_{k=0}^n X_k, n \in \mathbb{N}$. Assume that for some $n \in \mathbb{N}$:
\begin{eqnarray*}
    \mathbb{P} [ S_n < x] = \mathbb{P} [ S_{n+1} < x].
\end{eqnarray*}
To prove this result, we are reasoning reductio ad absurdum. Assume that $0< \mathbb{P} [ S_n < x] < 1 $. 

A first observation that the support of $X_1$ is a subset of $[0,x[$. If it is not the case, then:
$$
    \mathbb{P} [ S_{n+1} < x] \leq \mathbb{P} [ X_0 + \ldots + X_n < x] \mathbb{P} [ X_1 < x] < \mathbb{P} [ S_n < x].
$$
Let us define $M_{\max} = \sup \{ t \in Supp\{ X_1 \} \} $. In our setting, $M_{\max} > 0$ ($X_1$ is not reduced to zero) and we just proved that $M_{\max} \leq x$.

We also claim that $Supp\{S_n \} \cap [0,x[ \subset [0,x - M_{\max}[$ . If it is not the case, that is if 
\begin{equation} \label{eq:proof_lem_2_2_assump}
\mathbb P [x - M_{\max} \leq S_n < x] > 0,
\end{equation}
then
\begin{align} \nonumber
    \mathbb{P} [ S_{n+1} < x] & = \mathbb{P} [ S_{n+1} < x; S_n < x - M_{\max} ]  + \mathbb{P} [ S_{n+1} < x; S_n \geq x - M_{\max} ] \\ \label{eq:proof_lem_2_2}
    & = \mathbb{P} [ S_{n} < x - M_{\max} ] + \mathbb{P} [ S_{n+1} < x; S_n \geq x - M_{\max} ] .
\end{align}
The study of the second term must be treated by distinguishing the two following cases. 
\begin{itemize}
\item Assume that 
$\mathbb{P} [ x - M_{\max} \leq S_n < x ] > \mathbb{P} [ S_n = x - M_{\max} ]$. Then there necessarily exists $t \in ]x - M_{\max}, x[$ such that $\mathbb{P} [ x - M_{\max} \leq S_n < x ] > \mathbb{P} [ t \leq S_n < x ]$ and we deduce that: 
\begin{align*}
    \mathbb{P} [ S_{n+1} < x; S_n \geq x - M_{\max} ] & = \mathbb{P} [ S_{n+1} < x; S_n \geq t ] + \mathbb{P} [ S_{n+1} < x; x - M_{\max} \leq S_n < t ] \\
    & \leq  \mathbb{P} [  X_{n+1} < x-t; t \leq S_n < x ] + \mathbb{P} [ S_{n+1} < x; x - M_{\max} \leq S_n < t ] \\
    & = \mathbb{P} [  X_{n+1} < x-t] \mathbb{P} [  t \leq S_n < x ] + \mathbb{P} [ S_{n+1} < x; x - M_{\max} \leq S_n < t ] \\
    & < \mathbb{P} [  t \leq S_n < x ] + \mathbb{P} [ S_{n+1} < x; x - M_{\max} \leq S_n < t ] \\
    & < \mathbb{P} [  x - M_{\max} \leq S_n < x ].
\end{align*} 
We used that $\mathbb{P} [  X_{n+1} < x-t] < 1 $, since $x-t < M_{\max}$. With \eqref{eq:proof_lem_2_2}, we obtatin that $\mathbb{P} [ S_{n+1} < x] < \mathbb{P} [ S_{n} < x]$, which contradicts the assumption of this lemma. 

As a consequence we obtain that \eqref{eq:proof_lem_2_2_assump} becomes 
$$\mathbb P [x - M_{\max} \leq S_n < x] = \mathbb P [x - M_{\max} = S_n] > 0.$$

\item Assume that $\mathbb{P} [ x - M_{\max} \leq S_n < x ] = \mathbb{P} [ S_n = x - M_{\max} ] > 0 $. Define $ \text{Atom}(X_1) = \{ t \geq 0, \mathbb{P}[X_1 = t ] >0 \} $ the set of atoms of $X_1$. Remark that $\mathbb{P} [ S_n = x - M_{\max} ] \neq 0$ implies that $\text{Atom}(X_1)$ is not empty. Moreover we obtain the following equality
\begin{align*}
    \mathbb{P} [ S_{n+1} < x; S_n \geq x - M_{\max} ] & = \mathbb{P} [ X_{n+1} < M_{\max}; S_n = x - M_{\max} ] \\
    & = \mathbb{P} [ X_{n+1} < M_{\max}] \mathbb{P} [ S_n = x - M_{\max} ].
\end{align*}
If $M_{\max}$ is an atom of $X_1$, then $\mathbb{P} [ X_{n+1} < M_{\max}] < 1 $ and thus with \eqref{eq:proof_lem_2_2}, $\mathbb{P} [ S_{n+1} < x] < \mathbb{P} [ S_{n} < x]$, which is again absurd. Therefore $M_{\max}$ is not an atom of $X_1$ and for any atom $a$, $\mathbb P (X_1 = a) < \mathbb P(a \leq X_1 < M_{\max}) $. Define the set $A_n = \{ (x_1,\ldots,x_n) \in \text{Atom}(X_1)^n; x_1 + \ldots + x_n = x - M_{\max} \}$. We get
\begin{align*}
    \mathbb{P} [S_n = x - M_{\max} ] & = \sum_{(x_1,\ldots,x_n) \in A_n} \mathbb{P} [X_1 = x_1 ; \ldots ; X_n = x_n ] \\
    & < \sum_{(x_1,\ldots,x_n) \in A_n} \mathbb{P} [X_1 = x_1 ; \ldots ; x_n \leq X_n < M_{\max} ] \\
    & \leq \mathbb{P} [ x - M_{\max} \leq S_n < x ].
\end{align*}
This contradicts the fact that $\mathbb{P} [ x - M_{\max} \leq S_n < x ] = \mathbb{P} [ S_n = x - M_{\max} ]$.
\end{itemize}
We conclude that  $Supp\{S_n \} \cap [0,x[ \subset [0,x - M_{\max}[$, i.e. $ \mathbb{P} [ S_n < x ] = \mathbb{P} [ S_n < x - M_{\max} ]$. Since $0 < \mathbb{P} [ S_n < x] $, we must have $M_{\max} < x$.

Our assumption that $\mathbb{P} [ S_n < x] < 1$ implies that $n M_{\max} \geq x $. But we have
\begin{align*}
    \mathbb{P} [S_n < x - M_{\max} ] & \leq \mathbb{P} [ S_{n-1}< x - M_{\max} ] \\
    & = \mathbb{P} [ S_{n-1}< x - M_{\max} ] \mathbb{P} [ X_n \leq M_{\max}] \\
    &= \mathbb{P} [ S_{n-1}< x - M_{\max} ; X_n \leq M_{\max}] \\
    & \leq  \mathbb{P} [ S_n < x ] = \mathbb{P} [ S_n < x - M_{\max} ].
\end{align*}
Hence  
$$0<  \mathbb{P} [ S_n < x - M_{\max} ] =  \mathbb{P} [ S_{n-1}< x - M_{\max} ] < 1 .$$
Using the same ideas of the beginning of the proof, we obtain that $Supp\{S_{n-1} \} \cap [0,x -M_{\max}[ \subset [0,x - 2M_{\max}[$ and $\mathbb{P} [ S_{n-1}< x - M_{\max} ] = \mathbb{P} [ S_{n-1}< x - 2M_{\max} ]$.
Recursively we deduce that $\mathbb{P} [ S_n < x ] =  \mathbb{P} [ X_{1}< x - nM_{\max} ] = 0 $, or that $x \leq nM_{\max} < x$, which is absurd. 
The conclusion of the lemma follows:
\begin{eqnarray*}
    \mathbb{P} [ S_n < x] = \mathbb{P} [ S_{n+1} < x] = 0 \quad \text{or} \quad 1.
\end{eqnarray*}
\qed

\paragraph{Proof of Lemma \ref{lem:null_prob_induces_non_censoring}}
Formula \eqref{eq:prob_T_equal_C}
\begin{eqnarray*}
    \mathbb{P} \left[ T = C   \right] = \underset{n \in \mathbb{{N}}}{\sum} [c_{n,X} - c_{(n+1),X}] [c_{n,Z} - c_{(n+1),Z} ] 
\end{eqnarray*}
is obtained by integrating the density function \eqref{eq:density_sing_part} in Lemma \ref{lem:dens_fct_hitting_times} along the diagonal $\{ (u,v),u,v\geq 0, u=v \}$ or computing the limit at infinity in \eqref{eq:cdf_T_equal_C} (note that the function $u \mapsto \sum_{k = n+1}^{+ \infty} e^{-\lambda u} \frac{(\lambda u)^k}{k!}$ is the cumulative distribution function of the Erlang law with parameters $(\lambda,n+1)$). 

\medskip 

Assume that $T < C$ almost surely. The case $T>c$ almost surely is treated the same way. We naturally have $\mathbb{P} \left[ T = C   \right] = 0$.
Assume that $\mathbb{P} \left[ T = C   \right] = 0$. Since all the terms of the series
\begin{eqnarray*}
    \underset{n \in \mathbb{{N}}}{\sum} [c_{n,X}- c_{n+1,X}] [c_{n,Z} - c_{n+1,Z} ] 
\end{eqnarray*}
are non-negative, we necessarily have for any $n \in \mathbb{N}$:
\begin{align*}
    & c_{n,X} = c_{n+1,X} \quad \quad \text{or} \quad \quad c_{n,Z} = c_{n+1,Z}
\end{align*}
This situation can be separated in three different cases:
\begin{itemize}
    \item {\it Case 1}. If for all $n \in \mathbb{N}$, $\mathbb{P} [ X_0 + \ldots + X_n < x] = \mathbb{P} [ X_0 + \ldots + X_{n+1} < x]$, then $\mathbb{P} [ X_0 + \ldots + X_n < x] = \mathbb{P} [ X_0 < x] = 1, \forall n \in \mathbb{N}$ and since $X_1 \geq 0$ a.s. we must have $X_1 = 0$ a.s. which is not possible by assumption made at the beginning of Section \ref{sect:model_setting}. 
The same happens if  for all $n \in \mathbb{N}$, $\mathbb{P} [ Z_0 + \ldots + Z_n < z] = \mathbb{P} [ Z_0 + \ldots + Z_{n+1} < z]$. 
\item {\it Case 2}. Suppose that there exists $n_1 < n_2$ with $n_1,n_2 \in \mathbb{N}$ such that $c_{n_1,X} \neq c_{n_1+1,X}$ and $c_{n_2} = c_{n_2+1,X}$. From Lemma \ref{lem:equality_probab_term} we deduce that $c_{n_2,X} = 0$ or $c_{n_2,X} = 1$. But in the second case, we would have for any $n \leq n_2$, $c_{n,X} = 1$ and in particular $c_{n_1,X} = c_{n_1+1,X}$. Therefore $c_{n_2,X} = 0$ and thus for any $n \geq n_2$, $c_{n,X}  = 0$ and $X_1 \geq \frac{x}{n_2}$ a.s. 

Notice that since $c_{n_1,X}  \neq c_{n_1+1,X}$, necessarily $c_{n_1+1} < 1$ and thus for any $n > n_1$, $c_{n,X}< 1$.

Now we define the set $J=\{n \in \mathbb{N}, c_{n,X}\neq c_{n+1,X}\}$. This set is non empty ($n_1 \in J$) and bounded (by $n_2$), hence $J$ admits a maximum $N_{\max}$ and a minimum $N_{\min}$. Again from Lemma \ref{lem:equality_probab_term}, we deduce that $J = \{ N_{\min},\ldots,N_{\max}\}$. Remark that the hitting time $T$ occurs almost surely at or after the $N_{\min}$-th jump of the Poisson process and almost surely at or before the $N_{\max}+1$-th jump of the Poisson process.

%Moreover, since $\mathbb{P} [ X_0 + \ldots + X_{n_2} < x] = \mathbb{P} [ X_0 + \ldots + X_{n_2+1} < x ]$, and thus for any $n \geq n_2$ $\mathbb{P} [ X_0 + \ldots + X_{n} < x] = 0$. From this remark, it naturally comes that there exists a subset $J \subset \mathbb{N}$ such that $J$ admits a maximum $N_{\max}$ and a minimum $N_{\min}$ and for any $n \in \mathbb{N}$ such that $N_{\min} \leq n \leq N_{\max}$, necessarily $n \in J$ and $J$ is such that $\mathbb{P} [ X_0 + \ldots + X_n < x] \neq \mathbb{P} [ X_0 + \ldots + X_{n+1} < x]$ if and only if $n \in J$. 

In this configuration, for any $n \in J$, $c_{n,Z}= c_{n+1,Z}$. From Lemma \ref{lem:equality_probab_term} we deduce that 
$c_{N_{\min},Z} = 0 $ or 1. 
In the first case, 
$Z_1 \geq \frac{z}{N_{\min}}$ and thus the hitting time $C$ occurs almost surely strictly before the  $N_{\min}$-th jump of the Poisson process, that is 
%Finally we have that 
$C < T$ almost surely. If $c_{N_{\min},Z} = 1$, then $c_{N_{\max},Z} = 1$ and a.s. $C$ occurs after the $N_{\max}+1$-th jump of the Poisson process, that $C\geq T$ almost surely. 

\item {\it Case 3.} If for any $N \in \mathbb{N}$, there exists $n \geq N$, $c_{n,X} \neq c_{n+1,X}$, using the same type of arguments as in Case 2, we obtain that $C < T$ almost surely.
\end{itemize}
Lemma \ref{lem:null_prob_induces_non_censoring} is proved. 
\qed

\paragraph{Proof of Lemma \ref{lem:dens_fct_censor}}
The proof for Model II is immediately deduced from Model I proof. Consequently we only prove the result for the Model I. The first step to determine the density function of the censoring couple is to determine the cumulative distribution function $F_{(Y, \Delta)}$. \\
Let us determine $F_{(Y, \Delta)} (t, 1) $ for $t \geq 0$. We have:
\begin{align*}
   F_{(Y, \Delta)} (t, 1) = \mathbb{P} [ T \leq t , T \leq C  ] = \mathbb{P} [ T \leq t] - \mathbb{P} [ T \leq t , T > C  ] .
\end{align*}
By similar arguments as in the proof of Lemma \ref{lem:dens_fct_hitting_times}, we obtain that
\begin{eqnarray*}
    \mathbb{P} [ T \leq t  ] = \sum_{n=1}^{+ \infty} [ 1 - c_{n,X}] e^{-\lambda t} \frac{(\lambda t)^n}{n!}. 
\end{eqnarray*}
Then using the density function $f_{ac}$ given by \eqref{eq:density_abs_cont_part}
\begin{align*}
    & \mathbb{P} [ T \leq t , T > C  ] = \int_{\mathbb{R}^2} \mathbb{1}_{ \{ 0 \leq v < u \leq t \} } f_{ac} (u,v) du dv   \\
    & \quad = \int_0^t \int_0^{u} \sum_{n=1}^{+ \infty} \sum_{j=1}^{n} [ c_{n,X} - c_{n+1,X} ] [ c_{j-1,Z} - c_{j,Z} ] \lambda^2 e^{ - \lambda u} \frac{ (\lambda (u - v))^{n-j} }{ (n-j)!  } \frac{( \lambda v )^{j-1}}{(j-1)!}  du dv 
\end{align*}
Computing this integral immediately gives 
\begin{align*}
    F_{(Y, \Delta)} (t, 1) = & \sum_{n=1}^{+ \infty} [ 1 - c_{n,X}] e^{-\lambda t} \frac{(\lambda t)^n}{n!} 
    - \sum_{n=1}^{+ \infty} [c_{n,X} - c_{n+1,X} ] [ 1 - c_{n,Z}] \left[ 1 - \sum_{k=0}^n e^{-\lambda t} \frac{(\lambda t)^k}{k!} \right] .
\end{align*}
The determination of $F_{(Y, \Delta)} (t, 0) $ is done the same way
\begin{align*}
    F_{(Y, \Delta)} (t, 0) = \sum_{n=1}^{+ \infty} [ 1 - c_{n,Z}] e^{-\lambda t} \frac{(\lambda t)^n}{n!} 
    - \sum_{n=0}^{+ \infty} [c_{n,Z} - c_{n+1,Z}] [ 1 - c_{n+1,X}] \left[ 1 - \sum_{k=0}^n e^{-\lambda t} \frac{(\lambda t)^k}{k!} \right] .
\end{align*}
By deifferentiating, we obtain the density function and conclude the proof of this lemma.
\qed

\paragraph{Proof of Lemma \ref{lem:tech_resulu}.}
The statement of the lemma implicitly assumes that the function $t \mapsto \sum_{n=0}^{\infty} a_n t^n$ is well defined on $\mathbb{R}_+$ and that $a_n > 0$ for any $n$. 

Now suppose that $\lim_{n \xrightarrow{} + \infty} \dfrac{a_{n+1}}{a_{n}} = a$.
Then for all $\varepsilon > 0 $, there exists $N \in \mathbb{N}$ such that $\forall n \geq N$, $ \left| \dfrac{a_{n+1}}{a_n} - a \right| \leq \dfrac{\varepsilon}{2}$. Then for all $t > 0 $,
\begin{align*}
    | g(t) - a | = & \left| \frac{\sum_{n=0}^{+ \infty} a_{n+1}  t^n}{ \sum_{n=0}^{+ \infty} a_{n} t^n } - a \frac{\sum_{n=0}^{+ \infty} a_{n} t^n}{ \sum_{n=0}^{+ \infty} a_{n} t^n }  \right| \\
    = & \left| \frac{\sum_{n=0}^{N-1} a_{n+1} t^n}{ \sum_{n=0}^{+ \infty} a_{n} t^n } + \frac{\sum_{n=N}^{+ \infty} a_{n+1} t^n}{ \sum_{n=0}^{+ \infty} a_{n} t^n } -  a \frac{\sum_{n=0}^{N-1} a_{n} t^n}{ \sum_{n=0}^{+ \infty} a_{n} t^n } - a \frac{\sum_{n=N}^{+ \infty} a_{n} t^n}{ \sum_{n=0}^{+ \infty} a_{n} t^n } \right| \\
    \leq & \left| \frac{\sum_{n=0}^{N-1} a_{n+1} t^n}{ \sum_{n=0}^{+ \infty} a_{n} t^n } - a \frac{\sum_{n=0}^{N-1} a_{n} t^n}{ \sum_{n=0}^{+ \infty} a_{n} t^n } \right| + \left|  \frac{\sum_{n=N}^{+ \infty} a_{n+1} t^n}{ \sum_{n=0}^{+ \infty} a_{n} t^n }  - a \frac{\sum_{n=N}^{+ \infty} a_{n} t^n}{ \sum_{n=0}^{+ \infty} a_{n} t^n }  \right| 
    \\
    = & \left| \frac{\sum_{n=0}^{N-1} (a_{n+1} -a a_n) t^n}{ \sum_{n=0}^{+ \infty} a_{n} t^n } \right| + \left|  \frac{\sum_{n=N}^{+ \infty} \left( \frac{a_{n+1}}{a_n} - a \right) a_n t^n}{ \sum_{n=0}^{+ \infty} a_{n} t^n } \right| \\
    \leq & \frac{\sum_{n=0}^{N-1} | a_{n+1} -a a_n | t^n}{ \sum_{n=0}^{N} a_{n} t^n } + \frac{\varepsilon}{2} \frac{\sum_{n=N}^{+ \infty} a_n t^n}{ \sum_{n=0}^{+ \infty} a_{n} t^n } \\
    \leq & \frac{\varepsilon}{2} + \frac{\varepsilon}{2}.
\end{align*}
where the last inequality comes from the fact that the function is the quotient of a polynomial function of degree $N-1$ and of a polynomial function of degree $N$, so that
\begin{eqnarray*}
    \lim_{t \xrightarrow{} + \infty} \frac{\sum_{n=0}^{N-1} |a_{n+1} - a a_n| t^n}{ \sum_{n=0}^{N} a_{n} t^n } = 0.
\end{eqnarray*}
Finally $\lim_{t \xrightarrow{} + \infty} g(t) = a $.
\qed

\paragraph{Proof of Theorem \ref{theo:behaviour_c_(n+1)/c_n}}
 \underline{Assume that $(X_n)_{n \in \mathbb{N}^*} \in F_2$}. Denote $\phi$ the density of $X_1$ and $\phi_{S_n}$ the density function of $S_n = \sum_{k=1}^{n} X_k$, $n \in \mathbb{N}^*$. Let $n \in \mathbb{N}$. Firstly we have
\begin{equation}\label{eq:proof_thm_4_6}
    \frac{\mathbb{P}[ S_{n+1} < x  ] }{ \mathbb{P}[ S_{n} < x  ] } = \int_0^x \phi(t) \frac{\mathbb{P}[ S_{n} < x - t  ]}{\mathbb{P}[ S_{n} < x  ]} dt.
\end{equation}
We can study the quotient $\dfrac{\mathbb{P}[ S_{n} < x - t  ]}{\mathbb{P}[ S_{n} < x  ]}$. 
Consider $N_0$ be the minimal integer such that $\phi_{S_N}$ is non-decreasing on $[0,x]$.
Let $n \in \mathbb{N}$, such that $n \geq N_0$. Set up the Euclidean division of $n$ by $N_0$: $n = qN_0 + r$, where $q>0$ and $0 \leq r < N_0$. Then 
\begin{align*}
& \mathbb{P} [ X_1 + \ldots + X_n < x - t ] \\
 &\quad  =  \mathbb{P} [ X_1 + \ldots + X_{N_0} + \ldots + X_{N_0 (q-1) +1} + \ldots + X_{N_0 q} + X_{N_0 q+1} + \ldots + X_{N_0 q+r}  < x - t ] 
\end{align*}
Since the sequence $(X_i)_{i \in \mathbb{N}^*}$ is i.i.d., the variables $(X_1 + \ldots + X_{N_0})$, ... , $(X_{N_0 (q-1) +1} + \ldots + X_{N_0q })$, $X_{N_0 q+1}$, ... , $X_{N_0 q+r}$ are independent. Then:
\begin{align*} 
    & \mathbb{P} [ X_1 + \ldots + X_{N_0} + \ldots + X_{N_0 (q-1) +1} + \ldots + X_{N_0 q} + X_{N_0 q+1} + \ldots + X_{N_0 q+r}  < x - t ] \\
    &\quad = \int_{(\mathbb{R}+)^{q+r}} \phi_{S_{N_0}}(x_1) \ldots \phi_{S_{N_0}}(x_q) \phi(x_{q+1}) \ldots \phi(x_{q+r}) \mathbb{1}_{ \{ 0 \leq x_1 + \ldots + x_{q+r} < x-t \} } dx_1 \ldots dx_{q+r} \\
    &\quad = \int_{(\mathbb{R}+)^{q+r}} \phi_{S_{N_0}}(x_1) \ldots \phi_{S_{N_0}}(x_q) \phi(x_{q+1}) \ldots \phi(x_{q+r}) \mathbb{1}_{ \{ 0 \leq \frac{x}{x-t} x_1 + \ldots + \frac{x}{x-t} x_{q+r} < x \} } dx_1 \ldots dx_{q+r}.
\end{align*}
By the change of variables $ y_i = \frac{x}{x-t} x_i, i \in \{0, ..., q \}$ et $ y_i = x_i, i \in \{q+1, ..., q+r \}$, we get: 
\begin{align*}
    &  \int_{(\mathbb{R}+)^{q+r}} \phi_{S_{N_0}}(x_1) \ldots \phi_{S_{N_0}}(x_q) \phi(x_{q+1}) \ldots \phi(x_{q+r}) \mathbb{1}_{ \{ 0 \leq \frac{x}{x-t} x_1 + \ldots + \frac{x}{x-t} x_{q+r} < x \} } dx_1 \ldots dx_{q+r} \\
    & \quad = \left( \frac{x-t}{x} \right)^q \int_{(\mathbb{R}+)^{q+r}} \phi_{S_{N_0}} \left( \frac{x-t}{x} y_1 \right) \ldots \phi_{S_{N_0}} \left( \frac{x-t}{x} y_q \right) \phi ( y_{q+1} ) \ldots \phi(y_{q+r}) \\
     & \hspace{3cm}  \times \mathbb{1}_{ \{ 0 \leq y_1 + \ldots + y_q + \frac{x}{x-t} y_{q+1} + \ldots + \frac{x}{x-t} y_{q+r} < x \} } dy_1 \ldots dy_{q+r} .
\end{align*}
The condition $ \{ 0 \leq y_1 + \ldots + y_q + \frac{x}{x-t} y_{q+1} + \ldots + \frac{x}{x-t} y_{q+r} < x \} $ implies that for any $i$, $y_i \in [0,x]$. The non-decreasing property of $\phi_{S_{N_0}}$ on $[0,x]$ implies: 
\begin{align*}
    & \left( \frac{x-t}{x} \right)^q \int_{(\mathbb{R}+)^{q+r}} \phi_{S_{N_0}} \left( \frac{x-t}{x} y_1 \right) \ldots \phi_{S_{N_0}} \left( \frac{x-t}{x} y_q \right) \phi ( y_{q+1} ) \ldots \phi(y_{q+r}) \\
     & \hspace{3cm}  \times \mathbb{1}_{ \{ 0 \leq y_1 + \ldots + y_q + \frac{x}{x-t} y_{q+1} + \ldots + \frac{x}{x-t} y_{q+r} < x \} } dy_1 \ldots dy_{q+r} \\
   &    \leq \left( \frac{x-t}{x} \right)^q \int_{(\mathbb{R}+)^{q+r}} \phi_{S_{N_0}} (  y_1 ) \ldots \phi_{S_{N_0}} (  y_q ) \phi ( y_{q+1} ) \ldots \phi(y_{q+r}) \\
    & \hspace{3cm}   \times \mathbb{1}_{ \{ 0 \leq y_1 + \ldots + y_q + \frac{x}{x-t} y_{q+1} + \ldots + \frac{x}{x-t} y_{q+r} < x \} } dy_1 \ldots dy_{q+r}.
\end{align*}
Finally, $ \{ 0 \leq y_1 + \ldots + y_q + \frac{x}{x-t} y_{q+1} + \ldots + \frac{x}{x-t} y_{q+r} < x \} \subset \{ 0 \leq y_1 + \ldots + y_{q+r} < x \} $, since $\frac{x}{x-t} > 1 $.
So,
\begin{align*}
   &  \left( \frac{x-t}{x} \right)^q \int_{(\mathbb{R}+)^{q+r}} \phi_{S_{N_0}} (  y_1 ) \ldots \phi_{S_{N_0}} (  y_q ) \phi( y_{q+1} ) \ldots f(y_{q+r}) \\
   & \hspace{3cm}   \times \mathbb{1}_{ \{ 0 \leq y_1 + \ldots + y_q + \frac{x}{x-t} y_{q+1} + \ldots + \frac{x}{x-t} y_{q+r} < x \} } dy_1 \ldots dy_{q+r} \\
    &  \quad \leq \left( \frac{x-t}{x} \right)^q \int_{(\mathbb{R}+)^{q+r}} \phi_{S_{N_0}} (  y_1 ) \ldots \phi_{S_{N_0}} (  y_q ) \phi ( y_{q+1} ) \ldots \phi(y_{q+r}) \mathbb{1}_{ \{ 0 \leq y_1 + \ldots + y_{q+r} < x \} } dy_1 \ldots dy_{q+r} \\
     & \quad\leq \left( \frac{x-t}{x} \right)^q \mathbb{P} [ X_1 + \ldots + X_{Nq+r} < x] = \left( \frac{x-t}{x} \right)^q \mathbb{P} [ X_1 + \ldots + X_n < x]  .
\end{align*}
To conclude, we obtain
$$
    \frac{\mathbb{P} [ X_1 + \ldots + X_n < x - t ]}{\mathbb{P} [ X_1 + \ldots + X_n < x]} \leq \left( \frac{x-t}{x} \right)^q .% = \left( \frac{x-t}{x} \right)^{ -1 } \left( \frac{x-t}{x} \right)^{ q+1 } = \left( \frac{x-t}{x} \right)^{ -1 } \left( \frac{x-t}{x} \right)^{ \frac{1}{N_0}  N_0(q+1) }
$$
But $ q = \frac{n-r}{N_0} > \frac{n-N_0}{N_0} \geq 0$, and since $0 \leq \frac{x-t}{x} \leq 1 $, we finally get:
\begin{eqnarray*}
    \frac{\mathbb{P} [ X_1 + \ldots + X_n < x - t ]}{\mathbb{P} [ X_1 + \ldots + X_n < x]} \leq 
    \left( \frac{x-t}{x} \right)^{q} \leq  \left( \frac{x-t}{x} \right)^{ \frac{n}{N_0}-1 }.
\end{eqnarray*}
By applying the dominated convergence theorem inside equation \eqref{eq:proof_thm_4_6}, we obtain the desired result. \\
\underline{Assume that $(X_n)_{n \in \mathbb{N}^*} \in F_3$}. Let $x>0$ and $Atom^* = \{ t > 0, \mathbb{P}[X_1 = t ] > 0 \}  = \{x_i, \ i \in \mathbb N\}$ be the set of positive atoms of $X_1$. Define for any $t>0$
\begin{eqnarray*}
x_{\min} & =&  \inf \{ x_i  \in Atom^* \},\quad M_{\max,t} = \left\lceil \frac{t}{x_{\min}} \right\rceil ,\quad \Xi_t = \{ x_i \in Atom^* ; x_i < t \}  \\
\Diamond_t & = & \left\{ (x_{i_1}, \dots ,x_{i_k}) \in (Atom^*)^k; \; k \leq M_{\max,t} \; \text{and} \; \sum_{j=1}^k x_{i_j} < t \right\} .
\end{eqnarray*}
To simplify the formulas, let us re-write for $n \geq M_{\max,t}$, $k \leq M_{\max,t}$ and $ (x_{i_1}, \dots ,x_{i_k}) \in \Diamond_t$:
\begin{eqnarray*}
    B_{n,t,(x_{i_1}, \dots ,x_{i_k})} = \{ \exists j_1<\ldots<j_k \leq n; X_{j_1} = x_{i_1}, \ldots ,X_{j_k} = x_{i_k} \} \cap \{ \forall i \leq n+1;\forall h \leq k; i \neq j_h; X_i = 0 \}.
\end{eqnarray*}
The sets $\{ B_{n,t,(x_{i_1}, \dots ,x_{i_k})}, (x_{i_1}, \dots ,x_{i_k}) \in \Diamond_t \}$ are defined to constitute a partition of the event $\{ S_n < t \}$.  
Then consider $n \in \mathbb{N}^*$, $n \geq M_{\max,x}$. We obtain 
\begin{align} \nonumber
    \mathbb{P} [ S_{n+1} < x] = & \mathbb{P} [ S_{n+1} < x; X_{n+1} = 0 ] + \sum_{x_i \in \Xi_x} \mathbb{P} [ S_{n+1} < x; X_{n+1} = x_i] \\ \nonumber
    = & \mathbb{P} [X_1 = 0] \mathbb{P} [ S_{n} < x ] + \sum_{x_i \in \Xi_x} \mathbb{P} [X_1 = x_i] \mathbb{P} [ S_{n} < x - x_i] \\ \nonumber
    \leq & \mathbb{P} [X_1 = 0] \mathbb{P} [ S_{n} < x ] + \mathbb{P} [ S_{n} < x - x_{\min}] \sum_{x_i \in \Xi_x} \mathbb{P} [X_1 = x_i] \\ \label{eq:proof_atomic_1}
    \leq & \mathbb{P} [X_1 = 0] \mathbb{P} [ S_{n} < x ] + \mathbb{P} [ S_{n} < x - x_{\min}] . 
\end{align}
We prove that there exists $C_x>0$ such that
\begin{eqnarray*}
    \frac{\mathbb{P} [ S_{n} <  x - x_{\min}]}{\mathbb{P} [ S_{n} < x]} \leq \frac{C_x}{n} .
\end{eqnarray*}
Firstly we have %\red{\st{note that $M_{\max,x - x_{\min}} = M_{\max,x} - 1$} : on ne s'en sert pas ici}
\begin{align*}
    \mathbb{P} [ S_{n} < x - x_{\min}] = &  \sum_{k=0}^{M_{\max,x-x_{\min}}} \sum_{(x_{i_1}, \dots ,x_{i_k}) \in \Diamond_{x - x_{\min}}} \mathbb{P} [ S_{n} < x - x_{\min}; B_{n,x - x_{\min},(x_{i_1}, \dots ,x_{i_k})}] \\
    = & \sum_{k=0}^{M_{\max,x-x_{\min}}} \sum_{(x_{i_1}, \dots ,x_{i_k}) \in \Diamond_{x - x_{\min}}} \mathbb{P} [ B_{n,x - x_{\min},(x_{i_1}, \dots ,x_{i_k})}] \\
    = & \sum_{k=0}^{M_{\max,x-x_{\min}}} \sum_{(x_{i_1}, \dots ,x_{i_k}) \in \Diamond_{x - x_{\min}}} \frac{\mathbb{P} [ B_{n,x - x_{\min},(x_{i_1}, \dots ,x_{i_k})}]}{\mathbb{P} [ B_{n,x,(x_{i_1}, \dots ,x_{i_k}, x_{\min})}]} \mathbb{P} [ B_{n,x,(x_{i_1}, \dots ,x_{i_k}, x_{\min})} ].
\end{align*}
Note that
\begin{align*}
    \mathbb{P} [ B_{n,x - x_{\min},(x_{i_1}, \dots ,x_{i_k})}] = & \sum_{1 \leq j_1 < \dots < j_k \leq n} \mathbb{P} [ X_{j_1} = x_{i_1}, \ldots ,X_{j_k} = x_{i_k} ; \forall i \leq n+1;\forall h \leq k; i \neq j_h; X_i = 0 ] \\
    = & \sum_{1 \leq j_1 < \dots < j_k \leq n} \mathbb{P} [X_1 = 0]^{n-k} \prod_{j=1}^k \mathbb{P} [X_1 = x_{i_j}] \\
    = & \binom{n}{k} \mathbb{P} [X_1 = 0]^{n-k} \prod_{j=1}^k \mathbb{P} [X_1 = x_{i_j}].
\end{align*}
Moreover
\begin{align*}
    & \mathbb{P} [ B_{n,x,(x_{i_1}, \dots ,x_{i_k}, x_{\min})}] \\
    & =  \sum_{1 \leq j_1 < \dots < j_{k+1} \leq n} \mathbb{P} [ X_{j_1} = x_{i_1}, \ldots ,X_{j_k} = x_{i_k}, X_{j_{k+1}} = x_{\min} ; \forall i \leq n+1;\forall h \leq k; i \neq j_h; X_i = 0 ] \\
    & =  \sum_{1 \leq j_1 < \dots < j_{k+1} \leq n} \mathbb{P} [X_1 = 0]^{n-k-1} \mathbb{P} [X_1 = x_{\min}] \prod_{j=1}^k \mathbb{P} [X_1 = x_{i_j}] \\
    & =  \binom{n}{k+1} \mathbb{P} [X_1 = 0]^{n-k-1} \mathbb{P} [X_1 = x_{\min}] \prod_{j=1}^k \mathbb{P} [X_1 = x_{i_j} ].
\end{align*}
Therefore
\begin{eqnarray*}
    \frac{\mathbb{P} [ B_{n,x - x_{\min},(x_{i_1}, \dots ,x_{i_k})}]}{\mathbb{P} [ B_{n,x,(x_{i_1}, \dots ,x_{i_k}, x_{\min})}]} = \frac{k+1}{n-k} \frac{\mathbb{P} [X_1 = 0]}{\mathbb{P} [X_1 = x_{\min}]} \leq \frac{ M_{\max,x-x_{\min}} +1 }{n- M_{\max,x-x_{\min}}} \frac{\mathbb{P} [X_1 = 0]}{\mathbb{P} [X_1 = x_{\min}]}.
\end{eqnarray*}
To conclude we obtain
\begin{align*}
    \mathbb{P} [ S_{n} < x - x_{\min}] \leq & \frac{ M_{\max,x-x_{\min}} +1 }{n- M_{\max,x-x_{\min}}} \frac{\mathbb{P} [X_1 = 0]}{\mathbb{P} [X_1 = x_{\min}]} \sum_{k=0}^{M_{\max,x-x_{\min}}} \sum_{(x_{i_1}, \dots ,x_{i_k}) \in \Diamond_{x - x_{\min}}}  \mathbb{P} [ B_{n,x,(x_{i_1}, \dots ,x_{i_k}, x_{\min})}] \\
    \leq & \frac{ M_{\max,x-x_{\min}} +1 }{n- M_{\max,x-x_{\min}}} \frac{\mathbb{P} [X_1 = 0]}{\mathbb{P} [X_1 = x_{\min}]} \mathbb{P} [ S_{n} < x] \\
    \leq & \frac{  (M_{\max,x-x_{\min}} +1)^2 }{n} \frac{\mathbb{P} [X_1 = 0]}{\mathbb{P} [X_1 = x_{\min}]} \mathbb{P} [ S_{n} < x]
\end{align*}
since $ n \geq M_{\max,x} = M_{\max,x-x_{\min}} + 1$. We only have to consider $C_x =  M_{\max,x-x_{\min}} (M_{\max,x-x_{\min}} + 1) \frac{\mathbb{P} [X_1 = 0]}{\mathbb{P} [X_1 = x_{\min}]}$. Dividing \eqref{eq:proof_atomic_1} by $\mathbb{P} [S_n < x ]$ we finally get the announced result
\begin{eqnarray*}
    \left| \frac{\mathbb{P} [ S_{n+1} < x] }{ \mathbb{P} [ S_{n} < x] } - \mathbb{P} [X_1 = 0] \right| \leq \frac{C_x}{n}.
\end{eqnarray*}
\qed

\paragraph{Proof of Theorem \ref{theo:identifiability}}
Given $\theta , \theta' \in \Theta$, assume that:
\begin{eqnarray*}
    f_{(Y,\Delta),\theta} (t, \delta) = f_{(Y,\Delta),\theta'} (t, \delta), \qquad \forall t \geq 0, \forall \delta \in \{0;1\}.
\end{eqnarray*}
\medskip 
The remainder of the proof is given as follows: we use the asymptotic behaviour of the hazard function to prove the intensity parameter $\lambda$. We then use the previous identification to deduce that $I_n$ is verified. 
\noindent {\bf Step 1.} The first step consists of proving that $\lambda = \lambda'$. 

\noindent Under {\bf (H1)}, there exists a positive constant $\kappa_1 >0$ such that $X_1$ or $Z_1$ is greater than $\kappa_{1}$ almost surely. A priori $\kappa_1=\kappa_1(\alpha,\beta)$ depends on the parameter $\alpha$ or $\beta$.  %, where $c_{x,1} = \sup\{ c \geq 0 | \; \mathbb{P}[X_1 \geq c] = 1 \} >0 $ (resp. $c_{z,1} = \sup\{ c \geq 0 | \; \mathbb{P}[X_1 \geq c] = 1 \} >0 $).  

Now if we define $N_x(\alpha) = \inf \{ n \in \mathbb{N}^* \; | \; \; X_1 + \ldots + X_n \geq x, \; \; \text{a.s.} \}$ and $N_z(\beta) = \inf \{ n \in \mathbb{N}^* \; | \; \; Z_1 + \ldots + Z_n \geq z, \; \; \text{a.s.} \}$, the quantity $N_{\min}(\alpha,\beta)  = \min (N_x(\alpha),N_z(\beta))$ is finite, whatever the values of the parameters $\alpha$ and $\beta$. Although $\kappa_1(\alpha,\beta)$ and $\kappa_1(\alpha',\beta')$ could different, we will see that necessarily $N_{\min}(\alpha,\beta) =N_{\min}(\alpha',\beta')$. From Lemma \ref{lem:dens_fct_censor} we can re-write the density function as
 \begin{align*}
     f_{(Y, \Delta),\theta} (t, 1) = & \sum_{n=0}^{N_{\min}(\alpha,\beta) - 1 }  \big[ c_{n,X}(\alpha) - c_{n+1,X}(\alpha)  \big] c_{n,Z}(\beta)  \lambda e^{- \lambda t} \frac{(\lambda t)^n}{n!}.% \\
%     f_{(Y, \Delta),\theta} (t, 0) = & \sum_{n=0}^{N_{\min}(\alpha,\beta) - 1} \big[ c_{n,Z}(\beta) - c_{n+1,Z}(\beta) \big]c_{n+1,X}(\alpha) \lambda e^{- \lambda t} \frac{(\lambda t)^n}{n!}. 
\end{align*}
%Let $N_{\text{min}} = \min \{N_x, N_z \}$, and $N_{\text{min}}' = \min \{ N_x', N_z' \}$. 
%Since the support does not depend on the parameter $\alpha$ (resp. $\beta$),  both terms are finite, even if in each definition of the minimum, one of the two values could by $+ \infty$. 
Then for all $t$, $f_{(Y,\Delta),\theta} (t, 1) = f_{(Y,\Delta),\theta'} (t, 1) $ if and only if 
\begin{align*}
%     f_{(Y,\Delta),\theta} (t, 1) & =&  f_{(Y,\Delta),\theta'} (t, 1) \\
     %\iff 
  &    \sum_{n=0}^{N_{\min}(\alpha,\beta) -1}  [c_{n,X}(\alpha) - c_{n+1,X}(\alpha) ] c_{n,Z}(\beta) \lambda e^{- \lambda t} \frac{(\lambda t)^n}{n!} \\
  & \quad =  \sum_{n=0}^{N_{\min}(\alpha',\beta') -1 }  [c_{n,X}(\alpha') - c_{n+1,X}(\alpha') ] c_{n,Z}(\beta') \lambda' e^{- \lambda' t} \frac{(\lambda' t)^n}{n!}, 
\end{align*}
which implies for all $t \geq 0$
\begin{align*}
\dfrac{\lambda'}{\lambda}e^{ (\lambda' - \lambda)t  } & =   \frac{\sum_{n=0}^{N_{\min}(\alpha,\beta) -1}  [c_{n,X}(\alpha) - c_{n+1,X}(\alpha) ] c_{n,Z}(\beta)  \frac{(\lambda t)^n}{n!}}{\sum_{n=0}^{N_{\min}(\alpha',\beta') -1 }  [c_{n,X}(\alpha') - c_{n+1,X}(\alpha') ] c_{n,Z}(\beta')  \frac{(\lambda' t)^n}{n!}} \\
\Rightarrow{} \dfrac{\lambda'}{\lambda}e^{ (\lambda' - \lambda)t  } & \underset{t \xrightarrow{} + \infty}{ \sim } \frac{  [c_{N_{\min}(\alpha,\beta) -1,X}(\alpha) - c_{N_{\min}(\alpha,\beta),X}(\alpha) ] c_{N_{\min}(\alpha,\beta) -1,Z}(\beta)  \frac{(\lambda t)^{N_{\min}(\alpha,\beta) -1}}{(N_{\min}(\alpha,\beta) -1)!}}{  [c_{N_{\min}(\alpha',\beta') -1,X}(\alpha') - c_{N_{\min}(\alpha',\beta'),X}(\alpha') ] c_{N_{\min}(\alpha',\beta') -1,Z}(\beta')  \frac{(\lambda' t)^{N_{\min}(\alpha',\beta') -1}}{(N_{\min}(\alpha',\beta') -1)!}} \\
& \underset{t \xrightarrow{} + \infty}{ \sim } C t^{N_{\min}(\alpha,\beta) - N_{\min}(\alpha',\beta')}
\end{align*}
for some constant $C>0$. This equivalence at infinity is only possible if $\lambda = \lambda'$ and $N_{min}(\alpha,\beta)=N_{\min}(\alpha',\beta')$. 

\medskip
\noindent Under Condition {\bf (H2.i)}, from Theorem \ref{theo:behaviour_c_(n+1)/c_n}, we deduce that for any $(\alpha, \beta) \in \Theta_1 \times \Theta_2$:
\begin{eqnarray*}
    \lim_{n \xrightarrow{} + \infty }  \frac{\mathbb{P}_{\alpha} [ X_1 + \ldots + X_{n+1} < x] \mathbb{P}_{\beta} [ Z_1 + \ldots + Z_{n+1} < z]}{ \mathbb{P}_{\alpha} [ X_1 + \ldots + X_{n} < x] \mathbb{P}_{\beta} [ Z_1 + \ldots + Z_{n} < z]} = 0
\end{eqnarray*}
Since 
\begin{eqnarray*}
    f_{(Y,\Delta),(\lambda,\alpha,\beta)} (t, \delta) = f_{(Y,\Delta),(\lambda',\alpha',\beta')} (t, \delta), \quad \quad \quad \quad \forall t \geq 0, \forall \delta \in \{0;1\}
\end{eqnarray*}
we have for all $t \geq 0$
$ h_{Y,(\lambda,\alpha,\beta)} (t) = h_{Y,(\lambda',\alpha',\beta')} (t)$. Thus Lemma \ref{lem:tech_resulu} ensures that the hazard function $h_{Y,\theta}$ given by \eqref{eq:hazard-fct} converges to $\lambda$ as $t$ tends to infinity. We deduce that $\lambda = \lambda'$.  

\medskip
\noindent If Condition {\bf (H2.ii)} holds, from Theorem \ref{theo:behaviour_c_(n+1)/c_n}, we deduce that for any $(\alpha, \beta) \in \Theta_1 \times \Theta_2$:
$$
    \lim_{n \xrightarrow{} + \infty }  \frac{\mathbb{P}_{\alpha} [ X_1 + \ldots + X_{n+1} < x] }{ \mathbb{P}_{\alpha} [ X_1 + \ldots + X_{n} < x] } =\lim_{n \xrightarrow{} + \infty }  \frac{c_{n+1,X}(\alpha) }{ c_{n,X}(\alpha) }= \mathbb{P}_{\alpha} [X_1 = 0]  $$
    and 
$$  \lim_{n \xrightarrow{} + \infty }  \frac{ \mathbb{P}_{\beta} [ Z_1 + \ldots + Z_{n+1} < z]}{ \mathbb{P}_{\beta} [ Z_1 + \ldots + Z_{n} < z]} =\lim_{n \xrightarrow{} + \infty }  \frac{c_{n+1,Z}(\beta) }{ c_{n,Z}(\beta) }= \mathbb{P}_{\beta} [Z_1 = 0] .$$
Here we consider the censoring indicator from Model II. Thus, in order to use Lemma \ref{lem:tech_resulu}, we consider the functions: for all $t \geq 0$
\begin{align*}
    \frac{f_{(Y,\Delta),\theta}(t,2)}{f_{(Y,\Delta),\theta}(t,0)} = & \frac{\sum_{n=0}^{+ \infty}  [c_{n,X}(\alpha) - c_{n+1,X}(\alpha) ] [c_{n,Z}(\beta) - c_{n+1,Z}(\beta) ] \lambda e^{- \lambda t} \frac{(\lambda t)^n}{n!}}{\sum_{n=0}^{+ \infty} [c_{n,Z}(\beta) - c_{n+1,Z}(\beta) ] c_{n+1,X}(\alpha) \lambda e^{- \lambda t} \frac{(\lambda t)^n}{n!}} \\
    \frac{f_{(Y,\Delta), \theta}(t,2)}{f_{(Y,\Delta),\theta}(t,1)} = & \frac{\sum_{n=0}^{+ \infty}  [c_{n,X}(\alpha) - c_{n+1,X}(\alpha) ] [c_{n,Z}(\beta) - c_{n+1,Z}(\beta) ] \lambda e^{- \lambda t} \frac{(\lambda t)^n}{n!}}{\sum_{n=0}^{+ \infty} [c_{n,X}(\alpha) - c_{n+1,X}(\alpha) ] c_{n+1,Z}(\beta) \lambda e^{- \lambda t} \frac{(\lambda t)^n}{n!}}
\end{align*}
Then from Lemma \ref{lem:tech_resulu}, we have:
\begin{align*}
    \lim_{t \xrightarrow{} + \infty} \frac{f_{(Y,\Delta),\theta}(t,2)}{f_{(Y,\Delta),\theta}(t,0)} = & \lim_{n \xrightarrow{} + \infty} \frac{[c_{n,X}(\alpha) - c_{n+1,X}(\alpha) ] [c_{n,Z}(\beta) - c_{n+1,Z}(\beta) ]}{[c_{n,Z}(\beta) - c_{n+1,Z}(\beta) ] c_{n+1,X}(\alpha)} =\dfrac{1}{ \mathbb{P}_{\alpha} [ X_1 = 0]} - 1 \\
    \lim_{t \xrightarrow{} + \infty} \frac{f_{(Y,\Delta),\theta'}(t,2)}{f_{(Y,\Delta),\theta'}(t,0)} = &  \lim_{n \xrightarrow{} + \infty} \frac{[c_{n,X}(\alpha') - c_{n+1,X}(\alpha') ] [c_{n,Z}(\beta') - c_{n+1,Z}(\beta') ]}{[c_{n,Z}(\beta') - c_{n+1,Z}(\beta') ] c_{n+1,X}(\alpha')} = \dfrac{1}{\mathbb{P}_{\alpha'} [ X_1 = 0]} - 1 
\end{align*}
Consequently $\mathbb{P}_{\alpha} [ X_1 = 0] = \mathbb{P}_{\alpha'} [ X_1 = 0]$. \\
The same way using the limit at infinity of $ \frac{f_{(Y,\Delta),\theta}(t,2)}{f_{(Y,\Delta),\theta}(t,1)} $, we get $\mathbb{P}_{\beta} [ Z_1 = 0] = \mathbb{P}_{\beta'} [ Z_1 = 0]$. 
Finally, studying the limit at infinity of the hazard function of the random variable $Y$, we obtain that $ \lambda(1 - \mathbb{P} [X_1 = 0] \mathbb{P} [Z_1 = 0]) = \lambda'(1 - \mathbb{P}_{\alpha'} [X_1 = 0] \mathbb{P}_{\beta'} [Z_1 = 0])$ and since $\mathbb{P}_{\alpha} [X_1 = 0] \mathbb{P}_{\beta} [Z_1 = 0] = \mathbb{P}_{\alpha'} [X_1 = 0] \mathbb{P}_{\beta'} [Z_1 = 0] < 1$, we deduce that $\lambda = \lambda'$.

\bigskip
\noindent {\bf Step 2.} Since $\lambda = \lambda'$, from evaluating all the derivatives of the density function at $t=0$, an induction gives us the property $I_{\infty}$ if we are under assumption {\bf (H2)}, while under {\bf (H1)}, $I_{N_{\max}}$ is verified. 
\qed

\bigskip 

As explained in Section \ref{sect:prop_estim}, theorems from section \ref{sect:prop_estim} are based on convergence results from \cite{Whi82}. Six assumptions are introduced in this paper and can be verified to obtain the existence, its consistency and its asymptotic normality of the QMLE. 
In order to prove those properties, we define the assumptions {\bf (H3)} (compactness of the set of parameters) and {\bf (H4)} (regularity conditions) detailed here: 
\begin{assumption}[Condition {\bf (H4)}] \label{ass:consistency}
There exist two measures $\nu_1$ and $\nu_2$ respectively defined on the measurable space $(\mathbb{R}_+, \mathcal{B}(\mathbb{R}_+))$ such that:
\begin{enumerate}
    \item For all $\alpha \in \Theta_1$, $\mathbb{P}_{\alpha}$ admits a Radon-Nikodym measure with respect to $\nu_1$, $f_1( \; . \; ,\alpha) = d \mathbb{P}_{\alpha} / d \nu_1$ such that: 
    \begin{enumerate}
         \item For all $u \in Dom(f_1) \cap [0,x]$; $\alpha \mapsto f_1(u,\alpha)$ is continuous on $\Theta_1$. 
          \item There exists a $\nu_1$-measurable and integrable function $\overline{f_1}$ on $[0,x]$ such that $\forall \alpha \in \Theta_1$, $f_1(u,\alpha) \leq \overline{f_1} (u)$, for almost every $u \in [0,x]$.
    \end{enumerate}
    \item For all $\beta \in \Theta_2$, $\mathbb{P}_{\beta}$ admits a Radon-Nikodym measure with respect to $\nu_2$, $f_2( \; . \; ,\beta) = d \mathbb{P}_{\beta} / d \nu_2$ such that:
   \begin{enumerate}
    \item For all $u \in Dom(f_2) \cap [0,z]$; $\beta \mapsto f_2(u,\beta)$ is continuous $\Theta_2$.
    \item There exists a $\nu_2$-measurable and integrable function $\overline{f_2}$ on $[0,z]$ such that $\forall \beta \in \Theta_2$, $f_2(u,\beta) \leq \overline{f_2} (u)$, for almost every $u \in [0,z]$.
    \end{enumerate}
\end{enumerate}
\end{assumption}

\paragraph{Proof of Theorem \ref{thm:existence_MLE}.}
In the proof of this theorem, we only discuss the case where we consider the censoring couple $(Y,\Delta)$ from Model I. The proof with the strict censoring couple from Model II uses the exact same arguments.
To prove this theorem, we verify if A1 and A2 from \cite{Whi82} are verified.

\medskip 

\noindent \underline{Assumption A1.}
In our case, the random vector $(Y,\Delta)$ takes values in the Euclidean measurable space $(\mathbb{R}^2, \mathcal{B}(\mathbb{R}^2))$. We consider the measure $v= m \otimes d$, where $m$ denotes the Lebesgue measure on $(\mathbb{R}, \mathcal{B}(\mathbb{R}))$ and $d = \delta_0 + \delta_1$, (where $\delta_a$ is the Dirac measure in $a$). It is clear that the random vector $(Y,\Delta)$ admits a measurable Radon-Nikodym density function with respect to $v$, which is $f_{(Y, \Delta)}= f_{(Y,\Delta),\theta}$ given in Lemma \ref{lem:dens_fct_censor}.

\medskip 

\noindent \underline{Assumption A2.}
As recalled in our statistical context, the model is fully determined by the parameter $\lambda \in \mathbb{R}^*_+ $, the respective families of distribution functions $\{ \mathbb{P}_\alpha , \alpha \in \Theta_1 \}$ and $\{ \mathbb{P}_\beta , \beta \in \Theta_2 \}$, with $\Theta_1$ and $\Theta_2$ respectively subsets of $\mathbb{R}^{d_1}$ and $\mathbb{R}^{d_2}$. 

To ensure Assumption A2, it is required to have a compact subset of parameter, which justifies Condition {\bf (H3)}. It guarantees that $\Lambda \times \Theta_1 \times \Theta_2$ is a compact subset of the Euclidean space $\mathbb{R}^{1+d_1+d_2}$. For any $(\lambda,\alpha,\beta) \in \Lambda \times \Theta_1 \times \Theta_2$, we denote by $F_{(Y,\Delta),(\lambda,\alpha,\beta)}$ the distribution function induced by the parameters vector $\theta=(\lambda,\alpha,\beta)$ and by $\theta^0=(\lambda^0,\alpha^0,\beta^0) \in \Lambda \times \Theta_1 \times \Theta_2$ the true vector of parameters model. 
For any $(\lambda,\alpha,\beta) \in \Lambda \times \Theta_1 \times \Theta_2$, $F_{(Y,\Delta),(\lambda,\alpha,\beta)}$ admits a measurable Radon-Nikodym function with respect to $v$, $f_{(Y,\Delta),(\lambda,\alpha,\beta)} = dF_{(Y,\Delta),(\lambda,\alpha,\beta)} / dv$ (Lemma \ref{lem:dens_fct_censor}). 

The continuity in $\lambda$ of the density function $f_{((Y,\Delta),(\lambda,\alpha,\beta))}$ comes naturally for any $(t,\delta) \in \mathbb{R}_+ \times \{0;1\}$ (it is a uniformly convergent series of function on any compact). To obtain the continuity in $\alpha$ and $\beta$, we need to prove that for any $n \in \mathbb{N}$:
\begin{eqnarray*}
    \alpha \mapsto c_{n,X}(\alpha)= \mathbb{P}_{\alpha} \left[ \sum_{i=0}^n X_i < x \right] \quad \quad \quad \quad \text{and} \quad \quad \quad \quad \beta \mapsto c_{n,Z}(\beta)= \mathbb{P}_{\beta} \left[ \sum_{i=0}^n Z_i < z \right]
\end{eqnarray*}
are continuous functions. Since the proofs use the same arguments, we will only prove the result for the first function. \\
If $n=0$, the functions are constant functions equal to 1, so they are continuous. Let $n \in \mathbb{N}^*$. We can re-write:
\begin{eqnarray*}
        \mathbb{P}_{\alpha} \left[ \sum_{i=0}^n X_i < x \right] = \int_{\mathbb{R}^n_+} \mathbb{1}_{ \{ 0 \leq x_1 + \ldots x_n < x \} } \prod_{i=1}^n f_1(x_i,\alpha) d \nu_1 ( x_1) \ldots d \nu_1 ( x_n). 
\end{eqnarray*}
From \textbf{(H4)} \ref{ass:consistency}-1.(a), it comes that for any $(x_1, \ldots , x_n) \in \{ (y_1, \ldots , y_n) \in Dom(f_1); \; 0 \leq y_1 + \ldots + y_n < x \} $, $ \alpha \mapsto \prod_{i=1}^n f_1(x_i,\alpha)$ is a continuous function. Moreover from {\bf (H4)}-1.(b), for almost every $(x_1, \ldots , x_n) \in \{ (y_1, \ldots , y_n) \in Dom(f_1); \; 0 \leq y_1 + \ldots + y_n < x \} $:
\begin{eqnarray*}
    \prod_{i=1}^n f_1(x_i,\alpha)  \leq \prod_{i=1}^n \overline{f_1}(x_i)
\end{eqnarray*}
where 
\begin{eqnarray*}
    (x_1, \ldots , x_n) \mapsto \prod_{i=1}^n \overline{f_1}(x_i)
\end{eqnarray*}
is an integrable function on the domain $\{ (y_1, \ldots , y_n) \in Dom(f_1); \; 0 \leq y_1 + \ldots + y_n < x \} $. Finally, due to the dominated convergence theorem (applied to the continuity of parametrized integrals), we deduce that $\alpha \mapsto \mathbb{P} [ \sum_{i=0}^n X_i < x]$ is continuous. 
That proves that assumption A2 from \cite{Whi82} is verified. Thus, we can apply \cite[Theorem 2.1] {Whi82} and the statement of Theorem \ref{thm:existence_MLE} is proved. 
\qed

\paragraph{Auxiliary lemma.} In order to prove Theorem \ref{thm:consistency}, we need the following lemma.

\begin{lemma} \label{lem:tech_behavior_c_n}
Assume {\bf (H5)} is verified and that for some $N \in \mathbb{N}$, $\alpha^* \in \Theta_1$, we have:
\begin{eqnarray*}
    \mathbb{P}_{\alpha^*} [X_1 + \ldots + X_N < x ] = \mathbb{P}_{\alpha^*} [X_1 + \ldots + X_{N+1} < x ]
\end{eqnarray*}
Then, $\forall \alpha \in \Theta_1$, $\mathbb{P}_{\alpha} [X_1 + \ldots + X_N < x ] = \mathbb{P}_{\alpha} [X_1 + \ldots + X_{N+1} < x ] $.    
\end{lemma}
\begin{proof}
Assume that for some $N \in \mathbb{N}$, $\alpha^* \in \Theta_1$, we have:
\begin{eqnarray*}
     \mathbb{P}_{\alpha^*} [X_1 + \ldots + X_N < x ] = \mathbb{P}_{\alpha^*} [X_1 + \ldots + X_{N+1} <   x].
\end{eqnarray*}
From lemma \ref{lem:equality_probab_term}, this quantity is zero or one. If the lemma assumption is verified, then for some $\alpha^*$, $\textit{Supp}\{ X_1 + \ldots + X_N \} \subset [x, + \infty[$ and consequently, $\textit{Supp}\{ X_1 \} \subset [x/N, + \infty[$.
Since the support of $X_1$ does not depend on $\alpha^*$, then the same property is verified for any $\alpha \in \Theta_1$ and we obtain
\begin{eqnarray*}
    0 \leq \mathbb{P}_{\alpha} [X_1 + \ldots + X_{N+1} < x ] \leq \mathbb{P}_{\alpha} [X_1 + \ldots + X_{N} < x ] = 0.
\end{eqnarray*}

If $\mathbb{P}_{\alpha^*} [X_1 + \ldots + X_N < x ] = 1$, as in the proof of Lemma \ref{lem:equality_probab_term}, we define $M_{\max} = \sup \{ t \in \textit{Supp}\{ X_1 \} \}$ ; note that $M_{\max}$ does not depend on $\alpha^*$. Once again from the proof of Lemma \ref{lem:equality_probab_term}, we have $M_{\max} \leq x$ and that $\textit{Supp} \{ X_1 + \ldots + X_N \} \cap [0,x[ \subset [0,x-M_{\max}[ $.

Now let $\alpha \in \Theta_1$. Since $\textit{Supp} \{ X_1 + \ldots + X_N \} \cap [0,x[ \subset [0,x-M_{\max}[ $, %and $ M_{X_1}$ is not an atom of $X_1$ for $\alpha^*$, 
and the support of $X_1$ does not depend on the parameter. We have
\begin{align*}
    & \mathbb{P}_{\alpha} [X_1 + \ldots + X_{N} < x ] \\
    & = \mathbb{P}_{\alpha} [X_1 + \ldots + X_{N} < x - M_{\max} ] \mathbb{P}_{\alpha} [X_1 \leq M_{\max} ] \\
    & = \mathbb{P}_{\alpha} [X_1 + \ldots + X_{N} <x - M_{\max} ; X_{N+1} < M_{\max} ]  \leq \mathbb{P}_{\alpha} [X_1 + \ldots + X_{N+1} < x].
\end{align*}
So $\mathbb{P}_{\alpha} [X_1 + \ldots + X_{N} < x ] = \mathbb{P}_{\alpha} [X_1 + \ldots + X_{N+1} < x ]$, and the proof is done. 
%The case where $\textit{Supp} \{ X_1 + \ldots + X_N \} \cap [0,x[ \subset [0,x-M_{X_1}[ $ and $ M_{X_1}$ is an atom of $X_1$ is treated the same way.
\end{proof}

\paragraph{Notations}
In order to enlighten the formulas in the next theorem, we define
\begin{align*}
    \lambda_{\min} & = \min \{ \lambda \in \Lambda \}, \qquad  \lambda_{\max}  = \max \{ \lambda \in \Lambda \} \\ 
    m_{n,X} & = \min \{ c_{n,X}(\alpha), \quad \alpha \in \Theta_1  \},  \quad
    \overset{\sim}{m}_{n,X}  = \min \{ c_{n,X}(\alpha) - c_{n+1,X}(\alpha), \quad \alpha \in \Theta_1  \}  \\
    m_{n,Z} & = \min \{ c_{n,Z}(\beta), \quad \beta \in \Theta_2 \}  , \quad
    \overset{\sim}{m}_{n,Z} = \min \{ c_{n,Z}(\beta) - c_{n+1,Z}(\beta), \quad \beta \in \Theta_2  \}.
\end{align*}
Under Conditions {\bf (H3)} (compactness) and {\bf (H4)} (regularity), these quantities are finite and non-negative. 

\paragraph{Proof of Theorem \ref{thm:consistency}.}
In the proof of this theorem, we will only discuss on the case where we consider the censoring couple $(Y,\Delta)$ from Model I. The proof for the censoring couple $(Y,\Delta)$ from Model II uses the exact same arguments since the density of the model has a similar structure. \\
To prove this theorem, we need to check Conditions A1 to A3 of \cite{Whi82}. A1 and A2 have already been proved in the proof of Theorem \ref{thm:existence_MLE}. \\
We firstly prove that there exist $N_1$, $N_2$, $N_3$ in $\mathbb{N}$ such that $\overset{\sim}{m}_{N_1,X} m_{N_1,Z} > 0 $, $\overset{\sim}{m}_{N_2,Z} m_{N_2 +1,X} > 0 $ and if $f_{(Y,\Delta)}( \, . \, ,0) \neq 0$, $\overset{\sim}{m}_{N_3,X} \overset{\sim}{m}_{N_3,Z} > 0 $. 
To prove it, we use the fact that $\mathbb{P} \left[ T = C   \right] $ is positive. Hence, we deduce from the formula of Lemma \ref{lem:null_prob_induces_non_censoring} the existence of $N_1 \in \mathbb{N}$ such that
$$
(c_{N_1,X}(\alpha^0) - c_{N_1+1,X}(\alpha^0))(c_{N_1,Z}(\beta^0)-c_{N_1+1,Z}(\beta^0) ) >  0.
$$
And by applying Lemma \ref{lem:tech_behavior_c_n} using a reductio ad absurdum, we deduce that $\overset{\sim}{m}_{N_1,X} \overset{\sim}{m}_{N_1,Z} > 0 $.
Since 
$$
    (c_{n,X}(\alpha^0) - c_{n+1,X}(\alpha^0))(c_{n,Z}(\beta^0)-c_{n+1,Z}(\beta^0)) \leq 
    (c_{n,X}(\alpha^0) - c_{n+1,X}(\alpha^0))c_{n,Z}(\beta^0),
$$
and since the support of $Z_1$ does not depend on $\beta$, a similar proof justifies that $\overset{\sim}{m}_{N_1,X} m_{N_1,Z} > 0 $.
To prove the third inequality $\overset{\sim}{m}_{N_2,Z} m_{N_2 +1,X} > 0 $ for some $N_2 \in \mathbb{N}$  when $f_{(Y,\Delta)}( \, . \, ,0) \neq 0$, assume ad absurdum that it is not the case. We can prove from Lemma \ref{lem:tech_behavior_c_n} and from the fact that the support of $X_1$ does not depend on $\alpha$ that
$$
     f_{(Y,\Delta),\theta}(t,0) = 0 \quad \quad \quad \forall t \geq 0, \ \forall \theta \in \Theta.
$$
Hence $T \leq C$ a.s., which is also absurd.
\\
\noindent \underline{Assumption A3.} Uniqueness of the minimum of the Kullback-Leibler information criterion (hypothesis A3.b) is guaranteed here, by our parametric setting and our identifiability condition. \\
\noindent To have Condition A3.a of \cite{Whi82}, we need to obtain a function that uniformly bounds the family of density functions $f_{(Y,\Delta),\theta}$ on $\Lambda \times \Theta_1 \times \Theta_2$, independently of $\theta$. It is now easy to see that, if $\delta = 1$, for any $\theta \in \Lambda \times \Theta_1 \times \Theta_2$ and any $t > 0$
\begin{eqnarray*}
    e^{- \lambda_{\max} t} \sum_{n \geq 0} \overset{\sim}{m}_{n,X} m_{n,Z} \frac{(\lambda_{\min}t)^n}{n!} \leq f_{(Y,\Delta),\theta} (t,1) \leq e^{- \lambda_{\min} t} e^{ \lambda_{\max} t}.
\end{eqnarray*}
Since there exists $N_1 \in \mathbb{N}$ such that $\overset{\sim}{m}_{N_1,X} m_{N_1,Z} > 0 $, then for any $t > 0$
$$
    e^{- \lambda_{\max} t} \overset{\sim}{m}_{N_1,X} m_{N_1,Z} \frac{(\lambda_{\min}t)^{N_1}}{N_1!} \leq f_{(Y,\Delta),\theta} (t,1) \leq e^{(\lambda_{\max} - \lambda_{\min}) t}.
$$
and thus for any $t > 0$
$$
| \log f_{(Y,\Delta),\theta} (t,1) | \leq \max \left\{ | \lambda_{\max} - \lambda_{\min} | t ; \left| - \lambda_{\max} t + \log \left( \frac{\overset{\sim}{m}_{N_1,X} m_{N_1,Z}}{N_1!} \right) N_1 \log (\lambda_{\min} t ) \right|  \right\}.
$$
The bounding function is independent of $\theta \in \Lambda \times \Theta_1 \times \Theta_2$ and is of linear growth w.r.t. $t$. Hence it is integrable with respect to $dF_{(Y,\Delta),(\lambda^0,\alpha^0,\beta^0)}$, the density function with true parameter $(\lambda^0,\alpha^0,\beta^0)$. Similarly the fact that there exists $N_2$ such that $\overset{\sim}{m}_{N_2,Z} m_{N_2 +1,X} > 0 $ ensures the existence of a similar dominating function when $\delta = 0$ and  when $f_{(Y,\Delta)}( \, . \, ,0) \neq 0$. 

Since in our setting the function $g$ in \cite{Whi82} is equal to $f_{(Y,\Delta),(\lambda^0,\alpha^0,\beta^0)}$, from the previous bounds, it follows that $\mathbb{E}[ \log( g(Y,\Delta) ]$ exists. This achieves the proof of Theorem \ref{thm:consistency}. 
\qed
%A3)b) In his article, White explains that if $\exists (\lambda^0,\alpha^0,\beta^0) \in \Lambda \times \Theta_1 \times \Theta_2$ such that $g = f_{((Y,\Delta),(\lambda^0,\alpha^0,\beta^0))}$, then A3)b) is verified, which is the case here.

\bigskip
The following lemma will help to obtain proper uniform dominations of the parametric derivatives of the density functions in the proof of theorem \ref{theo:assymptotic_normaility}.

\begin{lemma} \label{lem::unif_domi}
    Assume that {\bf (H3)}, {\bf (H4)} and {\bf (H5)} are verified. Then there exist $C_1 > 1$ and $ N_1 \in \mathbb{N}^*$ such that 
    $$
   \forall n \geq N_1, \forall \alpha \in \Theta_1, \quad C_1 \mathbb{P}_{\alpha} [X_1 + \ldots + X_{n+1} < x ] \leq \mathbb{P}_{\alpha} [X_1 + \ldots + X_{n} < x ].$$
The same property holds for the sequence $(Z_n)_{n\in \mathbb N}$.    
 %$\exists C_2 > 1$, $ \exists N_2 \in \mathbb{N}^*, \forall n \geq N_2; \forall \beta \in \Theta_2; C_2 \mathbb{P}_{\beta} [Z_1 + \ldots + Z_{n+1} < z ] \leq \mathbb{P}_{\beta} [Z_1 + \ldots + Z_{n} < z ]$
\end{lemma}
\begin{proof}
We prove the result on the sequence $(X_n)_{n \in \mathbb{N}^*}$; the proof is the same for $(Z_n)_{n \in \mathbb{N}^*}$. Assumption
{\bf (H3)} guarantees that $\Theta_1$ is a compact subset of $\mathbb{R}^{d_1}$ and assumption {\bf (H5)} guarantees that the support of $X_1$ does not depend on $\alpha$.
We now need to distinguish whether $(X_n)_{n \in \mathbb{N}^*}$ belongs to the family $F_1$, $F_2$ or $F_3$ (see Definition \ref{def:families}). If $(X_n)_{n \in \mathbb{N}^*} \in F_1$, then $c_{n,X}(\alpha) = 0$ after some threshold $N_1$, which does not depend on $\alpha$ thanks to {\bf(H5)}. Any constant $C_1>1$ does the job. \\
If $(X_n)_{n \in \mathbb{N}^*} \in F_2$, let $N_1(\alpha) = \min \{ n \geq 2, c_{n-1,X}(\alpha)  < 1  \} $. Note that under assumptions {\bf (H3)} and {\bf(H4)}, there exists $N_1\geq 2$ such that for all $\alpha \in \Theta_1$ and $n\geq N_1$, $c_{n,X}(\alpha) < 1$. Furthermore if $c_{N_1,X}(\alpha) = 0$, then for all $ n\geq N_1$, $c_{n,X}(\alpha)= 0$ and any $C_1>1$ is suitable. Now we suppose that $0 < c_{N_1,X}(\alpha) < 1$. Then for all  $n \geq N_1$ and $\alpha \in \Theta_1$
\begin{align*}
    \frac{\mathbb{P}_{\alpha} [X_1 + \ldots + X_{n+1} < x ] }{\mathbb{P}_{\alpha} [X_1 + \ldots + X_{n} < x ]} & = \frac{ \int_0^x \phi_{S_{n+1-N_1},\alpha} (x - t) \mathbb{P}_{\alpha} [X_1 + \ldots + X_{N_1} <  t ]  dt }{ \int_0^x \phi_{S_{n+1-N_1},\alpha} (x-t) \mathbb{P}_{\alpha} [X_1 + \ldots + X_{N_1 -1} <  t ]  dt } \\
    & = \dfrac{ \int_0^x \phi_{S_{n+1-N_1},\alpha} (x-t) \mathbb{P}_{\alpha} [X_1 + \ldots + X_{N_1 -1} < t ] 
    \frac{\mathbb{P}_{\alpha} [X_1 + \ldots + X_{N_1} <  t ]}{\mathbb{P}_{\alpha} [X_1 + \ldots + X_{N_1 -1} < t ] }  dt }{ \int_0^x \phi_{S_{n+1-N_1},\alpha} (x-t) \mathbb{P}_{\alpha} [X_1 + \ldots + X_{N_1 -1} < t ]  dt }.
\end{align*}
The function 
$ h_{N_1} : (t,\alpha) \mapsto \frac{\mathbb{P}_{\alpha} [X_1 + \ldots + X_{N_1} <  t ]}{\mathbb{P}_{\alpha} [X_1 + \ldots + X_{N_1 -1} < t ] }$ is continuous on $]0,x] \times \Theta_1$ and for any $\alpha \in \Theta_1$, $ \lim_{t \xrightarrow{} 0} h_{N_1}(t,\alpha) = 0$. % $\lim_{t \xrightarrow{} 0^+} \frac{\mathbb{P}_{\alpha} [X_1 + \ldots + X_{N_1} <  t ]}{\mathbb{P}_{\alpha} [X_1 + \ldots + X_{N_1 -1} < t ] } = 0$. 
Thus, the function can be extended to a continuous function on the compact subset $[0,x] \times \Theta_1$, which is still denoted by $h_{N_1}$. \\
We have for any fixed $\alpha \in \Theta_1$, $\forall t \in [0,x]$, $h_{N_1} (t) < 1$. If it is not the case, for some $t \in ]0,x]$, it follows that $\mathbb{P}_{\alpha} [X_1 + \ldots + X_{N_1} <  t ] = \mathbb{P}_{\alpha} [X_1 + \ldots + X_{N_1-1} <  t ]$ and Lemma \ref{lem:equality_probab_term} gives that $\mathbb{P}_{\alpha} [X_1 + \ldots + X_{N_1} <  t ] \in \{0 ; 1 \}$. Since $X_n$ is in $F_2$, it cannot be zero. If it is one, then $c_{N_1,X}(\alpha) =1$, which contradicts our hypothesis. Thus $h_{N_1}$ is a continuous function strictly dominated by 1 on $[0,x] \times \Theta_1$. Hence we deduce that there exists $K_1 < 1$ such that, $h_{N_1}(t, \alpha) \leq K_1$, $\forall (t,\alpha) \in [0,x] \times \Theta_1$. Consequently
\begin{eqnarray*}
    \frac{\mathbb{P}_{\alpha} [X_1 + \ldots + X_{n+1} < x ] }{\mathbb{P}_{\alpha} [X_1 + \ldots + X_{n} < x ]} \leq K_1
\end{eqnarray*}
and $C_1= 1/K_1$ is a right constant. \\
If $(X_n)_{n \in \mathbb{N}^*} \in F_3$, in proof of Theorem \ref{theo:behaviour_c_(n+1)/c_n}, one can see that we obtained:
\begin{eqnarray*}
    \left| \frac{\mathbb{P}_{\alpha} [ S_{n+1} < x] }{ \mathbb{P}_{\alpha} [ S_{n} < x] } - \mathbb{P}_{\alpha} [X_1 = 0] \right| \leq \frac{C_x}{n}, \quad \quad \quad \forall n \geq M_{\max,x}
\end{eqnarray*}
where $C_x = (M_{\max,x-x_{\min}} + 1) \frac{\mathbb{P}_{\alpha} [X_1 = 0]}{\mathbb{P}_{\alpha} [X_1 = x_{\min}]}$, with $ M_{\max,x}$ and $M_{\max,x-x_{\min}}$ that do not depend on $\alpha$. Therefore
\begin{eqnarray*}
    \frac{\mathbb{P}_{\alpha} [ S_{n+1} < x] }{ \mathbb{P}_{\alpha} [ S_{n} < x] }  \leq (M_{\max,x-x_{\min}} + 1) \frac{\mathbb{P}_{\alpha} [X_1 = 0]}{\mathbb{P}_{\alpha} [X_1 = x_{\min}]} \frac{1}{n} \mathbb{P}_{\alpha} [X_1 = 0], \quad \quad \quad \forall n \geq M_{\max,x}
\end{eqnarray*}
where $ \alpha \mapsto \frac{\mathbb{P}_{\alpha} [X_1 = 0]}{\mathbb{P}_{\alpha} [X_1 = x_{\min}]}$ and $\alpha \mapsto \mathbb{P}_{\alpha} [X_1 = 0]$ are continuous functions on the compact subset $\Theta_1$. We can consider $C_1^{-1} = \sup_{ \alpha \in \Theta_1} \{ (\mathbb{P}_{\alpha} [X_1 = 0] +1)/2 \} $ as a uniform dominating constant.
\end{proof}

\begin{assumption}[Condition {\bf (H6)}] \label{ass:normality}
With the notations of Condition {\bf(H4)} (see Assumption \ref{ass:consistency}):
\begin{enumerate}
    \item For all $u \in Dom(f_1) \cap [0,x]$ the function $\alpha \mapsto f_1(u,\alpha)$ is twice differentiable on $ \overset{\circ}{\Theta_1}$ with:
    \begin{enumerate}
        \item $\forall \alpha \in \overset{\circ}{\Theta_1}$, $u \in ]0,x]  \mapsto  \partial_{\alpha} f_1(u,\alpha) $ is measurable.
        \item For any $i_1,i_2 \in \{1, \ldots , d_1 \}$, there exists $g_{1,i_1}$ and $h_{1,i_1,i_2}$ two integrable functions with respect to $\nu_1$ on $[0,x]$ such that: 
    \begin{eqnarray*}
        \left| \frac{\partial f_1(u,\alpha)}{\partial \alpha_{i_1}} \right| \leq g_{1,i_1} ( u) \quad \quad \text{and} \quad \quad \left| \frac{\partial^2 f_1(u,\alpha)}{\partial \alpha_{i_1} \partial \alpha_{i_2} } \right| \leq h_{1,i_1,i_2} ( u)  \quad \quad \nu_1 - a.e.
    \end{eqnarray*}
    \end{enumerate}
    
    \item For all $u \in Dom(f_2) \cap [0,z]$, the function $\beta \mapsto f_2(u,\beta)$ is two times differentiable on $ \overset{\circ}{\Theta_2}$ with:
    \begin{enumerate}
    \item $\forall \beta \in \overset{\circ}{\Theta_2}$, $u \in ]0,z]  \mapsto  \partial_{\beta} f_2(u,\beta) $ is measurable.
    \item For any $j_1,j_2 \in \{1, \ldots , d_2 \}$, there exists $g_{2,j_1}$ and $h_{2,j_1,j_2}$ two integrable functions with respect to $\nu_2$ on $[0,z]$ such that
    \begin{eqnarray*}
        \left| \frac{\partial f_2(u,\beta)}{\partial \beta_{j_1}} \right| \leq g_{2,j_1} ( u) \quad \quad \text{and} \quad \quad \left| \frac{\partial^2 f_2(u,\beta)}{\partial \beta_{j_1} \partial \beta_{j_2} } \right| \leq h_{2,j_1,j_2} ( u)  \quad \quad \nu_2 - a.e.
    \end{eqnarray*}
\end{enumerate}
\end{enumerate}

\end{assumption}

\paragraph{Proof of Theorem \ref{theo:assymptotic_normaility}.}
In the proof of this theorem, we only discuss the case where we consider the censoring couple $(Y,\Delta)$ for Model I. The proof for the strict censoring couple for Model II uses the exact same arguments since the density of the second couple has a similar structure. 

To prove this theorem, we need to verify assumption A1-A6 in \cite{Whi82}. The assumptions A1-A3 have already been proved in the previous proofs, so only the analysis of A4-A6 are made. 

\smallskip
\noindent \underline{Assumption A4.} Let us recall (Lemma \ref{lem:dens_fct_censor}) that, for any $(\lambda,\alpha,\beta) \in \Lambda \times \Theta_1 \times \Theta_2, t > 0$,
\begin{align*}
    f_{(Y, \Delta),(\lambda,\alpha,\beta)} (t, 1) & =  \sum_{n=0}^{+ \infty}  [c_{n,X}(\alpha) - c_{n+1,X}(\alpha) ] c_{n,Z}(\beta) \lambda e^{- \lambda t} \frac{(\lambda t)^n}{n!}   \\
    f_{(Y, \Delta),(\lambda,\alpha,\beta)} (t, 0) &=  \sum_{n=0}^{+ \infty} [c_{n,Z}(\beta) - c_{n+1,Z}(\beta) ] c_{n+1,X}(\alpha) \lambda e^{- \lambda t} \frac{(\lambda t)^n}{n!}  
\end{align*}
It easily comes from the previous assumption that A4 is verified in our case. In fact, if $\delta = 1$, then for all $t>0$, $i \in \{ 1,\ldots,d_1\}$ and $j \in \{ 1,\ldots,d_2\}$
$$\log ( f_{(Y, \Delta),(\lambda,\alpha,\beta)} (t, 1)) = - \lambda t + \log \left( \sum_{n=0}^{+ \infty}  [c_{n,X}(\alpha) - c_{n+1,X}(\alpha) ] c_{n,Z}(\beta) \lambda  \frac{(\lambda t)^n}{n!} \right).$$ 
From the proof of Theorem \ref{thm:consistency}, the term in the logarithm is positive and we can compute all the derivatives
\begin{align} \label{eq:deriv_log_lambda}
    \frac{\partial \log ( f_{(Y, \Delta),(\lambda,\alpha,\beta)} (t,1)) }{ \partial \lambda} & = - t + \frac{ \sum_{n=0}^{+ \infty}  [c_{n,X}(\alpha) - c_{n+1,X}(\alpha) ] c_{n,Z}(\beta) (n+1) \frac{(\lambda t)^n}{n!} }{ \sum_{n=0}^{+ \infty}  [c_{n,X}(\alpha) - c_{n+1,X}(\alpha) ] c_{n,Z}(\beta) \lambda  \frac{(\lambda t)^n}{n!}} \\ \label{eq:deriv_log_alpha}
    \frac{\partial \log ( f_{(Y, \Delta),(\lambda,\alpha,\beta)} (t, 1)) }{ \partial \alpha_i} & = \frac{ \sum_{n=0}^{+ \infty} \partial_{\alpha_i} [c_{n,X}(\alpha) - c_{n+1,X}(\alpha) ] c_{n,Z}(\beta) \frac{(\lambda t)^n}{n!} }{ \sum_{n=0}^{+ \infty}  [c_{n,X}(\alpha) - c_{n+1,X}(\alpha) ] c_{n,Z}(\beta)  \frac{(\lambda t)^n}{n!}}\\ \nonumber
    \frac{\partial \log ( f_{(Y, \Delta),(\lambda,\alpha,\beta)} (t, 1)) }{ \partial \beta_j} & = \frac{ \sum_{n=0}^{+ \infty} [c_{n,X}(\alpha) - c_{n+1,X}(\alpha) ] \partial_{\beta_j} c_{n,Z}(\beta) \frac{(\lambda t)^n}{n!} }{ \sum_{n=0}^{+ \infty} [c_{n,X}(\alpha) - c_{n+1,X}(\alpha) ] c_{n,Z}(\beta) \frac{(\lambda t)^n}{n!}}
\end{align}
The same computations hold when $\delta = 0$. Since for any $(\lambda,\alpha,\beta) \in \Lambda \times \Theta_1 \times \Theta_2$, all of these above functions are continuous in t, whether $\delta=0$ or $\delta=1$ on $(\mathbb{R}^*_+, \mathcal{B}(\mathbb{R}^*_+))$, it is clear that they are measurable in $(t,\delta) \in \mathbb{R}^*_+ \times \{0,1\}$. The fact that functions are continuously differentiable comes from all the assumptions made earlier and from theorem of derivation of parametrized integrals.

%, then $\forall t>0, \forall i \in \{ 1,\ldots,d_1\},\forall j \in \{ 1,\ldots,d_2\}$ \\
%$ \log ( f_{(Y, \Delta),(\lambda,\alpha,\beta)} (t, \delta)) = - \lambda t + \log \left( \sum_{n=0}^{+ \infty}  [c_{n,Z}(\beta) - c_{n+1,Z}(\beta) ] c_{n+1,X}(\alpha) \lambda  \frac{(\lambda t)^n}{n!} \right)$ \\
%$ \Rightarrow \frac{\partial \log ( f_{(Y, \Delta),(\lambda,\alpha,\beta)} (t, \delta)) }{ \partial \lambda} = - t + \frac{ \sum_{n=0}^{+ \infty}  [c_{n,Z}(\beta) - c_{n+1,Z}(\beta) ] c_{n+1,X}(\alpha) (n+1) \frac{(\lambda t)^n}{n!} }{ \sum_{n=0}^{+ \infty}  [c_{n,Z}(\beta) - c_{n+1,Z}(\beta) ] c_{n+1,X}(\alpha) \lambda  \frac{(\lambda t)^n}{n!}}$ \\
%$ \Rightarrow \frac{\partial \log ( f_{(Y, \Delta),(\lambda,\alpha,\beta)} (t, \delta)) }{ \partial \alpha_i} = \frac{ \sum_{n=0}^{+ \infty} [c_{n,Z}(\beta) - c_{n+1,Z}(\beta) ] \partial_{\alpha_i} c_{n+1,X}(\alpha) \frac{(\lambda t)^n}{n!} }{ \sum_{n=0}^{+ \infty}  [c_{n,Z}(\beta) - c_{n+1,Z}(\beta) ] c_{n+1,X}(\alpha)  \frac{(\lambda t)^n}{n!}}$  \\
%$ \Rightarrow \frac{\partial \log ( f_{(Y, \Delta),(\lambda,\alpha,\beta)} (t, \delta)) }{ \partial \beta_j} = \frac{ \sum_{n=0}^{+ \infty} \partial_{\beta_j} [c_{n,Z}(\beta) - c_{n+1,Z}(\beta) ] c_{n+1,X}(\alpha) \frac{(\lambda t)^n}{n!} }{ \sum_{n=0}^{+ \infty}  [c_{n,Z}(\beta) - c_{n+1,Z}(\beta) ] c_{n+1,X}(\alpha)  \frac{(\lambda t)^n}{n!}}$  \\

\smallskip
\noindent \underline{Assumption A5.} %Some work must be done to obtain suitable domination of the derivatives. Because 
We need to get uniform bounds for the first and second derivatives w.r.t. the parameters $\theta = (\lambda,\alpha,\beta) \in \Lambda \times \Theta_1 \times \Theta_2$. 

Let start with the first derivative w.r.t. $\lambda$. For $\delta= 1$, with \eqref{eq:deriv_log_lambda}, we have 
%$\mathbf{Study \; of \; the \; first \; function:}$ \\
%We will leave it to the reader to verify that: \\
%\underline{If $\delta = 1$:} \\
$$ \frac{\partial \log ( f_{(Y, \Delta),(\lambda,\alpha,\beta)} (t, 1)) }{ \partial \lambda} = - t + \frac{1}{\lambda} + t \frac{ \sum_{n=0}^{+ \infty}  [c_{n+1,X}(\alpha) - c_{n+2,X}(\alpha) ] c_{n+1,Z}(\beta) \frac{(\lambda t)^n}{n!} }{ \sum_{n=0}^{+ \infty}  [c_{n,X}(\alpha) - c_{n+1,X}(\alpha) ] c_{n,Z}(\beta)   \frac{(\lambda t)^n}{n!}}.$$
From Lemma \ref{lem::unif_domi},
%\begin{enumerate}
  %  \item $\exists C_1 > 1$, $ \exists N_1 \in \mathbb{N}^*, \forall n \geq N_1; \forall \alpha \in \Theta_1; C_1 \mathbb{P}_{\alpha} [X_1 + \ldots + X_{n+1} < x ] \leq \mathbb{P}_{\alpha} [X_1 + \ldots + X_{n} < x ]$ 
  %  \item $\exists C_2 > 1$, $ \exists N_2 \in \mathbb{N}^*, \forall n \geq N_2; \forall \beta \in \Theta_2; C_2 \mathbb{P}_{\beta} [Z_1 + \ldots + Z_{n+1} < z ] \leq \mathbb{P}_{\beta} [Z_1 + \ldots + Z_{n} < z ] $
%\end{enumerate}
%And since $N_1 < + \infty$, 
it follows that for all $n \geq N_1$ and $\alpha \in \Theta_1$,
$$
    0 \leq  [c_{n+1,X}(\alpha) - c_{n+2,X}(\alpha) ] \leq \frac{1}{C_1 -1} [c_{n,X}(\alpha) - c_{n+1,X}(\alpha) ] .
$$
And 
\begin{align*}
& \sum_{n=0}^{N_1}  [c_{n+1,X}(\alpha) - c_{n+2,X}(\alpha) ] c_{n+1,Z}(\beta) \frac{(\lambda t)^n}{n!} =\dfrac{1}{\lambda t} \sum_{m=1}^{N_1+1}  [c_{m,X}(\alpha) - c_{m+1,X}(\alpha) ] c_{m,Z}(\beta)  m \frac{(\lambda t)^m}{m!} \\
&\quad  \leq \dfrac{N_1+1}{\lambda t} \sum_{m=1}^{N_1+1}  [c_{m,X}(\alpha) - c_{m+1,X}(\alpha) ] c_{m,Z}(\beta) \frac{(\lambda t)^m}{m!} .
\end{align*}
Hence we immediately obtain the following uniform dominating function: for $t > 0$
\begin{eqnarray*}
    \left|  \frac{\partial \log ( f_{(Y, \Delta),(\lambda,\alpha,\beta)} (t,1)) }{ \partial \lambda}  \right| \leq t + \frac{1}{\lambda_{\min}} + \frac{N_1 + 1}{\lambda_{\min}} + \frac{t}{C_1 - 1} = m_1(t)
\end{eqnarray*}
For $\delta = 0$, similar reflections ensures the following uniform dominating function: for $t > 0$
\begin{eqnarray*}
    \left|  \frac{\partial \log ( f_{(Y, \Delta),(\lambda,\alpha,\beta)} (t, 0)) }{ \partial \lambda}  \right| \leq |t| + \frac{1}{\lambda_{\min}} + \frac{N_2 + 1}{\lambda_{\min}} + \frac{|t|}{C_2 - 1} = m_2(t)
\end{eqnarray*}
where we used Lemma \ref{lem::unif_domi} to have
    the existence of $C_2 > 1$ and $N_2 \in \mathbb{N}^*$ such that $\forall n \geq N_2$ and $\forall \beta \in \Theta_2$, $ C_2 \mathbb{P}_{\beta} [Z_1 + \ldots + Z_{n+1} < z ] \leq \mathbb{P}_{\beta} [Z_1 + \ldots + Z_{n} < z ] $. 
\noindent
Now we want to control the derivative
$$ 
(t,\delta) \mapsto \frac{\partial \log ( f_{(Y, \Delta),(\lambda,\alpha,\beta)} (t, \delta)) }{ \partial \alpha_{i_1}}, \quad i \in \{1,\ldots,d_1 \},
$$
given by \eqref{eq:deriv_log_alpha} for $\delta =1$. To obtain a convenient upper bound, we need to study the following terms:
$$
    \partial_{\alpha_i} [c_{n,X}(\alpha) - c_{n+1,X}(\alpha) ] \quad \quad \text{and} \quad \quad \partial_{\alpha_i} c_{n,X}(\alpha)
$$
for $n$ sufficiently large. As explained at the end of the verification of the assumption A4, those terms are well defined due to the theorem of derivation of parametrized integrals and moreover with {\bf (H6)}
\begin{align*}
    \partial_{\alpha_i}c_{n,X}(\alpha)&   =  n \int_{\mathbb{R}^n_+} \mathbb{1}_{\{ x_1 + \ldots + x_n < x \} } \prod_{k=1}^{n-1} f_1(x_k,\alpha) \partial_{\alpha_i} f_1(x_n,\alpha) d \nu_1 (x_1) \ldots  d \nu_1 (x_n)\\
    & = n \int_{0}^x \mathbb{P}_{\alpha} [X_1 + \ldots + X_{n-1} < x-x_n ]  \partial_{\alpha_i} f_1(x_n,\alpha)  d \nu_1 (x_n)
\end{align*}
thus 
\begin{align*}
    | \partial_{\alpha_i} c_{n,X}(\alpha) |  &\leq n \mathbb{P}_{\alpha} [X_1 + \ldots + X_{n-1} < x ] \int_0^x | \partial_{\alpha_i} f(x_n,\alpha) | d \nu_1 (x_n) \\
    & \leq n c_{n-1,X}(\alpha) \int_0^x  g_{1,i} ( x_n) d \nu_1 (x_n), 
\end{align*}
where $\int_0^x  g_{1,i} ( x_n) d \nu_1 (x_n)$ is finite. By the same way:
$$
    | \partial_{\alpha_i} [c_{n,X}(\alpha) - c_{n+1,X}(\alpha) ] | \leq  \big[ n c_{n-1,X}(\alpha)  + (n+1) c_{n,X}(\alpha) \big] \int_0^x  g_{1,i} ( x_n) d \nu_1 (x_n).
$$
Pose $\bar g_{1,i} = \int_0^x  g_{1,i} ( x_n) d \nu_1 (x_n)$ and  $M_{1,\min, \alpha} = \min \{ n \in \mathbb{N}; \; \; c_{n,X}(\alpha) - c_{n+1,X}(\alpha) > 0 \}$, $\alpha \in \Theta_1$. 
Combining Lemma \ref{lem:equality_probab_term} and  \ref{lem:tech_behavior_c_n}, $M_{1,\min, \alpha} = M_{1,\min}$ does not depend on $\alpha$, and 
for any $n \in \{0; \ldots ; M_{1,\min} - 1 \}$, 
$$
c_{n,X}(\alpha) - c_{n+1,X}(\alpha)  = \partial_{\alpha_i} \left( c_{n,X}(\alpha) - c_{n+1,X}(\alpha)  \right) = 0.
$$
Then using \eqref{eq:deriv_log_alpha}, we deduce
\begin{align*}
  %  \frac{\partial \log ( f_{(Y, \Delta),(\lambda,\alpha,\beta)} (t, \delta)) }{ \partial \alpha_i} = & \frac{ \sum_{n=N_{1,\min, \alpha}}^{+ \infty} \partial_{\alpha_i} [c_{n,X}(\alpha) - c_{n+1,X}(\alpha) ] c_{n,Z}(\beta) \frac{(\lambda t)^n}{n!} }{ \sum_{n=N_{1,\min, \alpha}}^{+ \infty}  [c_{n,X}(\alpha) - c_{n+1,X}(\alpha) ] c_{n,Z}(\beta)  \frac{(\lambda t)^n}{n!}} \\
    \left| \frac{\partial \log ( f_{(Y, \Delta),(\lambda,\alpha,\beta)} (t, 1)) }{ \partial \alpha_i} \right| & \leq  \bar g_{1,i} \frac{ \sum_{n=M_{1,\min}}^{+ \infty} n [c_{n-1,X}(\alpha) + c_{n,X}(\alpha) ] c_{n,Z}(\beta) \frac{(\lambda t)^n}{n!} }{ \sum_{n=M_{1,\min}}^{+ \infty}  [c_{n,X}(\alpha) - c_{n+1,X}(\alpha) ] c_{n,Z}(\beta)  \frac{(\lambda t)^n}{n!}} \\
    & \quad + \bar g_{1,i} \frac{ \sum_{n=M_{1,\min}}^{+ \infty} c_{n,X}(\alpha) c_{n,Z}(\beta) \frac{(\lambda t)^n}{n!} }{ \sum_{n=M_{1,\min}}^{+ \infty}  [c_{n,X}(\alpha) - c_{n+1,X}(\alpha) ] c_{n,Z}(\beta) \frac{(\lambda t)^n}{n!}} \\
    \leq & \bar g_{1,i} \lambda t \frac{ \sum_{n=M_{1,\min}}^{+ \infty}  [c_{n,X}(\alpha) + c_{n+1,X}(\alpha) ] c_{n+1,Z}(\beta) \frac{(\lambda t)^n}{n!} }{ \sum_{n=M_{1,\min}}^{+ \infty}  [c_{n,X}(\alpha) - c_{n+1,X}(\alpha) ] c_{n,Z}(\beta) \frac{(\lambda t)^n}{n!}} \\
   & \quad + \bar g_{1,i} \frac{ \sum_{n=M_{1,\min}}^{+ \infty} c_{n,X}(\alpha) c_{n,Z}(\beta) \frac{(\lambda t)^n}{n!} }{ \sum_{n=M_{1,\min}}^{+ \infty}  [c_{n,X}(\alpha) - c_{n+1,X}(\alpha) ] c_{n,Z}(\beta) \frac{(\lambda t)^n}{n!}} .%  \\
 %   \leq & c_{1,i} (\lambda t) \frac{ \sum_{n=N_{1,\min, \alpha} -1}^{N_1 - 1}  [c_{n,X}(\alpha) + c_{n+1,X}(\alpha) ] c_{n+1,Z}(\beta) \frac{(\lambda t)^n}{n!} }{ [c_{N_{1,\min, \alpha},X}(\alpha) - c_{N_{1,\min, \alpha}+1,X}(\alpha) ] c_{N_{1,\min, \alpha},Z}(\beta) \frac{(\lambda t)^{N_{1,\min, \alpha}}}{N_{1,\min, \alpha}!}} \\
  %  & + c_{1,i} (\lambda t) \frac{ \sum_{n=N_{1}}^{+ \infty}  [c_{n,X}(\alpha) + c_{n+1,X}(\alpha) ] c_{n+1,Z}(\beta) \frac{(\lambda t)^n}{n!} }{ \sum_{n=N_{1}}^{+ \infty}  [c_{n,X}(\alpha) - c_{n+1,X}(\alpha) ] c_{n,Z}(\beta) \frac{(\lambda t)^n}{n!}} %\\
%    & + c_{1,i} \frac{ \sum_{n=N_{1,\min, \alpha} }^{N_1 - 1}  c_{n,X}(\alpha) c_{n,Z}(\beta) \frac{(\lambda t)^n}{n!} }{[c_{N_{1,\min, \alpha},X}(\alpha) - c_{N_{1,\min, \alpha}+1,X}(\alpha) ] c_{N_{1,\min, \alpha},Z}(\beta) \frac{(\lambda t)^{N_{1,\min, \alpha}}}{N_{1,\min, \alpha}!}} \\
 %   & + c_{1,i} \frac{ \sum_{n=N_{1}}^{+ \infty} c_{n,X}(\alpha) c_{n,Z}(\beta) \frac{(\lambda t)^n}{n!} }{ \sum_{n=N_{1}}^{+ \infty}  [c_{n,X}(\alpha) - c_{n+1,X}(\alpha) ] c_{n,Z}(\beta) \frac{(\lambda t)^n}{n!}}
\end{align*}
And with Lemma \ref{lem::unif_domi} for any $n \geq N_1$, $\alpha \in \Theta_1$,
\begin{align*}
    [c_{n,X}(\alpha) + c_{n+1,X}(\alpha) ] & \leq \frac{C_1 + 1}{C_1 - 1} [c_{n,X}(\alpha) - c_{n+1,X}(\alpha) ] \\
%    \text{and} \quad \quad \quad \quad \quad \quad 
c_{n,X}(\alpha) & \leq \frac{C_1 + 1}{C_1 - 1} [c_{n,X}(\alpha) - c_{n+1,X}(\alpha) ].
\end{align*}
Hence
\begin{align*}
& \frac{ \sum_{n=M_{1,\min}}^{+ \infty}  [c_{n,X}(\alpha) + c_{n+1,X}(\alpha) ] c_{n+1,Z}(\beta) \frac{(\lambda t)^n}{n!} }{ \sum_{n=M_{1,\min}}^{+ \infty}  [c_{n,X}(\alpha) - c_{n+1,X}(\alpha) ] c_{n,Z}(\beta) \frac{(\lambda t)^n}{n!}} \\
& \quad \leq \frac{ \sum_{n=M_{1,\min}}^{N_1}  [c_{n,X}(\alpha) + c_{n+1,X}(\alpha) ] c_{n+1,Z}(\beta) \frac{(\lambda t)^n}{n!} }{ \sum_{n=M_{1,\min}}^{+ \infty}  [c_{n,X}(\alpha) - c_{n+1,X}(\alpha) ] c_{n,Z}(\beta) \frac{(\lambda t)^n}{n!}} + \dfrac{C_1+1}{C_1-1} \\
& \quad \leq \frac{ \sum_{n=M_{1,\min}}^{N_1}  [c_{n,X}(\alpha) + c_{n+1,X}(\alpha) ] c_{n+1,Z}(\beta) \frac{(\lambda t)^n}{n!} }{  [c_{M_{1,\min},X}(\alpha) - c_{M_{1,\min}+1,X}(\alpha) ] c_{M_{1,\min},Z}(\beta) \frac{(\lambda t)^{M_{1,\min}}}{M_{1,\min}!}} + \dfrac{C_1+1}{C_1-1}
\end{align*}
From Lemma \ref{lem:null_prob_induces_non_censoring}, since we work in a censoring model, for all $\beta \in \Theta_2$, $c_{M_{1,\min},Z}(\beta) > 0$. %Of course, if it is not the case, we would have a non sense on the censoring problem. The hitting times $T$ would almost surely occur strictly after $C$. A simple verification would give the density for $\delta = 1$, null for any $t>0$.
Consequently, we obtain the following domination:
\begin{align*}
    \left| \frac{\partial \log ( f_{(Y, \Delta),(\lambda,\alpha,\beta)} (t, 1)) }{ \partial \alpha_i} \right| & \leq \bar g_{1,i} \frac{C_1 + 1}{C_1 - 1} ( 1 + \lambda_{\max} t ) \\
    & +   \frac{2\bar g_{1,i} ( 1 + \lambda_{\max} t ) }{\overset{\sim}{m}_{M_{1,\min},X} m_{M_{1,\min},Z}} \sum_{n=0}^{N_1 - M_{1,\min}} (\lambda_{\max} t)^n.
\end{align*}
When $\delta = 0$,
%$$f_{(Y, \Delta),(\lambda,\alpha,\beta)} (t, 0) =  \sum_{n=0}^{+ \infty} [c_{n,Z}(\beta) - c_{n+1,Z}(\beta) ] c_{n+1,X}(\alpha) \lambda e^{- \lambda t} \frac{(\lambda t)^n}{n!}$$
\begin{align*}
    \left|  \frac{\partial \log ( f_{(Y, \Delta),(\lambda,\alpha,\beta)} (t, 0)) }{ \partial \alpha_i} \right|
    & \leq  \bar g_{1,i} \frac{ \sum_{n=0}^{+ \infty} [c_{n,Z}(\beta) - c_{n+1,Z}(\beta) ] (n+1) c_{n,X}(\alpha) \frac{(\lambda t)^n}{n!} }{ \sum_{n=0}^{+ \infty}  [c_{n,Z}(\beta) - c_{n+1,Z}(\beta) ] c_{n+1,X}(\alpha)  \frac{(\lambda t)^n}{n!}}  \\ 
    & =  \bar g_{1,i} \frac{ (1-c_{1,Z}(\beta)) + (\lambda t) \sum_{n=0}^{+ \infty} [c_{n+1,Z}(\beta) - c_{n+2,Z}(\beta) ] \frac{n+2}{n+1} c_{n+1,X}(\alpha) \frac{(\lambda t)^n}{n!} }{ \sum_{n=0}^{+ \infty}  [c_{n,Z}(\beta) - c_{n+1,Z}(\beta) ] c_{n+1,X}(\alpha) \frac{(\lambda t)^n}{n!}} \\
    & \leq \bar g_{1,i} \frac{ (1-c_{1,Z}(\beta)) + 2\lambda t \sum_{n=0}^{+ \infty} [c_{n+1,Z}(\beta) - c_{n+2,Z}(\beta) ]  c_{n+1,X}(\alpha) \frac{(\lambda t)^n}{n!} }{ \sum_{n=0}^{+ \infty}  [c_{n,Z}(\beta) - c_{n+1,Z}(\beta) ] c_{n+1,X}(\alpha) \frac{(\lambda t)^n}{n!}}.
\end{align*}
If $1-c_{1,Z}(\beta) > 0$, then 
$$
\frac{ (1-c_{1,Z}(\beta))  }{ \sum_{n=0}^{+ \infty}  [c_{n,Z}(\beta) - c_{n+1,Z}(\beta) ] c_{n+1,X}(\alpha) \frac{(\lambda t)^n}{n!}} \leq \dfrac{1}{c_{1,X}(\alpha)} \leq \frac{1}{m_{1,X}}. 
$$
Indeed if $c_{1,X}(\alpha)=0$, $X_1 \geq x$ a.s. which implies that $T \leq C$ a.s., which is excluded. Hence if $M_{2,\min}=M_{2,\min, \beta} = \min \{ n \in \mathbb{N}; \; \; c_{n,Z}(\beta) - c_{n+1,Z}(\beta) > 0 \}$, we obtain
\begin{align*}    
    \left|  \frac{\partial \log ( f_{(Y, \Delta),(\lambda,\alpha,\beta)} (t, 0)) }{ \partial \alpha_i} \right| \leq  \bar g_{1,i} \frac{1}{m_{1,X}} + 2 \bar g_{1,i} \lambda t \frac{ \sum_{n= N_{2,\min} -1}^{+ \infty} [c_{n+1,Z}(\beta) - c_{n+2,Z}(\beta) ] c_{n+1,X}(\alpha) \frac{(\lambda t)^n}{n!} }{ \sum_{n=M_{2,\min}}^{+ \infty}  [c_{n,Z}(\beta) - c_{n+1,Z}(\beta) ] c_{n+1,X}(\alpha) \frac{(\lambda t)^n}{n!}}. 
\end{align*}
As in the case $\delta = 1$, we separate the sums in two parts $n\leq N_2$ and $n \geq N_2$, where $N_2$ is given by Lemma \ref{lem::unif_domi}.  For $n \geq N_2$, we have a bound equal to $2 \bar g_{1,i} \lambda t \dfrac{1}{C_2-1}$. For $M_{2,\min}\leq n \leq N_2$, we obtain a polynomial bound as in the case $\delta =1$. Thus we obtain a polynomial bound
%\begin{align*} 
%    \\
 %   &\quad  \leq  \bar g_{1,i} \frac{1}{m_{1,X}} + 2 \bar g_{1,i} \lambda t \frac{ \sum_{n= M_{2,\min} -1}^{N_2 -1} [c_{n+1,Z}(\beta) - c_{n+2,Z}(\beta) ]  c_{n+1,X}(\alpha) \frac{(\lambda t)^n}{n!} }{ [c_{N_{2,\min, \beta},Z}(\beta) - c_{N_{2,\min, \beta}+1,Z}(\beta) ] c_{N_{2,\min, \beta}+1,X}(\alpha) \frac{(\lambda t)^{N_{2,\min, \beta}}}{N_{2,\min, \beta}!}} \\
 %   & + 2 c_{1,i} (\lambda t) \frac{ \sum_{n= N_{2} }^{+ \infty} [c_{n+1,Z}(\beta) - c_{n+2,Z}(\beta) ]  c_{n+1,X}(\alpha) \frac{(\lambda t)^n}{n!} }{ \sum_{n=N_{2}}^{+ \infty}  [c_{n,Z}(\beta) - c_{n+1,Z}(\beta) ] c_{n+1,X}(\alpha) \frac{(\lambda t)^n}{n!}}
%\end{align*}
%The second inequality is deduced from the fact that if $m_{n,X} = \min \{ \mathbb{P}_{\alpha} [ X_1 < x ], \quad \alpha \in \Theta_1  \} = 0$,then $\forall \alpha \in \Theta_1, \mathbb{P}_{\alpha} [ X_1 < x ] = 0$ (deduced from the previous lemma) and the model would have a non sense, for the same reason as in the previous domination, we would have that $T$ occurs almost surely before $C$. Same remark can be done on $\mathbb{P}_{\alpha} [X_1 + \ldots + X_{N_{2,\min, \beta}+1} < x ]$. \\
%Consequently, we obtain the following domination:
\begin{align*}
    \left|  \frac{\partial \log ( f_{(Y, \Delta),(\lambda,\alpha,\beta)} (t, 0)) }{ \partial \alpha_i} \right| \leq \bar g_{1,i} \frac{1}{m_{1,X}} + 2 \bar g_{1,i} \frac{1}{\tilde m_{M_{2,\min},Z} m_{M_{2,\min},X}}  \sum_{n=0}^{N_2 - M_{2,\min}} (\lambda_{\max} t)^n + \frac{2 \bar g_{1,i} (\lambda_{\max} t)}{C_2 -1}.
\end{align*}
Hence in both cases $\delta=1$ or $\delta =0$, we have a polynomial bound on the derivative. And any polynomial function is integrable with respect to $F_{(Y,\Delta),(\lambda^0,\alpha^0,\beta^0)}$.
The exact same reflections allow us to obtain the same type of inequalities for all of the remaining derivatives of first order 
$$ (t,\delta) \mapsto \frac{\partial \log ( f_{(Y, \Delta),(\lambda,\alpha,\beta)} (t, \delta)) }{ \partial \beta_{j_1}} $$
with $j_1 \in \{1,\ldots,d_2\}$
and of second order: 
\begin{align*}
     \frac{\partial^2 \log ( f_{(Y, \Delta),(\lambda,\alpha,\beta)} (t, \delta)) }{ \partial \lambda^2} \quad & ; \quad  \frac{\partial^2 \log ( f_{(Y, \Delta),(\lambda,\alpha,\beta)} (t, \delta)) }{ \partial \alpha_{i_1} \partial \lambda} \quad ; \quad  \frac{\partial^2 \log ( f_{(Y, \Delta),(\lambda,\alpha,\beta)} (t, \delta)) }{ \partial \beta_{j_1} \partial \lambda} \\
    \frac{\partial^2 \log ( f_{(Y, \Delta),(\lambda,\alpha,\beta)} (t, \delta)) }{ \partial \alpha_{i_1} \partial \alpha_{i_2}} \quad & ; \quad  \frac{\partial^2 \log ( f_{(Y, \Delta),(\lambda,\alpha,\beta)} (t, \delta)) }{ \partial \beta_{j_1} \partial \alpha_{i_2}} \quad ; \quad  \frac{\partial^2 \log ( f_{(Y, \Delta),(\lambda,\alpha,\beta)} (t, \delta)) }{ \partial \beta_{j_1} \partial \beta_{j_2}}
\end{align*}
for any $i_1,i_2 \in \{1,\ldots,d_1\}$, $j_1,j_2 \in \{1,\ldots,d_2\}$. 
%we will treat separately all the cases. All are not be treated since the ideas will be similar on the remaining functions. 
There are polynomial functions that respectively dominate all these derivatives, uniformly on the family of parameters $(\lambda,\alpha,\beta) \in \Lambda \times \Theta_1 \times \Theta_2$, which proves that A5 holds.
 
\smallskip
\noindent \underline{Assumption A6.} We already assume that $(\lambda^0,\alpha^0,\beta^0) \in \overset{\circ}{\Theta}$ in {\bf(H3)}.
Evoke that the matrices $A(\theta)$ and $B(\theta)$ are defined by Equations \eqref{eq:matrix_A_theta} and \eqref{eq:matrix_B_theta}. 
%In \cite{Whi82}, the author defines a regular point of the matrix $A(\theta)$ as a value for $\theta$ such that $A(\theta)$ has a constant rank in some open neighborhood of $\theta$. \textcolor{red}{Assume that $\theta^0$ is a regular point}. Then under our framework and from \cite[Theorem 3.1]{Whi82}, $A(\theta^0)$ is negative definite. 
Furthermore since the model is correctly specified, $A(\theta)=-B(\theta)$ ; we prove this property below (also see remark after \cite[Theorem 3.2]{Whi82}). We also assume that $A(\theta^0) = -B(\theta^0)$ is non-singular. Therefore Assumption A6 of \cite{Whi82} holds. 

%b) To prove that $B(\lambda^0,\alpha^0,\beta^0)$ is non-singular, we can use the theorem 3.1 from \cite{Whi82} which states that if the model is globally identifiable, and the fact that under our assumptions A1)-A3)a) and A4)-A6)a) are verified and $(\lambda^0,\alpha^0,\beta^0)$ is regular point of $(\lambda,\alpha,\beta) \mapsto A (\lambda,\alpha,\beta) $, then $A(\lambda^0,\alpha^0,\beta^0)$ is negative definite. \\

\smallskip
\noindent
The conclusion   of Theorem \ref{theo:assymptotic_normaility} follows now from  \cite[Theorem 3.2]{Whi82}. 
\qed

\medskip

\begin{lemma} \label{lem:A_equal_-B}
    For all $\theta \in \overset{\circ}{\Theta}$, $A(\theta)=-B(\theta)$. 
\end{lemma}
\begin{proof}
We only prove the result considering $(Y,\Delta)$ from Model I, since the proofs with Model II are the same.
To prove it, we can compute the terms of both matrices and we show that 
for all $i,j \in \{ 1, \ldots , 1+ d_1 + d_2 \}$
$$\mathbb{E}[ \partial^2 \log(f_{(Y,\Delta),\theta}/ \partial \theta_i \partial \theta_j ] = - \mathbb{E}[ \left( \partial \log(f_{(Y,\Delta),\theta}/ \partial \theta_j \right) \left( \partial \log(f_{(Y,\Delta),\theta}/ \partial \theta_i \right) ].$$ 
Indeed, if $(\theta_i,\theta_j) = (\lambda,\lambda)$:
\begin{align*}
    \mathbb{E} \left[ \frac{\partial^2 \log(f_{(Y,\Delta),(\lambda,\alpha,\beta)}) }{\partial \lambda^2} \right]  & =  \int_{0}^{+ \infty} \frac{\partial^2 \log(f_{(\lambda,\alpha,\beta)}(y,1)) }{\partial \lambda^2} f_{(\lambda,\alpha,\beta)}(y,1) dy\\
    & \quad + \int_{0}^{+ \infty} \frac{\partial^2 \log(f_{(\lambda,\alpha,\beta)}(y,0)) }{\partial \lambda^2} f_{(\lambda,\alpha,\beta)}(y,0) dy \\
    = &  \int_{0}^{+ \infty} \frac{\partial^2 f_{(\lambda,\alpha,\beta)}(y,1) }{\partial \lambda^2} dy - \int_{0}^{+ \infty} \frac{( \partial_{\lambda} f_{(\lambda,\alpha,\beta)}(y,1))^2 }{ f_{(\lambda,\alpha,\beta)}(y,1) } dy \\
    & \quad + \int_{0}^{+ \infty} \frac{\partial^2 f_{(\lambda,\alpha,\beta)}(y,0) }{\partial \lambda^2} dy - \int_{0}^{+ \infty} \frac{( \partial_{\lambda} f_{(\lambda,\alpha,\beta)}(y,0))^2 }{ f_{(\lambda,\alpha,\beta)}(y,0) } dy 
\end{align*}
and
\begin{align*}
    \mathbb{E} \left[ \Bigg( \frac{ \partial \log(f_{(\lambda,\alpha,\beta)}(Y,\Delta)) }{\partial \lambda} \Bigg)^2 \right] = \int_{0}^{+ \infty} \frac{( \partial_{\lambda} f_{(\lambda,\alpha,\beta)}(y,1))^2 }{ f_{(\lambda,\alpha,\beta)}(y,1) } dy + \int_{0}^{+ \infty} \frac{( \partial_{\lambda} f_{(\lambda,\alpha,\beta)}(y,0))^2 }{ f_{(\lambda,\alpha,\beta)}(y,0) } dy. 
\end{align*}
Now 
\begin{align*}
   & \int_{0}^{+ \infty} \frac{\partial^2 f_{(\lambda,\alpha,\beta)}(y,1) }{\partial \lambda^2} dy =  \int_{0}^{+ \infty} \partial_{\lambda}^2 \sum_{n \geq 0} (c_{n,X}(\alpha) - c_{n+1,X}(\alpha)) c_{n,Z}(\beta) e^{- \lambda y} \lambda \frac{(\lambda y)^n}{n!}  dy \\
     &  \quad =  \int_{0}^{+ \infty} \sum_{n \geq 0} (c_{n,X}(\alpha) - c_{n+1,X}(\alpha)) c_{n,Z}(\beta) e^{- \lambda y} \left[ (n+1)n - 2(n+1) y \lambda + y^2 \lambda^2 \right] \lambda^{n-1} \frac{y^n}{n!}  dy \\
     &\quad = \sum_{n \geq 0} (c_{n,X}(\alpha) - c_{n+1,X}(\alpha)) c_{n,Z}(\beta) \left[ n - 2(n+1) + (n+2) \right]\dfrac{n+1}{\lambda^2} = 0,
\end{align*}
where we used that for all $n\in \mathbb N$
$$\int_{0}^{+ \infty} e^{- \lambda y} \lambda^{n+1} \frac{(\lambda y)^n}{n!}  dy = 1$$
and the inversion of integral and series, which is valid thanks to the convergence theorem \ref{theo:behaviour_c_(n+1)/c_n}. % series function theo:behaviour_c_(n+1)/c_ntheo:behaviour_c_(n+1)/c_non any compact subset of $\mathbb{R}+$ and the existence of an integrable function that uniformly bound the derivated series function on $\Theta$. 
Since the structure of $ y \mapsto f_{(\lambda,\alpha,\beta)}(y,1)$ and $ y \mapsto f_{(\lambda,\alpha,\beta)}(y,0)$ is the same, we also obtain:
\begin{align*}
    \int_{0}^{+ \infty} \frac{\partial^2 f_{(\lambda,\alpha,\beta)}(y,0) }{\partial \lambda^2} dy = 0.
\end{align*}
Similar arguments show that 
\begin{align*}
    \int_{0}^{+ \infty} \frac{\partial^2 f_{(\lambda,\alpha,\beta)}(y,1) }{\partial \lambda \partial \theta_j} dy = \int_{0}^{+ \infty} \frac{\partial^2 f_{(\lambda,\alpha,\beta)}(y,0) }{\partial \lambda \partial \theta_j} dy = 0
\end{align*}
if $\theta_j = \alpha_k$, $1 \leq k \leq d_1 $ or $\theta_j = \beta_k$, $1 \leq k \leq d_2 $. 
Indeed,
\begin{align*}
    & \int_{0}^{+ \infty} \frac{\partial^2 f_{(\lambda,\alpha,\beta)}(y,1) }{\partial \lambda \partial \alpha_k} dy =  \int_{0}^{+ \infty} \partial_{\lambda} \partial_{\alpha_k} \sum_{n \geq 0} (c_{n,X}(\alpha) - c_{n+1,X}(\alpha)) c_{n,Z}(\beta) e^{- \lambda y} \lambda \frac{(\lambda y)^n}{n!}  dy \\ 
     &\quad = \int_{0}^{+ \infty}  \sum_{n \geq 0} \partial_{\alpha_k} (c_{n,X}(\alpha) - c_{n+1,X}(\alpha)) c_{n,Z}(\beta)  \left[ (n+1) - y \lambda \right] e^{- \lambda y} \lambda^n \frac{ y^n}{n!}  dy \\
    & \quad =  \sum_{n \geq 0} \partial_{\alpha_k}(c_{n,X}(\alpha) - c_{n+1,X}(\alpha)) c_{n,Z}(\beta) \left[ (n+1) - (n+1) \right] \dfrac{1}{\lambda} = 0.
\end{align*}
%The same way,
%\begin{align*}
 %   \int_{0}^{+ \infty} \frac{\partial^2 f_{(\lambda,\alpha,\beta)}(y,0) }{\partial \lambda \partial \alpha_k} dy = 0
%\end{align*}
%\underline{If $(\theta_i,\theta_j) = (\lambda,\beta_k)$:} (for some $k \in \{ 1 , \ldots , d_1 \}$), it is treated the same way as the previous one. \\
Finally if $(\theta_i,\theta_j) = (\alpha_k,\beta_l)$ for some $k \in \{ 1 , \ldots , d_1 \}$ , $l \in \{ 1 , \ldots , d_2 \}$, we have:
\begin{align*} \nonumber
    \int_{0}^{+ \infty} \frac{\partial^2 f_{(\lambda,\alpha,\beta)}(y,1) }{\partial \alpha_k \partial \beta_l} dy = & \int_{0}^{+ \infty}  \partial_{\alpha_k} \partial_{\beta_l} \sum_{n \geq 0} (c_{n,X}(\alpha) - c_{n+1,X}(\alpha)) c_{n,Z}(\beta) e^{- \lambda y} \lambda \frac{(\lambda y)^n}{n!}  dy \\  \nonumber
    = & \partial_{\alpha_k} \partial_{\beta_l} \sum_{n \geq 0} (c_{n,X}(\alpha) - c_{n+1,X}(\alpha)) c_{n,Z}(\beta) \int_{0}^{+ \infty} e^{- \lambda y} \lambda \frac{(\lambda y)^n}{n!}  dy \\ 
    = & \partial_{\alpha_k} \partial_{\beta_l} \sum_{n \geq 0} (c_{n,X}(\alpha) - c_{n+1,X}(\alpha)) c_{n,Z}(\beta)
\end{align*}
and
\begin{align*} 
    \int_{0}^{+ \infty} \frac{\partial^2 f_{(\lambda,\alpha,\beta)}(y,0) }{\partial \alpha_k \partial \beta_l} dy = & \int_{0}^{+ \infty}  \partial_{\alpha_k} \partial_{\beta_l} \sum_{n \geq 0} (c_{n,Z}(\beta) - c_{n+1,Z}(\beta)) c_{n+1,X}(\alpha) e^{- \lambda y} \lambda \frac{(\lambda y)^n}{n!}  dy \\    
    = & \partial_{\alpha_k} \partial_{\beta_l} \sum_{n \geq 0} (c_{n,Z}(\beta) - c_{n+1,Z}(\beta)) c_{n+1,X}(\alpha) \int_{0}^{+ \infty} e^{- \lambda y} \lambda \frac{(\lambda y)^n}{n!}  dy \\ 
    = & \partial_{\alpha_k} \partial_{\beta_l} \sum_{n \geq 0} (c_{n,Z}(\beta) - c_{n+1,Z}(\beta)) c_{n+1,X}(\alpha). 
\end{align*}
Thus,
\begin{align*}
    & \int_{0}^{+ \infty} \frac{\partial^2 f_{(\lambda,\alpha,\beta)}(y,1) }{\partial \alpha_k \partial \beta_l} dy + \int_{0}^{+ \infty} \frac{\partial^2 f_{(\lambda,\alpha,\beta)}(y,0) }{\partial \alpha_k \partial \beta_l} dy \\
    &\quad  =  \partial_{\alpha_k} \partial_{\beta_l} \sum_{n \geq 0} (c_{n,X}(\alpha) - c_{n+1,X}(\alpha)) c_{n,Z}(\beta) + (c_{n,Z}(\beta) - c_{n+1,Z}(\beta)) c_{n+1,X}(\alpha) \\
    & \quad =  \partial_{\alpha_k} \partial_{\beta_l} \sum_{n \geq 0} c_{n,X}(\alpha)c_{n,Z}(\beta)  - c_{n+1,Z}(\beta) c_{n+1,X}(\alpha) \\
    & \quad =  \partial_{\alpha_k} \partial_{\beta_l}  c_{0,X}(\alpha)c_{0,Z}(\beta) = 0,
\end{align*}
%
%\begin{align*}
 %   \int_{0}^{+ \infty} \frac{\partial^2 f_{(\lambda,\alpha,\beta)}(y,1) }{\partial \alpha_k \partial \beta_l} dy + \int_{0}^{+ \infty} \frac{\partial^2 f_{(\lambda,\alpha,\beta)}(y,0) }{\partial \alpha_k \partial \beta_l} dy = 0
%\end{align*}
since $ \alpha \mapsto c_{0,X}(\alpha)$ and $ \beta \mapsto c_{0,Z}(\beta)$ are constant functions. %, their respective derivatives are null functions. \\
We use the same arguments if $(\theta_i,\theta_j) = (\alpha_k,\alpha_l)$ or if $(\theta_i,\theta_j) = (\beta_k,\beta_l)$, which achieves the proof. 
\end{proof}

\end{document}